%% file: manuscript.tex
\documentclass[preprint]{elsarticle}

\usepackage{amsfonts}
\usepackage{amsmath}
\usepackage{amssymb}
\usepackage{amsthm}
\usepackage{booktabs}
\usepackage{mathrsfs}
\usepackage{cite}
\usepackage{times}
\usepackage{url}
\usepackage{hyperref}
\usepackage{fonts}
\usepackage{lineno}
\usepackage{yhmath}
\hypersetup{
  bookmarksnumbered = true,
  bookmarksopen=false,
  pdfborder=0 0 0,         
  pdffitwindow=true,      
  pdfnewwindow=true, 
  colorlinks=true,           
  linkcolor=blue,            
  citecolor=magenta,    
  filecolor=magenta,     
  urlcolor=cyan              
}

\newtheorem{lemma}{Lemma}
\newtheorem{theorem}{Theorem}

\newtheorem{rem}{Remark}
\newtheorem{cor}{Corollary}

\input{macros}

\begin{document}
\begin{frontmatter}
\title{Bound-Preserving Discontinuous Galerkin Methods for Conservative Phase Space Advection in Curvilinear Coordinates \tnoteref{ldrd-support}}
\date{\today}
\author[ornl,utk-phys]{Eirik Endeve\corref{cor}}
\ead{endevee@ornl.gov}

\author[ornl,utk-math]{Cory D. Hauck}
\ead{hauckc@ornl.gov}

\author[ornl,utk-math]{Yulong Xing}
\ead{xingy@math.utk.edu}

\author[utk-phys]{Anthony Mezzacappa}
\ead{mezz@utk.edu}

\cortext[cor]{Corresponding author. Tel.:+1 865 576 6349; fax:+1 865 241 0381}

\tnotetext[ldrd-support]{This research is sponsored, in part, by the Laboratory Directed Research and Development Program of Oak Ridge National Laboratory (ORNL), managed by UT-Battelle, LLC for the U. S. Department of Energy under Contract No. De-AC05-00OR22725.  It used resources of the Oak Ridge Leadership Computing Facility at ORNL provided through the INCITE program and a Director's Discretionary allocation.  The research of the second author is supported in part by NSF under Grant No. 1217170. The research of the third author is supported in part by NSF grant DMS-1216454.}

\address[ornl]{Computational and Applied Mathematics Group,
				Oak Ridge National Laboratory,
				Oak Ridge, TN 37831 USA }

\address[utk-math]{Department of Mathematics,
				University of Tennessee
				Knoxville, TN 37996-1320}
				
\address[utk-phys]{Department of Physics and Astronomy,
				University of Tennessee
				Knoxville, TN 37996-1200}

\input{abstract}


\end{frontmatter}

\input{intro}
\input{max}
\input{limiter}
\input{numerics}
\input{conclusion}
\input{appendix}
\clearpage

\bibliographystyle{plain}
\bibliography{./refs/dgCurv.bib}

\end{document}

%% file: macros.tex
\newcommand{\betau}[1]{\beta^{#1}}
\newcommand{\covderiv}[1]{\nabla_{#1}}
\newcommand{\collision}[1]{\mathbb{C}\left[#1\right]}
\newcommand{\cudd}[3]{\Gamma^{#1}_{~{#2}{#3}}}
\newcommand{\f}[2]{\frac{#1}{#2}}
\newcommand{\fourvelocityEd}[1]{n_{#1}}
\newcommand{\fourvelocityEu}[1]{n^{#1}}
\newcommand{\fDG}{f_{\mbox{\tiny \rm{DG}}}}
\newcommand{\gDG}{g_{\mbox{\tiny \rm{DG}}}}

\newcommand{\fhDG}[1]{\hat{f}_{\mbox{\tiny DG},{#1}}}
\newcommand{\gdd}[2]{g_{#1#2}}
\newcommand{\gmud}[2]{\gamma^{#1}_{~#2}}
\newcommand{\gmdd}[2]{\gamma_{#1#2}}
\newcommand{\jacPubdtld}[2]{P^{\bar{#1}}_{~\tilde{#2}}}
\newcommand{\jacPutlddb}[2]{P^{\tilde{#1}}_{~\bar{#2}}}
\newcommand{\lambdadtlddtld}[2]{\lambda_{\tilde{#1}\tilde{#2}}}
\newcommand{\lambdautldutld}[2]{\lambda^{\tilde{#1}\tilde{#2}}}
\newcommand{\ld}[1]{l_{#1}}
\newcommand{\lu}[1]{l^{#1}}
\newcommand{\ldb}[1]{l_{\bar{#1}}}
\newcommand{\lub}[1]{l^{\bar{#1}}}
\newcommand{\mdet}{\sqrt{-g}}
\newcommand{\pderiv}[2]{\frac{\partial #1}{\partial #2}}
\newcommand{\pmdet}{\sqrt{\lambda}}
\newcommand{\pd}[1]{p_{#1}}
\newcommand{\pu}[1]{p^{#1}}
\newcommand{\pdb}[1]{p_{\bar{#1}}}
\newcommand{\pub}[1]{p^{\bar{#1}}}
\newcommand{\putld}[1]{p^{\tilde{#1}}}
\newcommand{\rPerp}{R}
\newcommand{\smdet}{\sqrt{\gamma}}
\newcommand{\tauderiv}[1]{\f{\partial #1}{\partial\tau}}
\newcommand{\tauderivc}[1]{\f{d#1}{d\tau}}
\newcommand{\tetubd}[2]{e^{\bar{#1}}_{~#2}}
\newcommand{\tetudb}[2]{e^{#1}_{~\bar{#2}}}
\newcommand{\tetradEpsilon}{E}
\newcommand{\tetradPhi}{\Phi}
\newcommand{\tetradTheta}{\Theta}
\newcommand{\vect}[1]{\boldsymbol{#1}}
\newcommand{\xu}[1]{x^{#1}}
\newcommand{\xub}[1]{x^{\bar{#1}}}

\newcommand{\sumx}{\sum_{i=1}^{d_x}}
\newcommand{\sump}{\sum_{i=1}^{d_p}}
\newcommand{\sumz}{\sum_{i=1}^{d_z}}
\newcommand{\prodz}{\prod_{i=1}^{d_z}}
\newcommand{\tV}{{\tilde{V}}}
\newcommand{\tK}{{\tilde{\bK}}}
\newcommand{\tQ}{{\tilde{\bQ}}}
\newcommand{\tS}{{\tilde{\bS}}}

\newcommand{\tz}{{\tilde{\vect{z}}}}
\newcommand{\zL}{z_{\mbox{\tiny \rm L}}}
\newcommand{\zH}{z_{\mbox{\tiny \rm H}}}
\newcommand{\rL}{r_{\mbox{\tiny \rm L}}}
\newcommand{\rH}{r_{\mbox{\tiny \rm H}}}
\newcommand{\RL}{\rPerp_{\mbox{\tiny \rm L}}}
\newcommand{\RH}{\rPerp_{\mbox{\tiny \rm H}}}
\newcommand{\muL}{\mu_{\mbox{\tiny \rm L}}}
\newcommand{\muH}{\mu_{\mbox{\tiny \rm H}}}
\newcommand{\PhiL}{\tetradPhi_{\mbox{\tiny \rm L}}}
\newcommand{\PhiH}{\tetradPhi_{\mbox{\tiny \rm H}}}
\newcommand{\eL}{\tetradEpsilon_{\mbox{\tiny \rm L}}}
\newcommand{\eH}{\tetradEpsilon_{\mbox{\tiny \rm H}}}

\newcommand{\p}{\partial}
\newcommand{\dt}{\Delta t}
\newcommand{\dz}{\Delta z}

\newcommand{\quand}{\quad \mbox{and} \quad}


%% file: abstract.tex
\begin{abstract}
We extend the positivity-preserving method of Zhang \& Shu \citep{ZS2010a} to simulate the advection of neutral particles in phase space using curvilinear coordinates.  
The ability to utilize these coordinates is important for non-equilibrium transport problems in general relativity and also in science and engineering applications with specific geometries.  
The method achieves high-order accuracy using Discontinuous Galerkin (DG) discretization of phase space and strong stability-preserving, Runge-Kutta (SSP-RK) time integration.  
Special care in taken to ensure that the method preserves strict bounds for the phase space distribution function $f$; i.e., $f\in[0,1]$.  
The combination of suitable CFL conditions and the use of the high-order limiter proposed in \citep{ZS2010a} is sufficient to ensure positivity of the distribution function.  
However, to ensure that the distribution function satisfies the upper bound, the discretization must, in addition, preserve the divergence-free property of the phase space flow.  
Proofs that highlight the necessary conditions are presented for general curvilinear coordinates, and the details of these conditions are worked out for some commonly used coordinate systems (i.e., spherical polar spatial coordinates in spherical symmetry and cylindrical spatial coordinates in axial symmetry, both with spherical momentum coordinates).  
Results from numerical experiments --- including one example in spherical symmetry adopting the Schwarzschild metric --- demonstrate that the method achieves high-order accuracy and that the distribution function satisfies the maximum principle.  
\end{abstract}

%% file: intro.tex
\section{Introduction}
\label{sec:introduction}

In this paper, we design discontinuous Galerkin methods for the solution of the collisionless, \emph{conservative} Boltzmann equation
in general curvilinear coordinates
\begin{equation}
  \pderiv{f}{t}
  +\f{1}{\smdet}\sumx\pderiv{}{x^i}\big(\,\smdet\,F^{i}f\,\big) 
  +\f{1}{\pmdet}\sump\pderiv{}{p^{i}}\big(\,\pmdet\,G^{i}f\,\big)=0 
  \label{eq:ConservativeBoltzmannEquationCurvilinearIntro}
\end{equation}
that preserve, in the sense of local cell averages, the physical bounds on the distribution function $f=f(\vect{x},\vect{p},t)$.  
This function gives the density of particles with respect to the phase space measure $d\vect{x}\,d\vect{p}$.  
In Equation \eqref{eq:ConservativeBoltzmannEquationCurvilinearIntro}, $t\in\mathbb{R}^{+}$ represents time, and $x^{i}$ and $p^{i}$ are components of the position vector $\vect{x}\in\mathbb{R}^{d_x}$ and momentum vector $\vect{p}\in\mathbb{R}^{d_p}$, respectively.  
In general, $d_x=d_p=3$, but when imposing symmetries for simplified geometries, some dimensions may not need to be considered.
$F^{i}$ and $G^{i}$ are coefficients of the position space flux vector $\vect{F}f$ and the momentum space flux vector $\vect{G}f$, respectively, while ${\smdet}\ge0$ and ${\pmdet}\ge0$ are the determinants of the position space and momentum space metric tensors, respectively.
(See \ref{sec:equations} for more details. In particular, Equation \eqref{eq:ConservativeBoltzmannEquationCurvilinearIntro} is obtained from the conservative,  general relativistic Boltzmann equation in the limit of a \emph{time-independent} spacetime.)  
Equation \eqref{eq:ConservativeBoltzmannEquationCurvilinearIntro} must be supplemented with appropriate boundary and initial conditions which, at this point, are left unspecified.  

The upper and lower bounds on $f$ follow from the \emph{non-conservative} advection equation 
\begin{equation}\label{eq:AdvectiveBoltzmannEquationCurvilinearIntro}
 \pderiv{f}{t}+\sumx F^i\pderiv{f}{x^i}+ \sump G^i\pderiv{f}{p^i}=0, 
\end{equation}
which is formally equivalent to \eqref{eq:ConservativeBoltzmannEquationCurvilinearIntro} due to the divergence-free property of the phase space, or ``Liouville," flow
\begin{equation}
  \f{1}{\smdet}\sumx\pderiv{}{x^i}\big(\,\smdet\,F^{i}\,\big) 
  +\f{1}{\pmdet}\sump\pderiv{}{p^{i}}\big(\,\pmdet\,G^{i}\,\big)=0.
  \label{eq:divergenceFreePhaseSpaceFlowCurvilinearAbstract}
\end{equation}
Indeed, it is straightforward to show that \eqref{eq:AdvectiveBoltzmannEquationCurvilinearIntro} preserves the bounds of the initial and boundary data.  
(Here we assume that the distribution function $f(\vect{x},t,\vect{p})$ satisfies $f\in[0,1]~\forall t$.)  
We employ the conservative form for two major reasons: (1) it is mathematically convenient when discontinuities are present and (2) it leads naturally to numerical methods with conservative properties.  
The drawback is that preserving point-wise bounds on $f$ becomes non-trivial.  

Discontinuous Galerkin (DG) methods \citep[see e.g.,][and references therein]{cockburnShu_2001,cockburn_2003,hesthavenWarburton_2008} for phase space discretization are attractive for several reasons.  
First, they achieve high-order accuracy on a compact, local stencil so that data is only communicated with nearest neighbors, regardless of the formal order of accuracy.  
This leads to a high computation to communication ratio, and favorable parallel scalability on heterogeneous architectures \citep{klockner_etal_2009}.  
Second, they exhibit favorable properties when collisions are added to the right-hand side of \eqref{eq:ConservativeBoltzmannEquationCurvilinearIntro}.  
In particular, they recover the correct asymptotic behavior in the diffusion limit \citep{LM1989,A2001,guermondKanschat_2010}, which is characterized by frequent collisions with a material background and long time scales.  
To leverage these properties, it is important to preserve positivity in the phase space advection step since negative distribution functions are physically meaningless.  
In the case of fermions, $f$ is also bounded above (i.e., $f\le1$), which introduces Pauli blocking factors in the collision operator.  
Violation of these bounds can result in numerical difficulties due to nonlinearities that can come from material coupling \citep{MH2010}.  
Simply introducing a cutoff in the algorithm is unacceptable, since this results in loss of conservation --- a critical check on physical consistency.  

In this paper, we extend the approach introduced in \citep{ZS2010a} in order to preserve upper and lower bounds of scalar conservation laws.  
The approach has three basic ingredients.  
First, one expresses the update of the (approximate) cell average in a forward Euler step as a linear combination of conservative updates.  
This requires a quadrature representation of the current local polynomial approximation that calculates the cell average exactly.  
Second, a limiter is introduced which modifies the current polynomial approximation, making point-wise values satisfy the prescribed bound on the quadrature set while maintaining the cell average.  
These two steps ensure that the Euler update of the cell average satisfies the required bounds.  
The third and final step is to apply a Strong Stability-Preserving Runge-Kutta (SSP-RK) method \citep[e.g.,][]{gottlieb_etal_2001} which can be expressed as a convex combination of Euler steps and therefore preserves the same bounds as the Euler step.

The method from \citep{ZS2010a} has been extended and applied in many ways.  
Positivity-preserving DG and weighted essentially non-osciilatory (WENO) methods have been designed for convection-diffusion equations \citep{ZZS2013,Y2014}, the Euler equations with source terms \citep{ZS2011b}, the shallow water equations \citep{XZS2010}, multi-material flows \citep{CS2014}, the ideal MHD equations \citep{CLQX2013}, moment models for radiation transport \citep{OHF2012}, and PDEs involving global integral terms including a hierarchical size-structured population model \citep{ZZS2011}.  
The specific problem of maintaining a positive distribution function in phase space has been considered in \citep{CGP2012,QS2011,RS2011} for the case of Cartesian coordinates.  
In \citep{CGP2012}, the authors consider an Eulerian scheme for the Boltzmann-Poisson system with a linear collision operator.  
In \citep{QS2011,RS2011}, semi-Lagrangian schemes are used to approximate the Vlasov-Poisson system, which contains no collisions.  
In the current work we also ignore the effects of the collision operator, and consider the conservative phase space advection equation in \eqref{eq:ConservativeBoltzmannEquationCurvilinearIntro}.  
We enforce both the upper and lower bounds on $f$ for general curvilinear coordinates.  
This introduces some nontrivial differences.  In particular, 
\begin{enumerate}
\item The volume element in each computational phase space cell depends on the coordinates.  
This means that mass matrices can vary from cell to cell.  
It also complicates the quadrature needed for exact evaluation of the cell average.  
Finally, the balance between cell averages and fluxes that gives the proper bounds requires special treatment.  
These last two properties may lead to a reduced CFL condition.
\item The divergence-free property \eqref{eq:divergenceFreePhaseSpaceFlowCurvilinearAbstract} relies on a delicate balance between position space and momentum space divergences \citep[see e.g.,][]{cardall_etal_2013}.  
In the Cartesian case, each of these terms is individually zero, so that the balance between is not important. 
\end{enumerate}

This paper is organized as follows: in Section~\ref{sec:maximumPrinciple} we develop a high-order, bounded-preserving DG method for solving the conservative phase space advection equation given by \eqref{eq:ConservativeBoltzmannEquationCurvilinearIntro}.  
Details of the method are worked out for some commonly used phase space coordinates (i.e., spherical polar spatial coordinates in spherical symmetry --- including a general relativistic example adopting the Schwarzschild metric --- and cylindrical spatial coordinates in axial symmetry, both with spherical momentum coordinates).  
The limiter proposed in \citep{ZS2010a}, which ensures that point-wise values of $f$ satisfy the maximum principle, is briefly summarized in Section~\ref{sec:limiter}.  
Numerical results demonstrating that our high-order DG method satisfies the maximum principle for the specific cases considered in Section~\ref{sec:maximumPrinciple} are presented in Section~\ref{sec:numericalExamples}.  
We also evaluate the efficiency of the high-order DG methods for the conservative phase-space advection problem.
Summery and conclusions are given in Section~\ref{sec:conclusions}.  
This work is motivated by our objective to develop robust, high-order methods to simulate neutrino transport in core-collapse supernovae \citep[see e.g.,][for reviews]{mezzacappa_2005,kotake_etal_2006,janka_2012,burrows_2013}.  
Ultimately, this requires solving the general relativistic Boltzmann equation for the neutrino radiation field.  
Thus, for completeness (and to provide the proper context), in \ref{sec:equations} we list the conservative, general relativistic Boltzmann equation, as well as the limiting cases solved numerically in Section~\ref{sec:numericalExamples}.  
(These latter equations are derived directly from the former equation.)  

%% file: max.tex
\section{Bound-Preserving Numerical Methods for Phase Space Advection}
\label{sec:maximumPrinciple}

In this section, we present the bound-preserving (BP) method for the transport equation \eqref{eq:ConservativeBoltzmannEquationCurvilinearIntro}.  
We first examine the general case with curvilinear phase space coordinates and identify necessary conditions for the BP property.   
We then consider specific examples using commonly adopted phase space coordinates and show how to enforce these conditions in each case.

\subsection{The General Case}
\label{sec:maximumPrincipleGeneral}

\subsubsection{Preliminaries}
\label{sec:maximumPrincipleGeneralPrelim}

We denote the phase-space coordinates and flux coefficients using $\vect{z} = \big(\vect{x},\vect{p}\big)$ and $\vect{H}=\big(\vect{F},\vect{G}\big)$, respectively, and let $\tau:=\sqrt{\gamma \lambda} \geq 0$.  
Then we rewrite \eqref{eq:ConservativeBoltzmannEquationCurvilinearIntro} in the compact form%
\footnote{It will be necessary later on to split the phase space back up into position space and momentum space parts.}
\begin{equation}
  \p_t f+ \frac{1}{\tau}\sumz \p_{z^i} (\,\tau\,H^{i}\,f\,) = 0. 
  \label{eq:ConservativeBoltzmannEquationCurvilinearCompact}
\end{equation}
Note that both $\tau$ and $\vect{H}$ depend on $\vect{z}$; i.e., $\tau \colon\bbR^{d_z}\to\bbR$ and $\vect{H} \colon \bbR^{d_z}\to\bbR^{d_z}$, where $d_z := d_x + d_p$,

We divide the phase space domain $D$ into a disjoint union $\cT$ of open elements $\bK$, so that $D = \cup_{\bK \in \cT} \operatorname{cl}(\bK)$.  We require that each element is a box in the logical coordinates
\begin{equation}
  \bK = \{\vect{z} : z^{i} \in K^{i} := (\zL^{i}, \zH^{i})\}.
  \label{eq:phaseSpaceElement}
\end{equation}
We use $V_{\bK}$ to denote the volume of the phase space cell
\begin{equation}
  V_{\bK} = \int_{\bK}dV, \quad\text{where}\quad dV = \tau \prodz dz^{i}.  
\end{equation}

The proof of the bound-preserving property (cf. Section \ref{sec:maximumPrincipleGeneralProof}) requires some analysis on the surface of the cell.  
For this reason we introduce, for each $i$, the decomposition $ \vz = \{\tz^{i},z^{i}\}$ along with the associated notations
\begin{equation}
  d\tV^{i} = \prod_{j \ne i} dz^{j}
  \quand
  \tK^{i} = \otimes_{j \ne i} K^{j}.
\end{equation}
In particular, $dV=\tau\,d\tV^{i}\,dz^{i}$ and $\bK=\tK^{i}\otimes K^{i}$.  
We also define $\dz^{i}=\zH^{i}-\zL^{i}$.  

The approximation space for the DG method, $\bbV^k$, is
\begin{equation}\label{ldg:vhk}
  \bbV^k=\{v : v\big|_{\bK} \in \bbQ^k(\bK), \, \, \forall\ \bK\in \cT \},
\end{equation}
where $\bbQ^{k}$ is the space of tensor products of one-dimensional polynomials of maximal degree $k$.  
Note that functions in $\bbV^k$ can be discontinuous across element interfaces. 
 
The semi-discrete DG problem is to find $\fDG \in C^1([0,\infty);\bbV^k)$ (which approximates $f$ in \eqref{eq:ConservativeBoltzmannEquationCurvilinearCompact}), such that
\begin{equation}\label{eq:SemidiscreteDG}
  \p_t \int_{\bK} \fDG\,v\,dV
  + \sumz \int_{\tK^{i}} \big(\,v\,\tau\,\widehat{H^{i}\fDG}\big|_{\zH^{i}}-v\,\tau\,\widehat{H^{i}\fDG}\big|_{\zL^{i}}\,\big)\,d\tV^{i} 
  - \sumz \int_\bK H^{i}\fDG\,\pderiv{v}{z^{i}}\,dV = 0, 
\end{equation}
for all $v \in \bbV^k$ and $\bK \in \cT$.  Here $\widehat{H^{i}\fDG}$ is a numerical flux approximating the phase-space flux on the $i$th surface of the phase-space element $\bK$. 
The upwind flux is utilized in this paper which, due to the curvilinear coordinates, must take into account the fact that $H^i$ depends on the phase space coordinates $\vect{z}$.  
For any $\vect{z} \in D$ and $v \in \bbV^k$, $\widehat{H^{i}v}\big|_{z^{i}} = \cH^i(v(z^{i,-},\tz^{i}),v(z^{i,+},\tz^{i});\vect{z})$, where for each $i$, the numerical flux function $\cH^{i}: \bbR^{2}\times\bbR^{d_{z}} \to \bbR$ is given by
\begin{align}
 \cH^i(a,b;\vect{\zeta}) 
  &=(0\vee H^{i}|_{\zeta^{i}})\,a+(0\wedge H^{i}|_{\zeta^{i}})\,b,
  \label{eq:generalUpwindFlux}
\end{align}%
and for any $a,b \in \bbR$, $a \vee b \equiv \max(a,b)$ and $a \wedge b \equiv \min(a,b)$.  
Superscripts $-/+$, e.g., in the arguments of $v(z^{i,-/+},\tz^{i})$, indicate that the function is evaluated to the immediate left/right of $z^{i}$.  
As in the Cartesian case, $\cH^i$ is non-decreasing in the first argument and non-increasing in the second. 

\subsubsection{Proof of Bound-Preserving (BP) Property}
\label{sec:maximumPrincipleGeneralProof}

The space $\bbV^k$ contains the constant functions, and the choice $v=1$ in \eqref{eq:SemidiscreteDG} gives 
\begin{equation}
  \p_t \bar{f}_\bK 
  +\frac{1}{V_\bK}\sumz \int_{\tK^{i}} \big(\,\tau\,\widehat{H^{i}\fDG}\big|_{\zH^{i}} - \tau\,\widehat{H^{i}\fDG}\big|_{\zL^{i}}\,\big)\,d\tV^{i} = 0,
  \label{eq:ConservativeBoltzmannEquationCurvilinearCompactFiniteVolume}
\end{equation}
where $\bar{f}_\bK$ is the cell average of $\fDG$ in $\bK$:
\begin{equation}
  \bar{f}_\bK
  =\f{1}{V_{\bK}}\int_\bK\fDG\,dV
  =\f{1}{V_{\bK}}\int_{\tK^{i}}\int_{K^{i}}\fDG\,\tau\,dz^{i}\,d\tV^{i} \quad (\forall i).
\end{equation}
In this paper, we use SSP-RK time integrators, which are convex combinations of forward Euler time steps \citep{gottlieb_etal_2001}.  
Thus, without loss of generality, we consider here only a forward Euler time step for $\bar{f}_\bK$; i.e., 
\begin{equation}
  \bar{f}_{\bK}^{n+1}
  =\bar{f}_{\bK}^{n}
    -\frac{\dt}{V_\bK}\sumz \int_{\tK^{i}} 
    \big(\,\tau\,\widehat{H^{i}\fDG^{n}}\big|_{\zH^{i}} 
    - \tau\,\widehat{H^{i}\fDG^{n}}\big|_{\zL^{i}}\,\big)
    \,d\tV^{i},
    \label{eq:averageUpdate}
\end{equation}
where $\dt=t^{n+1}-t^{n}$ is the time step.  

Sufficient conditions to ensure $\bar{f}_{\bK}^{n+1} \geq 0$ are given in Lemma \ref{lemma1} below.  
However, to ensure $\bar{g}_{\bK}^{n+1}\equiv1-\bar{f}_{\bK}^{n+1}$ is also non-negative, additional conditions on the numerical divergence are needed.  
These are stated in Lemma \ref{lem:divFree}.  
The positivity-preserving properties of the DG scheme follow from Lemma~\ref{lem:gammPositive}, while the bound-preserving properties are summarized in Theorem~\ref{thm:upperLowerBound}.  

\begin{lemma}\label{lemma1}
  Let $\{s_{i}\}_{i=1}^{d_{z}}$ be a set of positive constants (independent of $\vect{z}$ and $t$) satisfying $\sumz s_{i} =1$.  
  If for each $i\in \{1,\ldots,d_{z}\}$,
  \begin{equation}
    \Gamma^{i}[\fDG^{n}](\tz^{i}) := \int_{K^{i}} \fDG^{n}\,\tau\,dz^{i}
    -\frac{\dt}{s_{i}}\big(\,\tau\,\widehat{H^{i}\fDG^{n}}\big|_{\zH^{i}}-\tau\,\widehat{H^{i}\fDG^{n}}\big|_{\zL^{i}}\,\big)
    \ge0 \quad (\forall\tz^{i}\in\tK^{i}),
    \label{eq:positivityCondition}
  \end{equation}
  then $\bar{f}_{\bK}^{n+1}\ge0$.  
  Moreover, if $\tS^{i}\in\tK^{i}$ is a set of quadrature points where the corresponding quadrature integrates $\Gamma^{i}[\fDG^{n}]$ over $\tK^{i}$ (i.e., the integral in \eqref{eq:averageUpdateGamma}) exactly, then $\Gamma^{i}[\fDG^{n}]$ need only be non-negative on $\tS^{i}$.  
\end{lemma}

\begin{proof}
  It is simple to show from \eqref{eq:averageUpdate} that
  \begin{equation}
    \bar{f}_{\bK}^{n+1} = \f{1}{V_{\bK}}\sumz s_{i}\int_{\tK^{i}}\Gamma^{i}[\fDG^{n}]\,d\tV^{i}.
    \label{eq:averageUpdateGamma}
  \end{equation}
  The result follows immediately.  
\end{proof}
\begin{rem}
  The motivation for the constants $s_i$ is the formula \eqref{eq:averageUpdateGamma}.  
  In order to maintain positivity, each of the terms $\Gamma^{i}[\fDG^{n}]$ will be controlled individually (cf. Lemma~\ref{lem:gammPositive} below.)
\end{rem}

\begin{lemma}\label{lem:divFree}
  Let $\gDG=1-\fDG$, and assume that the conditions for $\Gamma^{i}[\gDG^{n}]\ge0$ in Lemma~\ref{lemma1} hold.  
  Suppose that the divergence-free condition in \eqref{eq:divergenceFreePhaseSpaceFlowCurvilinearAbstract} is satisfied, i.e., that
  \begin{equation}
    \f{1}{V_{\bK}}\sumz\int_{\tK^{i}}\big(\,\tau\,H^{i}\big|_{\zH^{i}} - \tau\,H^{i}\big|_{\zL^{i}}\,\big)\,d\tV^{i}=0.
    \label{eq:divergenceFreePhaseSpaceFlowCurvilinearAbstractDiscrete}
  \end{equation}
  Then $\bar{g}_{\bK}^{n+1}\ge0$, which implies $\bar{f}_{\bK}^{n+1}\le1$.
\end{lemma}

\begin{proof}
  A direct calculation shows that for any $\vect{z} \in D$ and $v \in \bbV^k$,
  \begin{equation}
    \widehat{H^{i}(1-v)}\big|_{z^{i}} = H^i \big|_{z^{i}} - \widehat{H^{i}v}\big|_{z^{i}}.
  \end{equation}
  Thus, with $v = \fDG^{n}$, we find
  \begin{align}
  \bar{g}_{\bK}^{n+1}
  &= (1- \bar{f}_{\bK}^{n}) 
  + \frac{\dt}{V_\bK}\sumz \int_{\tK^{i}} \big(\,\tau\,\widehat{H^{i}\fDG^{n}}\big|_{\zH^{i}} - \tau\,\widehat{H^{i}\fDG^{n}}\big|_{\zL^{i}}\,\big)\,d\tV^{i} \nonumber \\
  &= \bar{g}_{\bK}^{n}
  -  \frac{\dt}{V_\bK}\sumz \int_{\tK^{i}} \big(\,\tau\,\widehat{H^{i}\gDG}\big|_{\zH^{i}} - \tau\,\widehat{H^{i}\gDG}\big|_{\zL^{i}}\,\big)\,d\tV^{i} \nonumber \\
  & \hspace{1.0cm} +  \frac{\dt}{V_\bK}\sumz \int_{\tK^{i}} \big(\,\tau\,H^i\big|_{\zH^{i}} - \tau\,H^i\big|_{\zL^{i}}\,\big)\,d\tV^{i}.
  \end{align}
  If the divergence-free condition in \eqref{eq:divergenceFreePhaseSpaceFlowCurvilinearAbstractDiscrete} holds, then the final term vanishes, leaving
  \begin{equation*}
    \bar{g}_{\bK}^{n+1}
    =\bar{g}_{\bK}^{n}
    -\f{\dt}{V_\bK}\sumz \int_{\tK^{i}} 
    		\big(\,\tau\,\widehat{H^{i}\gDG^{n}}\big|_{\zH^{i}} 
    		- \tau\,\widehat{H^{i}\gDG^{n}}\big|_{\zL^{i}}\,\big)
    \,d\tV^{i}
    =\f{1}{V_{\bK}}\sumz s_{i}\int_{\tK^{i}}\Gamma^{i}[\gDG^{n}]\,d\tV^{i}.
  \end{equation*}
  Since the conditions under Lemma~\ref{lemma1} hold, i.e., $\Gamma^{i}[\gDG^{n}]\ge0$, it follows that $\bar{g}_{\bK}^{n+1}\ge0$.
\end{proof}

We now proceed to find conditions for which \eqref{eq:positivityCondition} holds. 
To simplify notation, we temporarily drop the index $i$, setting $z = z^{i}$, $K = K^{i}$, $\Gamma=\Gamma^{i}$, etc. 
Let $\hat{Q}$ denote the $N$-point \emph{Gauss-Lobatto} quadrature rule on the interval $K=(\zL,\zH)$, with points
\begin{equation}\label{eq:QuadLoba}
  \hat{S}=\left\{\zL=\hat{z}_{1},\hat{z}_{2},\cdots,\hat{z}_{N-1},\hat{z}_{N}=\zH\right\}, 
\end{equation}
and weights $\hat{w}_{q} \in (0,1]$, normalized so that $\sum_{q} \hat{w}_{q} = 1$.  
(The hat is used to specifically denote the Gauss-Lobatto rule.)
This quadrature integrates polynomials in $z \in \bbR$ with degree up to $2N-3$ exactly.  
Thus if $\fDG\tau$ is such a polynomial,%
\footnote{In situations where $\tau$ is not a polynomial, one may use a polynomial approximation for $\tau$.}
then the integral of $\fDG$ is exact; i.e.,
\begin{equation}
  \label{eq:exactQuad}
  \int_{K} \fDG^{n}\,\tau\,dz = \hat{Q}[\fDG^{n}] \equiv
  \dz \sum_{q=1}^{N} \hat{w}_{q}\,\fhDG{q}^{n}\,\hat{\tau}_{q} 
\end{equation}
where $\fhDG{q}^{n} := \fDG^n(\hat{z}_q)$ and  $\hat{\tau}_{q} := \tau(\hat{z}_q)$. 

Using the quadrature rule \eqref{eq:exactQuad} and the numerical flux function $\cH$ in \eqref{eq:generalUpwindFlux}, we find
\begin{align}  \label{eq:positivityConditionRewritten}
  \frac{\Gamma[\fDG^{n}]}{\dz}
  &=
   \sum_{q=2}^{N-1}\hat{w}_{q}\,\fhDG{q}^{n}\,\hat{\tau}_{q}
  +\hat{\tau}_{1}\,\hat{w}_{1}\,\Phi_1(\fhDG{1}^{n,-},\fhDG{1}^{n,+})
  +\hat{\tau}_{N}\,\hat{w}_{N}\,\Phi_N(\fhDG{N}^{n,-},\fhDG{N}^{n,+}),
\end{align}  
where 
\begin{equation}\label{eq:Phi}
  \Phi_1(a,b)  = b + T_1\,\cH(a,b ;\hat{z}_1)  \quand
  \Phi_N(a,b)  = a - T_N\,\cH(a, b; \hat{z}_N), 
\end{equation}
with
\begin{equation}
  T_1 = \frac{\dt}{ \hat{w}_{1}\,s\,\dz} \quand T_N = \frac{\dt}{\hat{w}_{N}\,s\,\dz}.
\end{equation}  

\begin{lemma}\label{prop:PhiProperties}
  The functions $\Phi_1$ and $\Phi_N$ satisfy $\Phi_1(0,0) = \Phi_N(0,0)=0$.  
  In addition, the derivatives are
  \begin{align*}
    \pderiv{\Phi_{1}}{a}&=T_{1}\,(0\vee H(\hat{z}_1)), &\quad\
    \pderiv{\Phi_{N}}{a}&=1-T_{N}\,(0\vee H(\hat{z}_N)), \\
    \pderiv{\Phi_{1}}{b}&=1+T_{1}\,(0\wedge H(\hat{z}_1)), &\quad\
    \pderiv{\Phi_{N}}{b}&=-T_{N}\,(0\wedge H(\hat{z}_N)).  
  \end{align*}
\end{lemma}
\begin{proof}
The proof is a direct calculation using the definitions of $\Phi_{1}$ and $\Phi_{N}$ in \eqref{eq:Phi} and $\cH^{i}$ in \eqref{eq:generalUpwindFlux}.
\end{proof}

The following lemma establishes sufficient conditions for $\bar{f}_{\bK}^{n+1}\ge0$ due to \eqref{eq:averageUpdateGamma}.  
\begin{lemma}\label{lem:gammPositive}
  Suppose that
  \begin{enumerate}
    \item The quadrature rule $\hat{Q}$ integrates $\fDG^{n}\,\tau$ exactly.
    \item For all $\hat{z}_q \in \hat{S}$ and all $\tz\in\tK$, $\fDG^{n}(\hat{z}_q,\tz) \ge 0$.
    \item The time step $\dt$ is chosen such that 
    \begin{equation} \label{eq:ConditionDeltat1}
      1-T_1\,|0\wedge H(\hat{z}_{1},\tz)| \geq 0 \quand 1-T_N\,(0\vee H(\hat{z}_{N},\tz)) \geq 0
    \end{equation}
    for all $\tz\in\tK$. 
  \end{enumerate}
  Then $\Gamma[\fDG^{n}](\tz) \geq 0$.
\end{lemma}

\begin{proof}
  To show that $\Gamma[\fDG^{n}](\tz) \ge0$, it is sufficient to show that each of the three terms in Equation \eqref{eq:positivityConditionRewritten} are non-negative.  
  The first term is non-negative by assumption $2$, and the fact that the quadrature weights are positive and $\hat{\tau}_{q}\ge0$.  
  To handle the second and third term in \eqref{eq:positivityConditionRewritten}, containing $\Phi_1$ and $\Phi_N$, respectively, we note that $\partial\Phi_{1}/\partial a$ and $\partial\Phi_{N}/\partial b$ are always non-negative while $\partial\Phi_{1}/\partial b$ and $\partial\Phi_{N}/\partial a$ are non-negative under the CFL constraint in \eqref{eq:ConditionDeltat1}.  
  Therefore
  \begin{equation*}
    0 = \Phi_{1}(0,0) \leq \Phi_1(\fhDG{1}^{n,-},\fhDG{1}^{n,+}) \quand
    0 = \Phi_{N}(0,0) \leq \Phi_N(\fhDG{N}^{n,-},\fhDG{N}^{n,+}).
  \end{equation*}
  Hence $\Gamma[\fDG^{n}] \ge 0$.  
\end{proof}
\begin{rem}
  The condition in \eqref{eq:ConditionDeltat1} must be satisfied for each phase space dimension---that is for all $i\in\{1,\ldots,d_z\}$ (cf. Lemma \ref{lemma1}).  
  In particular, it requires
  \begin{equation}
    \dt\le\min\Big[\,\f{1}{|0\wedge H^i(\hat{z}^i_{1},\tz^i)|},\f{1}{(0\vee H^i(\hat{z}^i_{N^i},\tz^i))}\,\Big]\,\hat{w}_{N^i}\,s_{i}\,\dz^{i}.
    \label{eq:generalCFL}
  \end{equation}
  for all $i\in\{1,\ldots,d_{z}\}$.  
  If $\tS^{i}\in\tK^{i}$ are the quadrature points of the quadrature used to integrate $\Gamma^{i}[\fDG^{n}]$ over $\tK$, then \eqref{eq:generalCFL} must hold for all $\tz^{i}\in\tS^{i}$.  
\end{rem}

We now return to \eqref{eq:averageUpdateGamma} and introduce a quadrature rule for evaluating $\bar{f}_{\bK}^{n+1}$.  
For each $i\in \{1,\ldots,d_z\}$, let $\tQ^{i} \colon C^0(\tK^i) \to \bbR$ be a quadrature rule with positive weights and points $\tS^{i} \subset \operatorname{cl}(\tK^{i})$.
Let $\hat{\bS}^i = \tS^{i} \otimes \hat{S}^i$ and define $\hat{\bQ}^i \colon C^0(\bK) \to \bbR$ by $\hat{\bQ}^i = \tQ^{i} \circ \hat{Q}^i$, where $\hat{Q}^i \colon C^0(K^i) \to \bbR$ is the Gauss-Lobatto quadrature rule with points $\hat{S}^i$.  
With the quadrature rule, we combine the previous results to establish the following theorem.  
\begin{theorem}\label{thm:upperLowerBound}
  Suppose that 
  \begin{enumerate}
    \item For all $i\in\{1,\ldots,d_{z}\}$, the Gauss-Lobatto quadrature rule $\hat{Q}^{i}$ is chosen such that $\eqref{eq:exactQuad}$ holds.  
    \item For all $i\in\{1,\ldots,d_{z}\}$, the quadrature rule $\tQ^{i}$ integrates $\Gamma^i[\fDG]$ over $\tK^{i}$ exactly, and preserves a discrete version of the divergence-free condition \eqref{eq:divergenceFreePhaseSpaceFlowCurvilinearAbstractDiscrete}; i.e.,  
    \begin{equation}
    \f{1}{V_{\bK}}\sumz \tQ^{i}\big(\,\tau\,H^{i}\big|_{\zH^{i}} - \tau\,H^{i}\big|_{\zL^{i}}\,\big)=0.  
    \label{eq:divergenceFreePhaseSpaceFlowCurvilinearAbstractDiscreteQuad}
  \end{equation}
    \item For all $i\in\{1,\ldots,d_z\}$ and all $\vz \in \hat{\bS}^i$, $0\le\fDG^{n}(\vz)\le1$.  
    \item The time step $\dt$ satisfies \eqref{eq:generalCFL} for all $i\in\{1,\ldots,d_z\}$ and all $\tz^{i}\in\tS^{i}$.  
  \end{enumerate}
  Then $0\le\bar{f}_{\bK}^{n+1}\le1$.
\end{theorem}

\begin{proof}
  For each $i\in\{1,\ldots,d_z\}$, the fact that $\fDG^{n} \ge0$ on $\hat{\bS}^{i}$ implies, via Lemma \ref{lem:gammPositive}, that $\Gamma^{i}[\fDG] \geq 0$ on $\tS^{i}$. 
  Repeating the same argument with $\gDG^{n}\ge0$ shows that $\Gamma^i[\gDG^{n}] \geq 0$ as well; this fact will be used later in ensuring the upper bound on $\bar{f}_{\bK}^{n+1}$.  
  Returning to the lower bound, we compute $\bar{f}_{\bK}^{n+1}$ using \eqref{eq:averageUpdateGamma}
  \begin{align}
    \bar{f}_{\bK}^{n+1}
    = \frac{1}{V_{\bK}} \sumz s_{i}\,\tQ^i(\Gamma^{i}[\fDG]) \geq 0.
  \end{align}
  Here we have replaced the integral over $\tK^i$ in \eqref{eq:averageUpdate} by the quadrature rule $\tQ^i$, which by assumption~$2$ is exact.  
  This ensures the lower bound on $\bar{f}_{\bK}^{n+1}$.  
  To ensure the upper bound, we invoke Lemma~\ref{lem:divFree}, using \eqref{eq:divergenceFreePhaseSpaceFlowCurvilinearAbstractDiscreteQuad} and the fact that $\Gamma^{i}[\gDG^{n}] \geq 0$ to conclude that $\bar{g}_{\bK}^{n+1} = 1 - \bar{f}_{\bK}^{n+1} \geq 0$. 
  This concludes the proof.
\end{proof}

\begin{rem}
  It is important to note that the quadratures $\{\tQ^i\}_{i=1}^{d_z}$ must be used in the implementation of the method.  
  The Gauss-Lobatto quadrature $\{\hat{Q}^i\}_{i=1}^{d_z}$ need not be.  
  Rather, the latter is used to obtain the CFL condition in assumption~$4$ of Theorem~\ref{thm:upperLowerBound}.  
  In practice, a different quadrature $Q^i$ (typically Gauss-Legendre) is used to evaluate integrals over $K^i$.  
\end{rem}
\begin{rem}
  The first condition in assumption~$2$ of Theorem~\ref{thm:upperLowerBound} is actually stronger than necessary.  
  Indeed, each $\tQ^{i}$ needs only to integrate $\fDG^{n}$ over $\tK^{i}$ exactly; the surface integrals of the flux terms can be approximate, provided \eqref{eq:divergenceFreePhaseSpaceFlowCurvilinearAbstractDiscreteQuad} still holds.  
  In such cases, the numerical method yields a quadrature approximation of $\bar{f}_{\bK}^{n+1}$ that is still provably bound-preserving.  
  We maintain here the condition in assumption~$2$ for simplicity and note that it is, in fact, satisfied for the examples we consider in Section~\ref{sec:numericalExamples}.
\end{rem}

In the following subsections, we explore specific examples in more detail and examine the implication of \eqref{eq:ConditionDeltat1} on the time step in each case.  
In each of the examples, we specify the quadratures, i.e., $\tQ^{i}$, needed to integrate $\Gamma^{i}[\fDG^{n}]$ and satisfy the divergence-free condition in \eqref{eq:divergenceFreePhaseSpaceFlowCurvilinearAbstractDiscreteQuad}.  
The bound-enforcing limiter in \citep{ZS2010a}, which we discuss briefly in Section~\ref{sec:limiter}, is used to ensure assumption~$3$ in Theorem~\ref{thm:upperLowerBound}.  

\subsection{Spherical Symmetry, Flat Spacetime (1D~x+1D~p)}
\label{sec:maximumPrincipleSphericalSymmetry}

For a flat, spherically symmetric spacetime, adopting spherical polar phase space coordinates, the phase-space is $D = \{(r,\mu)\in\bbR^2 : r\geq 0,\, \mu \in [-1,1]\}$ and the collisionless Boltzmann equation (cf. \eqref{eq:ConservativeBoltzmannEquationSphericalSymmetryFlatApp}),
\begin{equation}
  \pderiv{f}{t}
  +\f{1}{r^{2}}\pderiv{}{r}\Big(\,r^{2}\,\mu\,f\,\Big)
  +\pderiv{}{\mu}\Big(\,\big(1-\mu^{2}\big)\,\f{1}{r}\,f\,\Big)=0,
  \label{eq:ConservativeBoltzmannEquationSphericalSymmetryFlat}
\end{equation}
takes the form of \eqref{eq:ConservativeBoltzmannEquationCurvilinearCompact} with $z^{1}=r$, $z^{2}=\mu$, $\tau=r^{2}$, $H^{1}=\mu\equiv H^{(r)}$, and $H^{2}=(1-\mu^2)/r\equiv H^{(\mu)}$.  
Here the position coordinate $r$ is the radial distance from the origin and the momentum coordinate $\mu$ is the cosine of the angle between the particle direction of flight and the radial direction.  
(See \ref{sec:equations} for further details.)  

For this case, the phase-space element is%
\footnote{To clarify the presentation, we modify slightly the general notation of the previous section, replacing indices $i$ with the appropriate coordinate names.  
For example, $K^{1}=K^{(r)}$ and $K^{2}=K^{(\mu)}$.  Similar modifications are made in the following sections.}
\begin{equation}
  \bK=\{(r,\mu) \in \bbR^2: r \in K^{(r)}:=(\rL,\rH),\, \mu\in K^{(\mu)}:=(\muL,\muH) \},
\end{equation}
and, for any $v\in\bbV^{k}$ the upwind numerical fluxes are given by
\begin{align*}
  \widehat{H^{(r)}v }(r,\mu) 
  &=\f{1}{2}\big(\mu+|\mu|\big)\,v(r^-,\mu)+\f{1}{2}\big(\mu-|\mu|\big)\,v(r^+,\mu),  \\
  \widehat{H^{(\mu)}v}(r,\mu)
  &=\f{1}{r}(1-\mu^{2})\,v(r,\mu^{-}). 
\end{align*}
Then, for any $(r,\mu) \in D$ and any $v \in \bbV^k$, the DG method is as follows: 
\textit{Find $\fDG \in \bbV^k$ such that}
\begin{align}
  &
  \int_{\bK}\p_{t}\fDG\,v\,r^2drd\mu
  -\int_{\bK}H^{(r)}\fDG\,\p_{r}v\,r^{2}drd\mu
  -\int_{\bK}H^{(\mu)}\fDG\,\p_{\mu}v\,r^2 drd\mu \nonumber \\
  &\hspace{12pt}
  +\rH^2\int_{K^{(\mu)}}\widehat{H^{(r)}\fDG }(\rH,\mu)\,v(\rH^{-},\mu)\,d\mu
  -\rL^2\int_{K^{(\mu)}}\widehat{H^{(r)}\fDG }(\rL,\mu)\,v(\rL^{+},\mu)\,d\mu \nonumber \\
  &\hspace{12pt}
  +\int_{K^{(r)}}\widehat{H^{(\mu)}\fDG }(r,\muH)\,v(r,\muH^{-})\,r^2dr
  -\int_{K^{(r)}}\widehat{H^{(\mu)}\fDG }(r,\muL)\,v(r,\muL^{+})\,r^2dr = 0, 
  \label{eq:ConservativeBoltzmannSphericalSymmetryDG}
\end{align}
\textit{for all $v \in \bbV^{k}$ and all $\bK\in D$.} 
In particular, the update for the cell-average is
\begin{align}
  \bar{f}_{\bK}^{n+1}
  &=\bar{f}_{\bK}^{n}
  -\f{\Delta t}{V_{\bK}}
  \Big\{\,
    \rH^{2}\int_{K^{(\mu)}}\widehat{H^{(r)}\fDG^{n}}(\rH,\mu)\,d\mu
    -\rL^{2}\int_{K^{(\mu)}}\widehat{H^{(r)}\fDG^{n}}(\rL,\mu)\,d\mu \nonumber \\
    &\hspace{72pt}
    +\int_{K^{(r)}}\widehat{H^{(\mu)}\fDG^{n}}(r,\muH)\,r^{2}dr
    -\int_{K^{(r)}}\widehat{H^{(\mu)}\fDG^{n}}(r,\muL)\,r^{2}dr
  \,\Big\},
  \label{eq:averageUpdateSphericalSymmetry}
\end{align}
where $V_{\bK}=\int_{\bK}r^{2}dr\,d\mu$.  

To satisfy the first two conditions in Theorem~\ref{thm:upperLowerBound}, we define the quadratures
\begin{equation}
  \hat{\bQ}^{(r)}=\tilde{\bQ}^{(r)} \circ \hat{Q}^{(r)}
  \quand
  \hat{\bQ}^{(\mu)}=\tilde{\bQ}^{(\mu)} \circ \hat{Q}^{(\mu)}
  \label{eq:quadratureSphericalSymmetry}
\end{equation}
where $\tilde{\bQ}^{(r)}=Q^{(\mu)}$ and $\tilde{\bQ}^{(\mu)}=Q^{(r)}$, and $Q^{(r)}$ and $Q^{(\mu)}$ are $L^{(r)}$- and $L^{(\mu)}$-point Gauss-Legendre quadratures on $K^{(r)}$ and $K^{(\mu)}$, respectively.  
We denote the Gaussian weights with $\{w_{\alpha}\}_{\alpha=1}^{L}$.  
Similarly, $\hat{Q}^{(r)}$ and $\hat{Q}^{(\mu)}$ are $N^{(r)}$- and $N^{(\mu)}$-point Gauss-Lobatto quadratures on $K^{(r)}$ and $K^{(\mu)}$, respectively.  
The sets of quadrature points associated with $\hat{\bQ}^{(r)}$ and $\hat{\bQ}^{(\mu)}$ are denoted
\begin{equation}
  \hat{\bS}^{(r)}=\tilde{\bS}^{(r)}\otimes\hat{S}^{(r)} \quand
  \hat{\bS}^{(\mu)}=\tilde{\bS}^{(\mu)}\otimes\hat{S}^{(\mu)},
  \label{eq:quadratureSetSphericalSymmetry}
\end{equation}
respectively, where $\tilde{\bS}^{(r)}=S^{(\mu)}$ and $\tilde{\bS}^{(\mu)}=S^{(r)}$, and
\begin{equation}\label{gauss}
  S^{(r)}
  = \{r_{\alpha}:\alpha=1,\cdots ,L^{(r)}\} \quand
  S^{(\mu)}
  = \{\mu_{\alpha}:\alpha=1,\cdots ,L^{(\mu)}\}
\end{equation}
are the Gaussian quadrature points on $K^{(r)}$ and $K^{(\mu)}$, respectively.  
Similarly, 
\begin{equation}\label{gaussLobatto}
  \hat{S}^{(r)}
  = \{\hat{r}_{\alpha}:\alpha=1,\cdots ,N^{(r)}\} \quand
  \hat{S}^{(\mu)}
  = \{\hat{\mu}_{\alpha}:\alpha=1,\cdots ,N^{(\mu)}\}
\end{equation}
are the Gauss-Lobatto quadrature points.  
The Gauss-Lobatto weights, $\{\hat{w}_{\alpha}\}_{\alpha=1}^{N}$, are normalized such that $\sum_{\alpha}\hat{w}_{\alpha}=1$.  
We note that the Gauss-Lobatto quadrature is only introduced to derive the CFL conditions needed to ensure the BP properties of the scheme.  
In the numerical scheme for $\fDG$, we use the Gauss-Legendre quadrature $\bQ=Q^{(r)}\circ Q^{(\mu)}$ to compute the volume integrals in \eqref{eq:ConservativeBoltzmannSphericalSymmetryDG}.  
Gaussian quadratures $Q^{(r)}$ and $Q^{(\mu)}$ are also used to evaluate the flux integrals.  

It is straightforward to show that the discretization for the cell-average in \eqref{eq:averageUpdateSphericalSymmetry} satisfies the divergence-free condition in \eqref{eq:divergenceFreePhaseSpaceFlowCurvilinearAbstractDiscrete} exactly; i.e.,
\begin{align}
  &\f{1}{V_{\bK}}
  \Big\{\,
    \rH^{2}\int_{K^{(\mu)}}H^{(r)}(\rH,\mu)\,d\mu-\rL^{2}\int_{K^{(\mu)}}H^{(r)}(\rL,\mu)\,d\mu \nonumber \\
  &\hspace{36pt}
    +\int_{K^{(r)}}H^{(\mu)}(r,\muH)\,r^{2}dr-\int_{K^{(r)}}H^{(\mu)}(r,\muL)\,r^{2}dr
  \,\Big\}=0, 
  \label{eq:divergenceFreeSphericalSymmetry}
\end{align}
provided $L^{(r)}$, $L^{(\mu)}\ge1$.

To ensure the numerical solutions to \eqref{eq:ConservativeBoltzmannEquationSphericalSymmetryFlat} satisfy the maximum principle, we need to prove the conditions in Theorem \ref{thm:upperLowerBound}, which we state in the following Corollary
\begin{cor}
  Let the update for the cell average be given by \eqref{eq:averageUpdateSphericalSymmetry}.  
  Consider the quadratures in \eqref{eq:quadratureSphericalSymmetry} with $N^{(r)}\ge(k+5)/2$, $L^{(r)},N^{(\mu)}\ge(k+3)/2$, and $L^{(\mu)}\ge(k+1)/2$.  
  Let the polynomial $\fDG^n\in\bbV^{k}$ satisfy $0\le\fDG^{n}\le1$ in the set of quadrature points
  \begin{equation}
    S=\hat{\bS}^{(r)}\,\cup\,\hat{\bS}^{(\mu)}.  
    \label{eq:quadratureSetSphericalSymmetry}
  \end{equation}
  Let the time step $\dt$ satisfy the CFL condition 
  \begin{equation}
    \dt\leq\hat{w}_{N^{(r)}}\,s_{1}\,\Delta r/|\mu_{\alpha}|
    \quand
    \dt\leq\hat{w}_{N^{(\mu)}}\,s_{2}\,r_{\alpha}\,\Delta\mu/(1-\muH^{2}).  
    \label{eq:CFLsphericalSymmetry}
  \end{equation}
  It follows that $0\leq\bar{f}_{\bK}^{n+1}\leq1$.  
\end{cor}

\begin{proof}
  With the quadratures in \eqref{eq:quadratureSphericalSymmetry}, evaluation of the current cell-average gives
  \begin{equation}
    \f{V_{\bK}\,\bar{f}^{n}_{\bK}}{\Delta r\,\Delta\mu}
    =s_{1}\sum_{\alpha,\beta\in\hat{\bS}^{(r)}}\hat{w}_{\alpha}\,w_{\beta}\,\fDG(\hat{r}_{\alpha},\mu_{\beta})\,\hat{r}_{\alpha}^{2}
    +s_{2}\sum_{\alpha,\beta\in\hat{\bS}^{(\mu)}}w_{\alpha}\,\hat{w}_{\beta}\,\fDG(r_{\alpha},\hat{\mu}_{\beta})\,r_{\alpha}^{2},
  \end{equation}
  which is exact, and non-negative since $0\le\fDG^{n}\le1$ in $S$.  
  The quadratures also evaluate the integrals in \eqref{eq:divergenceFreeSphericalSymmetry} exactly, so that the divergence-free condition holds.  
  To compute the bound-preserving CFL conditions, we consider the $r$ and $\mu$ dimensions independently.  
  In the $r$ dimension we have $|0\wedge H^{(r)}(\rL,\mu_{\alpha})|=|0\wedge\mu_{\alpha}|$ and $(0\vee H^{(r)}(\rH,\mu_{\alpha}))=(0\vee\mu_{\alpha})$, so that
  \begin{align}
    \pderiv{\Phi_{1}^{(r)}}{b}
    &=1-\f{\dt}{\hat{w}_{1}\,s_{1}\,\Delta r}\,|0\wedge\mu_{\alpha}|, \\
    \pderiv{\Phi_{N^{(r)}}^{(r)}}{a}
    &=1-\f{\dt}{\hat{w}_{N^{(r)}}\,s_{1}\,\Delta r}\,(0\vee\mu_{\alpha}),
  \end{align}
  which are non-negative provided the first condition in \eqref{eq:CFLsphericalSymmetry} holds.  
  Next, we consider the $\mu$ dimension.  
  Since $H^{(\mu)}=(1-\mu^2)/r \geq 0$, we have $|0\wedge H^{(\mu)}(r_{\alpha},\muL)|=0$ and $(0\vee H^{(\mu)}(r_{\alpha},\muH))=(1-\muH^{2})/r_{\alpha}$, so that $\partial\Phi_{1}^{(\mu)}/\partial b=1$ and
  \begin{equation}
    \pderiv{\Phi_{N^{(\mu)}}^{(\mu)}}{a}
    =1-\f{\Delta t}{\hat{w}_{N^{(\mu)}}\,s_{2}\,\Delta\mu}\f{(1-\muH^{2})}{r_{\alpha}}, 
  \end{equation}
  which are non-negative if the second condition in \eqref{eq:CFLsphericalSymmetry} holds.  
  Therefore, the CFL condition in \eqref{eq:generalCFL} becomes \eqref{eq:CFLsphericalSymmetry}.  
  It follows that $0\le\bar{f}_{\bK}^{n+1}\le1$.  
\end{proof}

\subsection{Spherical Symmetry, Curved Spacetime (1D~x+2D~p)}
\label{sec:maximumPrincipleSphericalSymmetryGR}

In this section, we consider the spherically symmetric problem in curved spacetime.  
Adopting spherical polar phase space coordinates, the phase space is $D=\{(r,\mu,\tetradEpsilon)\in\bbR^{3} : r\ge0\,, \mu \in [-1,1]\,, \tetradEpsilon\ge0\}$ and the collisionless Boltzmann equation (cf. \eqref{eq:ConservativeBoltzmannEquationSphericalSymmetryGRApp}),
\begin{align}
  &
  \f{1}{\alpha}\pderiv{f}{t}
  +\f{1}{\alpha\,\psi^{6}\,r^{2}}\pderiv{}{r}\Big(\,\alpha\,\psi^{4}\,r^{2}\,\mu\,f\,\Big)
  -\f{1}{\tetradEpsilon^{2}}\pderiv{}{\tetradEpsilon}
  \Big(\,\tetradEpsilon^{3}\,\f{1}{\psi^{2}\,\alpha}\pderiv{\alpha}{r}\,\mu\,f\,\Big) \nonumber \\
  & \hspace{12pt}
  +\pderiv{}{\mu}
  \Big(\,\big(1-\mu^{2}\big)\,\psi^{-2}\,
    \Big\{\,
      \f{1}{r}
      +\f{1}{\psi^{2}}\pderiv{\psi^{2}}{r}
      -\f{1}{\alpha}\pderiv{\alpha}{r}
    \,\Big\}\,f
  \,\Big)
  =0, 
  \label{eq:ConservativeBoltzmannEquationSphericalSymmetryGR}
\end{align}
can be written in the form of \eqref{eq:ConservativeBoltzmannEquationCurvilinearCompact} with $z^{1}=r$, $z^{2}=\mu$, $z^{3}=\tetradEpsilon$, $\tau=\psi^{6}\,r^{2}\,\tetradEpsilon^{2}$, $H^{1}=\f{\alpha}{\psi^{2}}\mu\equiv H^{(r)}$, $H^{2}=\f{\alpha}{\psi^{2}\,r}\big(1-\mu^{2}\big)\,\Psi\equiv H^{(\mu)}$, and $H^{3}=-\tetradEpsilon \f{1}{\psi^{2}}\pderiv{\alpha}{r}\,\mu\equiv H^{(\tetradEpsilon)}$, where we have defined
\begin{equation}
  \Psi=1+r\,\p_{r}\ln\psi^{2}-r\,\p_{r}\ln\alpha.
\end{equation}
As in the previous section, the position coordinate $r$ is the radial distance from the origin and the momentum coordinate $\mu$ is the cosine of the angle between the particle propagation direction and the radial direction.  
In addition, $\tetradEpsilon$ is the particle energy, while $\alpha(r)\ge0$ is the lapse function and $\psi(r)\ge0$ is the conformal factor.  
(See \ref{sec:equations} for further details.)  
Equation \eqref{eq:ConservativeBoltzmannEquationSphericalSymmetryGR} reduces to \eqref{eq:ConservativeBoltzmannEquationSphericalSymmetryFlat} in the case of a flat spacetime ($\alpha=\psi=1$).  

The phase space element is now
\begin{equation}
  \bK=\{(r,\mu,\tetradEpsilon)\in\bbR^{3} : r\in K^{(r)},\, \mu\in K^{(\mu)},\, \tetradEpsilon\in K^{(\tetradEpsilon)}:=(\eL,\eH)\},
\end{equation}
and, for any $v\in\bbV^{k}$ the upwind numerical fluxes are given by
\begin{align*}
  \widehat{H^{(r)}v}(r,\mu,\tetradEpsilon)
  &=\f{\alpha}{\psi^{2}}
  \Big\{\,
    \f{1}{2}\big(\mu+|\mu|\big)\,v(r^{-},\mu,\tetradEpsilon)+\f{1}{2}\big(\mu-|\mu|\big)\,v(r^{+},\mu,\tetradEpsilon)
  \,\Big\}, \\
  \widehat{H^{(\mu)}v}(r,\mu,\tetradEpsilon)
  &=\f{\alpha}{\psi^{2}\,r}\,(1-\mu^{2})\,
  \Big\{\,
    \f{1}{2}\big(\Psi+|\Psi|\big)\,v(r,\mu^{-},\tetradEpsilon) \nonumber \\
    &\hspace{108pt}
    +\f{1}{2}\big(\Psi-|\Psi|\big)\,v(r,\mu^{+},\tetradEpsilon)
  \,\Big\}, \\
  \widehat{H^{(\tetradEpsilon)}v}(r,\mu,\tetradEpsilon)
  &=-\f{\tetradEpsilon}{\psi^{2}}
  \Big\{\,
    \f{1}{2}\big(\p_{r}\alpha\,\mu-|\p_{r}\alpha\,\mu|\big)\,v(r,\mu,\tetradEpsilon^{-}) \nonumber \\
    &\hspace{66pt}
    +\f{1}{2}\big(\p_{r}\alpha\,\mu+|\p_{r}\alpha\,\mu|\big)\,v(r,\mu,\tetradEpsilon^{+})
  \,\Big\}.
\end{align*}
(Note the sign difference on the energy flux term.)  
Then, for any $(r,\mu,\tetradEpsilon) \in D$ and any $v \in \bbV^{k}$, the DG method is as follows:
\textit{Find $\fDG \in \bbV^{k}$ such that}
\begin{align}
  &
  \int_{\bK}\p_{t}\fDG\,v\,dV
  -\int_{\bK}H^{{r}}\fDG\,\p_{r}v\,dV
  -\int_{\bK}H^{(\mu)}\fDG\,\p_{\mu}v\,dV
  -\int_{\bK}H^{(\tetradEpsilon)}\fDG\,\p_{\tetradEpsilon}v\,dV \nonumber \\
  & \hspace{12pt}
  + \int_{\tK^{(r)}}\widehat{H^{(r)}\fDG}(\rH,\mu,\tetradEpsilon)\,v(\rH^{-},\mu,\tetradEpsilon)\,\tau(\rH,\tetradEpsilon)\,d\tV^{(r)} \nonumber \\
  & \hspace{48pt}
  - \int_{\tK^{(r)}}\widehat{H^{(r)}\fDG}(\rL,\mu,\tetradEpsilon)\,v(\rL^{+},\mu,\tetradEpsilon)\,\tau(\rL,\tetradEpsilon)\,d\tV^{(r)} \nonumber \\
   & \hspace{12pt}
  + \int_{\tK^{(\mu)}}\widehat{H^{(\mu)}\fDG}(r,\muH,\tetradEpsilon)\,v(r,\muH^{-},\tetradEpsilon)\,\tau(r,\tetradEpsilon)\,d\tV^{(\mu)} \nonumber \\
  & \hspace{48pt}
  - \int_{\tK^{(\mu)}}\widehat{H^{(\mu)}\fDG}(r,\muL,\tetradEpsilon)\,v(r,\muL^{+},\tetradEpsilon)\,\tau(r,\tetradEpsilon)\,d\tV^{(\mu)} \nonumber \\
  & \hspace{12pt}
  + \int_{\tK^{(\tetradEpsilon)}}\widehat{H^{(\tetradEpsilon)}\fDG}(r,\mu,\eH)\,v(r,\mu,\eH^{-})\,\tau(r,\eH)\,d\tV^{(\tetradEpsilon)} \nonumber \\
  & \hspace{48pt}
  - \int_{\tK^{(\tetradEpsilon)}}\widehat{H^{(\tetradEpsilon)}\fDG}(r,\mu,\eL)\,v(r,\mu,\eL^{+})\,\tau(r,\eL)\,d\tV^{(\tetradEpsilon)}=0, 
  \label{eq:ConservativeBoltzmannSphericalSymmetryGRDG}
\end{align}
\textit{for all $v \in \bbV^{k}$ and all $\bK\in\cT$}.  
In \eqref{eq:ConservativeBoltzmannSphericalSymmetryGRDG}, we have defined phase-space volume and ``area" elements
\begin{equation}
  dV=\tau\,dr\,d\mu\,d\tetradEpsilon,\quad
  d\tV^{(r)}=d\mu\,d\tetradEpsilon,\quad
  d\tV^{(\mu)}=dr\,d\tetradEpsilon,\quad
  d\tV^{(\tetradEpsilon)}=dr\,d\mu,
\end{equation}
and subelements
\begin{equation}
  \tK^{(r)}=K^{(\mu)}\otimes K^{(\tetradEpsilon)},\quad
  \tK^{(\mu)}=K^{(r)}\otimes K^{(\tetradEpsilon)},\quad
  \tK^{(\tetradEpsilon)}=K^{(r)}\otimes K^{(\mu)}.  
\end{equation}
The update for the cell-average in \eqref{eq:averageUpdate} becomes
\begin{align}
  \bar{f}_{\bK}^{n+1}
  &=\bar{f}_{\bK}^{n}
  -\f{\Delta t}{V_{\bK}}
  \Big\{\,
    \psi^{6}(\rH)\,\rH^{2}\int_{\tK^{(r)}}\widehat{H^{(r)}\fDG^{n}}(\rH,\mu,\tetradEpsilon)\,\tetradEpsilon^{2}\,d\tV^{(r)} \nonumber \\
    &\hspace{96pt}
    -\psi^{6}(\rL)\,\rL^{2}\int_{\tK^{(r)}}\widehat{H^{(r)}\fDG^{n}}(\rL,\mu,\tetradEpsilon)\,\tetradEpsilon^{2}\,d\tV^{(r)} \nonumber \\
    &\hspace{72pt}
    +\int_{\tK^{(\mu)}}\widehat{H^{(\mu)}\fDG^{n}}(r,\muH,\tetradEpsilon)\,\psi^{6}(r)\,r^{2}\,\tetradEpsilon^{2}\,d\tV^{(\mu)} \nonumber \\
    &\hspace{96pt}
    -\int_{\tK^{(\mu)}}\widehat{H^{(\mu)}\fDG^{n}}(r,\muL,\tetradEpsilon)\,\psi^{6}(r)\,r^{2}\,\tetradEpsilon^{2}\,d\tV^{(\mu)} \nonumber \\
    &\hspace{72pt}
    +\eH^{2}\int_{\tK^{(\tetradEpsilon)}}\widehat{H^{(\tetradEpsilon)}\fDG^{n}}(r,\mu,\eH)\,\psi^{6}(r)\,r^{2}\,d\tV^{(\tetradEpsilon)} \nonumber \\
    &\hspace{96pt}
    -\eL^{2}\int_{\tK^{(\tetradEpsilon)}}\widehat{H^{(\tetradEpsilon)}\fDG^{n}}(r,\mu,\eL)\,\psi^{6}(r)\,r^{2}\,d\tV^{(\tetradEpsilon)}
  \,\Big\},
  \label{eq:averageUpdateSphericalSymmetryGR}
\end{align}
where $V_{\bK}=\int_{\bK}\psi^{6}\,r^{2}dr\,d\mu\,\tetradEpsilon^{2}d\tetradEpsilon$.  

To satisfy the first two conditions in Theorem~\ref{thm:upperLowerBound}, we define the quadratures
\begin{equation}
  \hat{\bQ}^{(r)}
  =\tilde{\bQ}^{(r)} \circ \hat{Q}^{(r)}, \quad
  \hat{\bQ}^{(\mu)}
  =\tilde{\bQ}^{(\mu)} \circ \hat{Q}^{(\mu)}, \quand
  \hat{\bQ}^{(\tetradEpsilon)}
  =\tilde{\bQ}^{(\tetradEpsilon)} \circ \hat{Q}^{(\tetradEpsilon)},
  \label{eq:quadratureSphericalSymmetryGR}
\end{equation}
where $\tilde{\bQ}^{(r)}=Q^{(\mu)} \circ Q^{(\tetradEpsilon)}$, $\tilde{\bQ}^{(\mu)}=Q^{(r)} \circ Q^{(\tetradEpsilon)}$, and $\tilde{\bQ}^{(\tetradEpsilon)}=Q^{(r)} \circ Q^{(\mu)}$, and, in addition to the quadratures defined in Section~\ref{sec:maximumPrincipleSphericalSymmetry}, $Q^{(\tetradEpsilon)}$ is an $L^{(\tetradEpsilon)}$-point Gauss-Legendre quadrature on $K^{(\tetradEpsilon)}$.  
Similarly, $\hat{Q}^{(\tetradEpsilon)}$ is an $N^{(\tetradEpsilon)}$-point Gauss-Lobatto quadrature on $K^{(\tetradEpsilon)}$.  
The quadrature points associated with $\hat{\bQ}^{(r)}$, $\hat{\bQ}^{(\mu)}$, and $\hat{\bQ}^{(\tetradEpsilon)}$ are
\begin{equation}
  \hat{\bS}^{(r)}=\tS^{(r)}\otimes\hat{S}^{(r)}, \quad
  \hat{\bS}^{(\mu)}=\tS^{(\mu)}\otimes\hat{S}^{(\mu)},\quand
  \hat{\bS}^{(\tetradEpsilon)}=\tS^{(\tetradEpsilon)}\otimes\hat{S}^{(\tetradEpsilon)}.  
\end{equation}
We let $S^{(r)}$ and $S^{(\mu)}$ be the Gaussian sets defined in \eqref{gauss}, and $\hat{S}^{(r)}$ and $\hat{S}^{(\mu)}$ the Gauss-Lobatto sets defined in \eqref{gaussLobatto}.  
In a similar manner, we let
\begin{equation}
  S^{(\tetradEpsilon)}
  =\{\tetradEpsilon_{\alpha}:\alpha=1,\ldots,L^{(\tetradEpsilon)}\}\quand
  \hat{S}^{(\tetradEpsilon)}
  =\{\hat{\tetradEpsilon}_{\alpha}:\alpha=1,\ldots,N^{(\tetradEpsilon)}\}
\end{equation}
be Gaussian and Gauss-Lobatto quadrature points on $K^{(\tetradEpsilon)}$ with associated weights $\{w_{\alpha}\}_{\alpha=1}^{L^{(\tetradEpsilon)}}$ and $\{\hat{w}_{\alpha}\}_{\alpha=1}^{N^{(\tetradEpsilon)}}$, respectively, normalized so that $\sum_{\alpha}w_{\alpha}=\sum_{\alpha}\hat{w}_{\alpha}=1$.  
Then we let 
\begin{equation}
  \tS^{(r)}=S^{(\mu)} \otimes S^{(\tetradEpsilon)}, \quad
  \tS^{(\mu)}=S^{(r)} \otimes S^{(\tetradEpsilon)}, \quand
  \tS^{(\tetradEpsilon)}=S^{(r)} \otimes S^{(\mu)}.  
\end{equation}
Again, the Gauss-Lobatto quadrature is only introduced to derive the CFL conditions needed to ensure the BP properties.  
For implementation, we use the Gauss-Legendre quadrature $\bQ=Q^{(r)} \circ Q^{(\mu)} \circ Q^{(\tetradEpsilon)}$ to compute the volume integrals in \eqref{eq:ConservativeBoltzmannSphericalSymmetryGRDG}.  
Gaussian quadratures $\tilde{\bQ}^{(r)}$, $\tilde{\bQ}^{(\mu)}$, and $\tilde{\bQ}^{(\tetradEpsilon)}$ are also used to evaluate the flux integrals.  

For this problem, the divergence-free condition in \eqref{eq:divergenceFreePhaseSpaceFlowCurvilinearAbstractDiscrete} becomes
\begin{align}
  \f{1}{V_{\bK}}
  \Big\{\,
    &\psi^{6}(\rH)\,\rH^{2}\int_{\tK^{(r)}}H^{(r)}(\rH,\mu,\tetradEpsilon)\,\tetradEpsilon^{2}\,d\tV^{(r)} \nonumber \\
    &\hspace{36pt}
    -\psi^{6}(\rL)\,\rL^{2}\int_{\tK^{(r)}}H^{(r)}(\rL,\mu,\tetradEpsilon)\,\tetradEpsilon^{2}\,d\tV^{(r)} \nonumber \\
    &+\int_{\tK^{(\mu)}}H^{(\mu)}(r,\muH,\tetradEpsilon)\,\psi^{6}(r)\,r^{2}\,\tetradEpsilon^{2}\,d\tV^{(\mu)} \nonumber \\
    &\hspace{36pt}
    -\int_{\tK^{(\mu)}}H^{(\mu)}(r,\muL,\tetradEpsilon)\,\psi^{6}(r)\,r^{2}\,\tetradEpsilon^{2}\,d\tV^{(\mu)} \nonumber \\
    &+\eH^{2}\int_{\tK^{(\tetradEpsilon)}}H^{(\tetradEpsilon)}(r,\mu,\eH)\,\psi^{6}(r)\,r^{2}\,d\tV^{(\tetradEpsilon)} \nonumber \\
    &\hspace{36pt}
    -\eL^{2}\int_{\tK^{(\tetradEpsilon)}}H^{(\tetradEpsilon)}(r,\mu,\eL)\,\psi^{6}(r)\,r^{2}\,d\tV^{(\tetradEpsilon)}
  \,\Big\} = 0.
  \label{eq:divergenceFreeSphericalSymmetryGR}
\end{align}
In Equation~\eqref{eq:ConservativeBoltzmannSphericalSymmetryGRDG}, we approximate the derivatives $\p_{r}\alpha$ and $\p_{r}\psi^{4}$ in $K^{(r)}$ with polynomials and compute $\alpha$ and $\psi^{4}$ from
\begin{equation}
  \alpha(r)=\alpha(\rL)+\int_{\rL}^{r}\p_{r}\alpha(r')\,dr'\quad\mbox{and}\quad
  \psi^{4}(r)=\psi^{4}(\rL)+\int_{\rL}^{r}\p_{r}\psi^{4}(r')\,dr',
\end{equation}
where the Gaussian quadrature rule is used to evaluate the integrals exactly.  
With this choice, it is straightforward to show that the discretization satisfies the divergence-free condition \eqref{eq:divergenceFreeSphericalSymmetryGR}, provided $L^{(\mu)}\ge1$, $L^{(\tetradEpsilon)}\ge2$, while $L^{(r)}$ depends on the degree of the polynomials approximating $\p_{r}\alpha$ and $\p_{r}\psi^{4}$.  

To ensure the numerical solutions to \eqref{eq:ConservativeBoltzmannEquationSphericalSymmetryGR} satisfy the maximum principle, we need to prove the conditions in Theorem~\ref{thm:upperLowerBound}, which we state in the following Corollary
\begin{cor}
  Let the update for the cell average be given by \eqref{eq:averageUpdateSphericalSymmetryGR}.  
  Consider the quadratures in \eqref{eq:quadratureSphericalSymmetryGR} with $N^{(r)}\ge(k+k_{\psi}+5)/2$, $L^{(r)}\ge(k+k_{\psi}+3)/2$, $N^{(\mu)}\ge(k+3)/2$, $L^{(\mu)}\ge(k+1)/2$, $N^{(\tetradEpsilon)}\ge(k+5)/2$, and $L^{(\tetradEpsilon)}\ge(k+3)/2$, where $k_{\psi}$ is the degree of the polynomial used to approximate $\psi^{6}(r)$.  
  (We also let the degree of the polynomials approximating $\p_{r}\alpha$ and $\p_{r}\psi^{4}$ be equal to $k_{\psi}$.)
  Let the polynomial $\fDG^n\in\bbV^{k}$ satisfy $0\le\fDG^{n}\le1$ in the quadrature set
  \begin{equation}
    S
    =\hat{\bS}^{(r)}\,\cup\,\hat{\bS}^{(\mu)}\,\cup\,\hat{\bS}^{(\tetradEpsilon)}.  
    \label{eq:quadratureSetSphericalSymmetryGR}
  \end{equation}
  Let the time step $\dt$ satisfy
  \begin{align}
    \dt
    &\le\min\big(\psi^{2}(\rL)/\alpha(\rL),\psi^{2}(\rH)/\alpha(\rH)\big)\,\hat{w}_{N^{(r)}}\,s_{1}\,\Delta r/|\mu_{\alpha}|, 
    \label{eq:CFLSphericalSymmetryRadius} \\
    \dt
    &\le\f{\hat{w}_{N^{(\mu)}}\,s_{2}\,r_{\alpha}\,\Delta\mu\,\psi^{2}(r_{\alpha})/\alpha(r_{\alpha})}{|\Psi(r_{\alpha})|\,\max\big((1-\muL^{2}),(1-\muH^{2})\big)}, 
    \label{eq:CFLSphericalSymmetryAngle} \\
    \dt
    &\le\f{\hat{w}_{N^{(\tetradEpsilon)}}\,s_{3}\,\Delta\tetradEpsilon\,\psi^{2}(r_{\alpha})}{|\p_{r}\alpha(r_{\alpha})\,\mu_{\beta}|\,\eH}.  
    \label{eq:CFLSphericalSymmetryEnergy}
  \end{align}
  It follows that $0 \leq \overline{f}_{\bK}^{n+1} \leq 1$.  
\end{cor}
\begin{proof}
  With the quadratures in \eqref{eq:quadratureSphericalSymmetryGR}, evaluation of the current cell-average gives
  \begin{align}
    \f{V_{\bK}\,\bar{f}_{\bK}^{n}}{\Delta r\,\Delta\mu\,\Delta\tetradEpsilon}
      &
      =s_{1}\sum_{\alpha,\beta,\gamma\in\hat{\bS}^{(r)}}
      \hat{w}_{\alpha}\,w_{\beta}\,w_{\gamma}\,\fDG^{n}(\hat{r}_{\alpha},\mu_{\beta},\tetradEpsilon_{\gamma})\,\psi^{6}(\hat{r}_{\alpha})\,\hat{r}_{\alpha}^{2}\,\tetradEpsilon_{\gamma}^{2} \nonumber \\
      &\hspace{12pt}
      +s_{2}\sum_{\alpha,\beta,\gamma\in\hat{\bS}^{(\mu)}}
      w_{\alpha}\,\hat{w}_{\beta}\,w_{\gamma}\,\fDG^{n}(r_{\alpha},\hat{\mu}_{\beta},\tetradEpsilon)\,\psi^{6}(r_{\alpha})\,r_{\alpha}^{2}\,\tetradEpsilon_{\gamma}^{2} \nonumber \\
      &\hspace{12pt}
      +s_{3}\sum_{\alpha,\beta,\gamma\in\hat{\bS}^{(\tetradEpsilon)}}
      w_{\alpha}\,w_{\beta}\,\hat{w}_{\tetradEpsilon}\,\fDG^{n}(r_{\alpha},\mu_{\beta},\hat{\tetradEpsilon}_{\gamma})\,\psi^{6}(r_{\alpha})\,r_{\alpha}^{2}\,\hat{\tetradEpsilon}_{\gamma}^{2},
  \end{align}
  which is exact, and non-negative since $0\le\fDG^{n}\le1$ in $S$.  
  The divergence-free condition holds since the quadratures in \eqref{eq:quadratureSphericalSymmetryGR} evaluate the integrals in \eqref{eq:divergenceFreeSphericalSymmetryGR} exactly.  
  To compute the bound-preserving CFL conditions, we consider the three dimensions ($r$, $\mu$, and $\tetradEpsilon$) independently.  
  Defining $\widehat{\alpha}=\alpha/\psi^{2}>0$, we have in the radial dimension $|0\wedge H^{(r)}(\rL,\mu_{\alpha},\tetradEpsilon_{\beta})|=\widehat{\alpha}(\rL)\,|0\wedge\mu_{\alpha}|$, and $(0\vee H^{(r)}(\rH,\mu_{\alpha},\tetradEpsilon_{\beta}))=\widehat{\alpha}(\rH)\,(0\vee\mu_{\alpha})$, so that
  \begin{align}
    \pderiv{\Phi_{1}^{(r)}}{b}
    &=1-\f{\Delta t}{\hat{w}_{1}\,s_{1}\,\Delta r}\,\widehat{\alpha}(\rL)\,|0\wedge\mu_{\alpha}|, \\
    \pderiv{\Phi_{N^{(r)}}^{(r)}}{a}
    &=1-\f{\Delta t}{\hat{w}_{N^{(r)}}\,s_{1}\,\Delta r}\,\widehat{\alpha}(\rH)\,(0\vee\mu_{\alpha}),
  \end{align}
  which are non-negative if \eqref{eq:CFLSphericalSymmetryRadius} holds. 
  Therefore, the CFL condition in \eqref{eq:generalCFL} becomes \eqref{eq:CFLSphericalSymmetryRadius}. 
  In the $\mu$ dimension, $|0\wedge H^{(\mu)}(r_{\alpha},\muL,\tetradEpsilon_{\beta})|=(1-\muL^{2})\,\widehat{\alpha}(r_{\alpha})/r_{\alpha}\,|0\wedge\Psi(r_{\alpha})|$ and $(0\vee H^{(\mu)}(r_{\alpha},\muH,\tetradEpsilon_{\beta}))=(1-\muL^{2})\,\widehat{\alpha}(r_{\alpha})/r_{\alpha}\,(0\vee\Psi(r_{\alpha}))$, which give
  \begin{align}
    \pderiv{\Phi_{1}^{(\mu)}}{b}
    &=1-\f{\Delta t}{\hat{w}_{1}\,s_{2}\,\Delta\mu}\,\f{(1-\muL^{2})\,\widehat{\alpha}(r_{\alpha})}{r_{\alpha}}\,|0\wedge\Psi(r_{\alpha})|, \\
    \pderiv{\Phi_{N^{(\mu)}}^{(\mu)}}{a}
    &=1-\f{\Delta t}{\hat{w}_{N^{(\mu)}}\,s_{2}\,\Delta\mu}\,\f{(1-\muH^{2})\,\widehat{\alpha}(r_{\alpha})}{r_{\alpha}}\,(0\vee\Psi(r_{\alpha})), 
  \end{align}
  which are non-negative if \eqref{eq:CFLSphericalSymmetryAngle} holds, so that the CFL condition in \eqref{eq:generalCFL} becomes \eqref{eq:CFLSphericalSymmetryAngle}.  
  Finally, in the $\tetradEpsilon$ dimension we have $|0\wedge H^{(\tetradEpsilon)}(r_{\alpha},\mu_{\beta},\eL)|=\eL\,|0\wedge-\p_{r}\alpha(r_{\alpha})\,\mu_{\beta}|/\psi^{2}(r_{\alpha})$ and $(0\vee H^{(\tetradEpsilon)}(r_{\alpha},\mu_{\beta},\eH))=\eH\,(0\vee-\p_{r}\alpha(r_{\alpha})\,\mu_{\beta})/\psi^{2}(r_{\alpha})$, so that
  \begin{align}
    \pderiv{\Phi_{1}^{(\tetradEpsilon)}}{b}
    &=1-\f{\dt}{\hat{w}_{1}\,s_{3}\,\Delta\tetradEpsilon}\,\f{\eL}{\psi^{2}(r_{\alpha})}\,|0\wedge-\p_{r}\alpha(r_{\alpha})\,\mu_{\beta}|, \\
    \pderiv{\Phi_{N^{(\tetradEpsilon)}}^{(\tetradEpsilon)}}{a}
    &=1-\f{\Delta t}{\hat{w}_{N^{(\tetradEpsilon)}}\,s_{3}\,\Delta\tetradEpsilon}\,\f{\eH}{\psi^{2}(r_{\alpha})}\,(0\vee-\p_{r}\alpha(r_{\alpha})\,\mu_{\beta}),
  \end{align}
  which are non-negative if \eqref{eq:CFLSphericalSymmetryEnergy} holds, so that the CFL condition in \eqref{eq:generalCFL} becomes \eqref{eq:CFLSphericalSymmetryEnergy}.  
  It follows that $0\le\bar{f}_{\bK}^{n+1}\le1$.  
\end{proof}

\subsection{Axial Symmetry, Flat Spacetime (2D~x+2D~p)}
\label{sec:maximumPrincipleAxialSymmetry}

For a flat, axially symmetric spacetime, adopting cylindrical spatial coordinates and spherical momentum coordinates, the phase space is $D = \{(\rPerp,z,\mu,\tetradPhi)\in\bbR^{4} : \rPerp\ge0,\,z\in\bbR,\,\mu\in[-1,1],\,\tetradPhi\in[0,\pi]\}$ and the collision-less Boltzmann equation (cf. \eqref{eq:ConservativeBoltzmannEquationAxialSymmetryFlatApp}),
\begin{align}
  \pderiv{f}{t}
  +\f{1}{\rPerp}\pderiv{}{\rPerp}\Big(\rPerp\,\sqrt{1-\mu^{2}}\,\cos\tetradPhi\,f\Big)
  +\pderiv{}{z}\Big(\mu\,f\Big)
  -\f{1}{\rPerp}\pderiv{}{\tetradPhi}
  \Big(
    \sqrt{1-\mu^{2}}\,\sin\tetradPhi\,f
  \Big)
  =0,
  \label{eq:ConservativeBoltzmannEquationAxialSymmetryFlat}
\end{align}
can be rewritten in the form of \eqref{eq:ConservativeBoltzmannEquationCurvilinearCompact} with $z^{1}=\rPerp$, $z^{2}=z$, $z^{3}=\mu$, $z^{4}=\tetradPhi$, $\tau=\rPerp$, $H^{1}=\sqrt{1-\mu^{2}}\,\cos\tetradPhi\equiv H^{(\rPerp)}$, $H^{2}=\mu\equiv H^{(z)}$, $H^{3}=0$, and $H^{4}=-\sqrt{1-\mu^{2}}\,\sin\tetradPhi/\rPerp\equiv H^{(\tetradPhi)}$.  
Here the position coordinates $\rPerp$ and $z$ are the distance from the $z$-axis and the distance along the $z$-axis, respectively.  
The momentum coordinate $\mu$ is the cosine of the angle between the particle direction of flight and the $z$ direction, and $\tetradPhi$ is the angle between the projected particle direction of flight and the $\rPerp$ direction.  
(See \ref{sec:equations} for further details.)  

The phase space element is now
\begin{align}
  \bK
  &=\{(\rPerp,z,\mu,\tetradPhi)\in\bbR^{4} :
  \rPerp\in K^{(\rPerp)}:=(\RL,\RH),
  z\in K^{(z)}:=(\zL,\zH), \nonumber \\
  &\hspace{48pt}
  \mu\in K^{(\mu)}:=(\muL,\muH),
  \tetradPhi\in K^{(\tetradPhi)}:=(\PhiL,\PhiH)\},
\end{align}
and, for all $v\in\bbV^{k}$ the upwind numerical fluxes are given by
\begin{align*}
  \widehat{H^{(\rPerp)}v}(\rPerp,z,\mu,\tetradPhi)
  &=\sqrt{1-\mu^{2}}\,
  \Big\{\,
    \f{1}{2}\big(\cos\tetradPhi+|\cos\tetradPhi|\big)\,v(\rPerp^{-},z,\mu,\tetradPhi) \nonumber \\
    &\hspace{68pt}
    +\f{1}{2}\big(\cos\tetradPhi-|\cos\tetradPhi|\big)\,v(\rPerp^{+},z,\mu,\tetradPhi)
  \,\Big\}, \\
  \widehat{H^{(z)}v}(\rPerp,z,\mu,\tetradPhi)
  &=\f{1}{2}\big(\mu+|\mu|\big)\,v(\rPerp,z^{-},\mu,\tetradPhi)+\f{1}{2}\big(\mu-|\mu|\big)\,v(\rPerp,z^{+},\mu,\tetradPhi),  \\
  \widehat{H^{(\tetradPhi)}v}(\rPerp,z,\mu,\tetradPhi)
  &=-\sqrt{1-\mu^{2}}\,\sin\tetradPhi\,v(\rPerp,z,\mu,\tetradPhi^{+}) /\rPerp.  
\end{align*}
(Note that $\sin\tetradPhi\ge0$ in axial symmetry since $\tetradPhi\in[0,\pi]$, and $\sqrt{1-\mu^{2}}\ge0$ since $\mu\in[-1,1]$.)  
Then, for any $(\rPerp,z,\mu,\tetradPhi) \in D$ and any $v \in \bbV^{k}$, the DG method is as follows:
\textit{Find $\fDG \in \bbV^{k}$ such that}
\begin{align}
  &
  \int_{\bK}\p_{t}\fDG\,v\,dV
  -\int_{\bK}H^{(\rPerp)}\fDG\,\p_{\rPerp}v \,dV
  -\int_{\bK}H^{(z)}\fDG\,\p_{z}v\,dV 
  -\int_{\bK}H^{(\tetradPhi)}\fDG\,\p_{\tetradPhi}v\,dV \nonumber \\
  &\hspace{12pt}
  + \RH\int_{\tK^{(\rPerp)}}\widehat{H^{(\rPerp)}\fDG}(\RH,z,\mu,\tetradPhi) \,v(\RH^{-},z,\mu,\tetradPhi)\,d\tV^{(\rPerp)} \nonumber \\
  &\hspace{48pt}
   - \RL\int_{\tK^{(\rPerp)}}\widehat{H^{(\rPerp)}\fDG}(\RL,z,\mu,\tetradPhi) \,v(\RL^{+},z,\mu,\tetradPhi)\,d\tV^{(\rPerp)} \nonumber \\
  &\hspace{12pt}
  + \int_{\tK^{(z)}}\widehat{H^{(z)}\fDG}(\rPerp,\zH,\mu,\tetradPhi) \,v(\rPerp,\zH^{-},\mu,\tetradPhi)\,\rPerp\, d\tV^{(z)} \nonumber \\
  &\hspace{48pt}
  - \int_{\tK^{(z)}}\widehat{H^{(z)}\fDG}(\rPerp,\zL,\mu,\tetradPhi) \,v(\rPerp,\zL^{-},\mu,\tetradPhi)\,\rPerp\, d\tV^{(z)}   \nonumber \\
  &\hspace{12pt}
  + \int_{\tK^{(\tetradPhi)}}\widehat{H^{(\tetradPhi)}\fDG}(\rPerp,z,\mu,\PhiH) \,v(\rPerp,z,\mu,\PhiH^{-})\,\rPerp\, d\tV^{(\tetradPhi)} \nonumber \\
  &\hspace{48pt}
  - \int_{\tK^{(\tetradPhi)}}\widehat{H^{(\tetradPhi)}\fDG}(\rPerp,z,\mu,\PhiL) \,v(\rPerp,z,\mu,\PhiL^{+})\,\rPerp\, d\tV^{(\tetradPhi)} 
  = 0,  \label{eq:ConservativeBoltzmannAxialSymmetryDG}
\end{align}
\textit{for all $v \in \bbV^{k}$ and all $\bK\in\cT$.}  
In \eqref{eq:ConservativeBoltzmannAxialSymmetryDG}, we have defined phase-space volume element
\begin{equation}
  dV=\rPerp\,d\rPerp\,dz\,d\mu\,d\tetradPhi, 
\end{equation}
``area" elements
\begin{equation}
  d\tV^{(\rPerp)}=dz\,d\mu\,d\tetradPhi, 
  \quad
  d\tV^{(z)}=d\rPerp\,d\mu\,d\tetradPhi,
  \quad
  d\tV^{(\tetradPhi)}=d\rPerp\,dz\,d\mu, 
\end{equation}
and subelements
\begin{equation}
  \tK^{(\rPerp)}=K^{(z)}\otimes K^{(\mu)}\otimes K^{(\tetradPhi)},\,
  \tK^{(z)}=K^{(\rPerp)}\otimes K^{(\mu)}\otimes K^{(\tetradPhi)},\,
  \tK^{(\tetradPhi)}=K^{(\rPerp)}\otimes K^{(z)}\otimes K^{(\mu)}.  
\end{equation}
In particular, in axial symmetry, the update for the cell-averaged distribution function in \eqref{eq:averageUpdate} becomes
\begin{align}
  \bar{f}_{\bK}^{n+1}
  &=\bar{f}_{\bK}^{n}
  -\f{\Delta t}{V_{\bK}}
  \Big\{\,
    \RH\int_{\tK^{(\rPerp)}}\widehat{H^{(\rPerp)}\fDG^{n}}(\RH,z,\mu,\tetradPhi)\,d\tV^{(\rPerp)} \nonumber \\
    & \hspace{96pt}
    -\RL\int_{\tK^{(\rPerp)}}\widehat{H^{(\rPerp)}\fDG^{n}}(\RL,z,\mu,\tetradPhi)\,d\tV^{(\rPerp)} \nonumber \\
    & \hspace{72pt}
    +\int_{\tK^{(z)}}\widehat{H^{(z)}\fDG^{n}}(\rPerp,\zH,\mu,\tetradPhi)\,\rPerp\,d\tV^{(z)} \nonumber \\
    & \hspace{96pt}
    -\int_{\tK^{(z)}}\widehat{H^{(z)}\fDG^{n}}(\rPerp,\zL,\mu,\tetradPhi)\,\rPerp\,d\tV^{(z)} \nonumber \\
    & \hspace{72pt}
    +\int_{\tK^{(\tetradPhi)}}\widehat{H^{(\tetradPhi)}\fDG^{n}}(\rPerp,z,\mu,\PhiH)\,\rPerp\,d\tV^{(\tetradPhi)} \nonumber \\
    & \hspace{96pt}
    -\int_{\tK^{(\tetradPhi)}}\widehat{H^{(\tetradPhi)}\fDG^{n}}(\rPerp,z,\mu,\PhiL)\,\rPerp\,d\tV^{(\tetradPhi)}
  \,\Big\},
  \label{eq:averageUpdateAxialSymmetry}
\end{align}
where $V_{\bK}=\int_{\bK}\rPerp\,d\rPerp\,dz\,d\mu\,d\tetradPhi$.  

To satisfy the first two conditions in Theorem~\ref{thm:upperLowerBound}, we define the quadratures
\begin{equation}
  \hat{\bQ}^{(\rPerp)}
  =\tQ^{(\rPerp)} \circ \hat{Q}^{(\rPerp)}, \quad
  \hat{\bQ}^{(z)}
  =\tQ^{(z)} \circ \hat{Q}^{(z)}, \quand
  \hat{\bQ}^{(\tetradPhi)}
  =\tQ^{(\tetradPhi)} \circ \hat{Q}^{(\tetradPhi)},
  \label{eq:quadratureAxialSymmetry}
\end{equation}
where $\tQ^{(\rPerp)}=Q^{(z)}\circ Q^{(\mu)}\circ Q^{(\tetradPhi)}$, $\tQ^{(z)}=Q^{(\rPerp)}\circ Q^{(\mu)}\circ Q^{(\tetradPhi)}$, and $\tQ^{(\tetradPhi)}=Q^{(\rPerp)}\circ Q^{(z)}\circ Q^{(\mu)}$.  
Analogous to the previous sections, $Q^{(\rPerp)}$, $Q^{(z)}$, $Q^{(\mu)}$, and $Q^{(\tetradPhi)}$ are $L^{(\rPerp)}$-, $L^{(z)}$-, $L^{(\mu)}$-, and $L^{(\tetradPhi)}$-point Gauss-Legendre quadratures on $K^{(\rPerp)}$, $K^{(z)}$, $K^{(\mu)}$, and $K^{(\tetradPhi)}$, respectively.  
Similarly, $\hat{Q}^{(\rPerp)}$, $\hat{Q}^{(z)}$, and $\hat{Q}^{(\tetradPhi)}$ denote $N^{(\rPerp)}$-, $N^{(z)}$-, and $N^{(\tetradPhi)}$-point Gauss-Lobatto quadratures.  
The quadrature points associated with \eqref{eq:quadratureAxialSymmetry} are
\begin{equation}
  \hat{\bS}^{(\rPerp)}
  =\tS^{(\rPerp)}\otimes\hat{S}^{(\rPerp)}, \quad
  \hat{\bS}^{(z)}
  =\tS^{(z)}\otimes\hat{S}^{(z)}, \quand
  \hat{\bS}^{(\tetradPhi)}
  =\tS^{(\tetradPhi)}\otimes\hat{S}^{(\tetradPhi)},
\end{equation}
where the Gauss-Lobatto quadrature points are
\begin{align}
  \hat{S}^{(\rPerp)}
  &=\{\hat{\rPerp}_{\alpha} : \alpha = 1,\ldots,N^{(\rPerp)}\}, \\
  \hat{S}^{(z)}
  &=\{\hat{z}_{\alpha} : \alpha = 1,\ldots,N^{(z)}\}, \\
  \hat{S}^{(\tetradPhi)}
  &=\{\hat{\tetradPhi}_{\alpha} : \alpha = 1,\ldots,N^{(\tetradPhi)}\}, 
\end{align}
with weights $\{\hat{w}_{\alpha}\}_{\alpha=1}^{N^{(\rPerp)}}$, $\{\hat{w}_{\alpha}\}_{\alpha=1}^{N^{(z)}}$, and $\{\hat{w}_{\alpha}\}_{\alpha=1}^{N^{(\tetradPhi)}}$, normalized so that $\sum_{\alpha}\hat{w}_{\alpha}=1$.  
Moreover,
\begin{equation}
  \tS^{(\rPerp)} = S^{(z)} \otimes S^{(\mu)} \otimes S^{(\tetradPhi)},\,
  \tS^{(z)} = S^{(\rPerp)} \otimes S^{(\mu)} \otimes S^{(\tetradPhi)},\,
  \tS^{(\tetradPhi)} = S^{(\rPerp)} \otimes S^{(z)} \otimes S^{(\mu)},
\end{equation}
where the Gaussian quadrature points are
\begin{align*}
  S^{(\rPerp)}&=\{\rPerp_{\alpha} : \alpha=1,\ldots,L^{(\rPerp)}\}, & S^{(z)}&=\{z_{\alpha} : \alpha=1,\ldots,L^{(z)}\}, \\
  S^{(\mu)}&=\{\mu_{\alpha} : \alpha=1,\ldots,L^{(\mu)}\}, & S^{(\tetradPhi)}&=\{\tetradPhi_{\alpha} : \alpha=1,\ldots,L^{\tetradPhi}\},
\end{align*}
with associated weights $\{w_{\alpha}\}_{\alpha=1}^{N^{(\rPerp)}}$, $\{w_{\alpha}\}_{\alpha=1}^{N^{(z)}}$, $\{w_{\alpha}\}_{\alpha=1}^{N^{(\mu)}}$, and $\{w_{\alpha}\}_{\alpha=1}^{N^{(\tetradPhi)}}$; all normalized so that $\sum_{\alpha}w_{\alpha}=1$.  

For the axially symmetric problem, the divergence-free condition in \eqref{eq:divergenceFreePhaseSpaceFlowCurvilinearAbstractDiscrete} becomes
\begin{align}
  &\f{1}{V_{\bK}}
  \Big\{\,
    \RH\int_{\tK^{(\rPerp)}}H^{(\rPerp)}(\RH,z,\mu,\tetradPhi)\,d\tV^{(\rPerp)}
    -\RL\int_{\tK^{(\rPerp)}}H^{(\rPerp)}(\RL,z,\mu,\tetradPhi)\,d\tV^{(\rPerp)} \nonumber \\
    &
    +\int_{\tK^{(z)}}H^{(z)}(\rPerp,\zH,\mu,\tetradPhi)\,\rPerp\,d\tV^{(z)}
    -\int_{\tK^{(z)}}H^{(z)}(\rPerp,\zL,\mu,\tetradPhi)\,\rPerp\,d\tV^{(z)} \nonumber \\
    &
    +\int_{\tK^{(\tetradPhi)}}H^{(\tetradPhi)}(\rPerp,z,\mu,\PhiH)\,\rPerp\,d\tV^{(\tetradPhi)}
    -\int_{\tK^{(\tetradPhi)}}H^{(\tetradPhi)}(\rPerp,z,\mu,\PhiL)\,\rPerp\,d\tV^{(\tetradPhi)}
  \,\Big\} \nonumber \\
  &=
  \f{\Delta\rPerp\,\Delta z}{V_{\bK}}\int_{K^{(\mu)}}\sqrt{1-\mu^{2}}\,d\mu\,
  \Big\{\,
    \int_{K^{(\tetradPhi)}}\cos\tetradPhi\,d\tetradPhi
    -\big(\sin\PhiH-\sin\PhiL\big)
  \,\Big\}=0.
  \label{eq:divergenceFreeAxialSymmetry}
\end{align}
On the right-hand side of \eqref{eq:divergenceFreeAxialSymmetry}, the integral over the cosine emanates from the flux in the $\rPerp$ dimension, and is not exact if the Gauss-Legendre quadrature is used, and the terms inside the curly brackets cancel only to the accuracy of the quadrature rule.  
However, in the DG scheme we evaluate the integrals over $\tilde{\bK}^{(\rPerp)}$ in \eqref{eq:ConservativeBoltzmannAxialSymmetryDG}, containing the cosine, by performing an integration by parts, which leads to exact cancellation.  
(The integral of $\sqrt{1-\mu^{2}}$ over $K^{(\mu)}$ is also not exact with the Gauss-Legendre quadrature.  
However, this term appears in the exact same way for the $\rPerp$ and $\tetradPhi$ dimension fluxes, and cancel when the integration by parts discussed above is used.)
Thus, the discretization satisfies the divergence-free condition \eqref{eq:divergenceFreeAxialSymmetry}, provided only $L^{(\rPerp)}\ge1$.  

To ensure the numerical solutions to \eqref{eq:ConservativeBoltzmannEquationAxialSymmetryFlat} satisfy the maximum principle, we need to prove the conditions in Theorem~\ref{thm:upperLowerBound}, which we state in the following Corollary
\begin{cor}
  Let the update for the cell average be given by \eqref{eq:averageUpdateAxialSymmetry}.  
  Consider the quadratures in \eqref{eq:quadratureAxialSymmetry} with $N^{(\rPerp)}\ge(k+4)/2$, $N^{(z)},N^{(\tetradPhi)}\ge(k+3)/2$, and $L^{(\rPerp)}\ge(k+2)/2$, $L^{(z)},L^{(\mu)},L^{(\tetradPhi)}\ge(k+1)/2$.  
  Let the polynomial $\fDG^{n}\in\bbV^{k}$ satisfy $0\le\fDG^{n}\le1$ in the quadrature set
  \begin{equation}
    S=\hat{\bS}^{(\rPerp)}\,\cup\,\hat{\bS}^{(z)}\,\cup\,\hat{\bS}^{(\tetradPhi)}.
  \end{equation}
  Let the time step $\dt$ satisfy the CFL condition 
  \begin{equation}
    \frac{\dt}{\Delta\rPerp}
    \leq\frac{\hat{w}_{N^{(\rPerp)}}\,s_{1}}{\sqrt{1-\mu_{\beta}^{2}}\,|\cos\tetradPhi_{\gamma}|},\,
    \frac{\dt}{\Delta z}
    \leq\frac{\hat{w}_{N^{(z)}}\,s_{2}}{|\mu_{\beta}|},\,
    \frac{\dt}{\rPerp_{\alpha}\,\Delta\tetradPhi}
    \leq\frac{\hat{w}_{N^{(\tetradPhi)}}\,s_{3}}{\sqrt{1-\mu_{\gamma}^{2}}\sin\PhiL}.
    \label{eq:CFLaxialSymmetry}
  \end{equation}
  It follows that $0\le\bar{f}_{\bK}^{n+1}\le1$.  
\end{cor}

\begin{proof}
  With the quadratures in \eqref{eq:quadratureAxialSymmetry}, evaluation of the current cell-average gives
  \begin{align}
    \f{V_{\bK}\,\bar{f}_{\bK}^{n}}{\Delta\rPerp\,\Delta z\,\Delta\mu\,\Delta\tetradPhi}
    &
    =s_{1}\sum_{\alpha,\beta,\gamma,\delta\in\hat{\bS}^{(\rPerp)}}
    \hat{w}_{\alpha}\,w_{\beta}\,w_{\gamma}\,w_{\delta}\,\fDG^{n}(\hat{\rPerp}_{\alpha},z_{\beta},\mu_{\gamma},\tetradPhi_{\delta})\,\hat{\rPerp}_{\alpha} \nonumber \\
    &\hspace{12pt}
    +s_{2}\sum_{\alpha,\beta,\gamma,\delta\in\hat{\bS}^{(z)}}
    w_{\alpha}\,\hat{w}_{\beta}\,w_{\gamma}\,w_{\delta}\,\fDG^{n}(\rPerp_{\alpha},\hat{z}_{\beta},\mu_{\gamma},\tetradPhi_{\delta})\,\rPerp_{\alpha} \nonumber \\
    &\hspace{12pt}
    +s_{3}\sum_{\alpha,\beta,\gamma,\delta\in\hat{\bS}^{(\tetradPhi)}}
    w_{\alpha}\,w_{\beta}\,w_{\gamma}\,\hat{w}_{\delta}\,\fDG^{n}(\rPerp_{\alpha},z_{\beta},\mu_{\gamma},\hat{\tetradPhi}_{\delta})\,\rPerp_{\alpha}, 
  \end{align}
  which is exact, and non-negative since $0\le\fDG^{n}\le1$ in $S$.  
  The divergence-free condition in \eqref{eq:divergenceFreeAxialSymmetry} holds exactly, since integration by parts is used for the integral with the cosine.  
  To compute the bound-preserving CFL conditions, we consider the $\rPerp$, $z$, and $\tetradPhi$ dimensions independently.  
  In the $\rPerp$ dimension we have $|0 \wedge H^{(\rPerp)}(\RL,z_{\alpha},\mu_{\beta},\tetradPhi_{\gamma})|=\sqrt{1-(\mu_{\beta}^{2})}\,|0\wedge\cos\tetradPhi_{\gamma}|$ and $(0 \vee H^{(\rPerp)}(\RH,z_{\alpha},\mu_{\beta},\tetradPhi_{\gamma}))=\sqrt{1-(\mu_{\beta}^{2})}\,(0\vee\cos\tetradPhi_{\gamma})$ so that
  \begin{align}
    \pderiv{\Phi_{1}}{b}
    &=1-\f{\dt}{\hat{w}_{1}\,s_{1}\,\Delta\rPerp}\,\sqrt{1-\mu_{\beta}^{2}}\,|0\wedge\cos\tetradPhi_{\gamma}|, \\
    \pderiv{\Phi_{N^{(\rPerp)}}}{a}
    &=1-\f{\dt}{\hat{w}_{N^{(\rPerp)}}\,s_{1}\,\Delta\rPerp}\,\sqrt{1-\mu_{\beta}^{2}}\,(0\vee\cos\tetradPhi_{\gamma}),
  \end{align}
  which are non-negative provided the first condition in \eqref{eq:CFLaxialSymmetry} holds.  
  In the $z$ dimension we have $|0 \wedge H^{(z)}(\rPerp_{\alpha},\zL,\mu_{\beta},\tetradPhi_{\gamma})|=|0\wedge\mu_{\beta}|$ and $(0 \vee H^{(z)}(\rPerp_{\alpha},\zH,\mu_{\beta},\tetradPhi_{\gamma}))=(0\vee\mu_{\beta})$ so that
  \begin{equation}
    \pderiv{\Phi_{1}}{b}=1-\f{\dt}{\hat{w}_{1}\,s_{2}\,\Delta z}\,|0\wedge\mu_{\beta}|
    \quand
    \pderiv{\Phi_{N^{(z)}}}{a}=1-\f{\dt}{\hat{w}_{N^{(z)}}\,s_{2}\,\Delta z}\,(0\vee\mu_{\beta}),
  \end{equation}
  which are non-negative provided the second condition in \eqref{eq:CFLaxialSymmetry} holds.  
  Finally, in the $\tetradPhi$ dimension we have $|0 \wedge H^{(\tetradPhi)}(\rPerp_{\alpha},z_{\beta},\mu_{\gamma},\PhiL)|=\sqrt{1-\mu_{\gamma}^{2}}\,\sin\PhiL/\rPerp_{\alpha}$ and $(0 \vee H^{(\tetradPhi)}(\rPerp_{\alpha},z_{\beta},\mu_{\gamma},\PhiH))=0$, which give
  \begin{equation}
    \pderiv{\Phi_{1}^{(\tetradPhi)}}{b}
    =1-\f{\dt}{\hat{w}_{N^{(\tetradPhi)}}\,s_{3}\,\Delta\tetradPhi}\,\f{\sqrt{1-\mu_{\gamma}^{2}}}{\rPerp_{\alpha}}\,\sin\PhiL
  \end{equation}
  and $\partial\Phi_{N^{(\tetradPhi)}}^{(\tetradPhi)}/\partial a=1$, which are non-negative provided the third condition in \eqref{eq:CFLaxialSymmetry} holds.  
  It follows that $0\le\bar{f}_{\bK}^{n+1}\le1$.  
\end{proof}

%% file: limiter.tex
\section{Bound-Enforcing Limiter for the DG Scheme}
\label{sec:limiter}

Bound-preserving DG methods for the conservative phase space advection problem were developed in Section \ref{sec:maximumPrinciple}.  
The numerical method is designed to preserve the physical bounds of the cell averaged distribution function (i.e., $0\le\bar{f}_{\bK}\le1$), provided sufficiently accurate quadratures $\{\hat{\bQ}^{i}\}_{i=1}^{d_{z}}$ are specified, specific CFL conditions are satisfied, \emph{and} that the polynomial approximating the distribution function inside a phase space element $\bK$ at time $t^{n}$ is bounded in a set of quadrature points (denoted $S$; cf. assumption~$3$ in Theorem~\ref{thm:upperLowerBound}).  
In the DG method, we use the limiter proposed by Zhang \& Shu in \citep{ZS2010a} to enforce bounds on the distribution function.  
Then the DG scheme ensures that the cell averaged distribution, obtained by solving the \emph{conservative} phase space advection problem, satisfies the following maximum principle: 
if for some initial time $t=t^{n}$ we have $0\le\fDG^{n}\le1$ in a finite set of quadrature points $S\in\bK$, then $0\le\bar{f}_{\bK}^{n+1}\le1$ for all $\bK\in\cT$.  

In addition to the quadratures and CFL conditions (cf. Section \ref{sec:maximumPrinciple}), we must ensure that the polynomial approximating the distribution function inside a phase space element satisfies $\fDG^{n}\in[0,1]$ in $S$.  
To this end, we use the limiter suggested by \citep{ZS2010a} \citep[see also][for a review]{ZS2011} and replace the polynomial $\fDG^{n}(\vect{z})$ with the ``limited'' polynomial
\begin{equation}
  \tilde{f}_{\mbox{\tiny DG}}^{n}(\vect{z})=\vartheta\,\fDG^{n}(\vect{z})+(\,1-\vartheta\,)\,\bar{f}_{\bK}^{n}, 
  \label{eq:limitedPolynomial}
\end{equation}
where the limiter $\vartheta$ is given by
\begin{equation}
  \vartheta=\min\Big\{\Big|\f{M-\bar{f}_{\bK}^{n}}{M_{S}-\bar{f}_{\bK}^{n}}\Big|,\Big|\f{m-\bar{f}_{\bK}^{n}}{m_{S}-\bar{f}_{\bK}^{n}}\Big|,1\Big\}, 
  \label{eq:limiter}
\end{equation}
with $m=0$ and $M=1$, and
\begin{equation}
  M_{S}=\max_{\vect{z} \in S}\fDG^{n}(\vect{z}), \qquad m_{S}=\min_{\vect{z} \in S}\fDG^{n}(\vect{z}), 
\end{equation}
and $S$ represents the finite set of quadrature points in $\bK$; cf. \eqref{eq:quadratureSetSphericalSymmetry}, \eqref{eq:quadratureSetSphericalSymmetryGR}, and \eqref{eq:quadratureAxialSymmetry}.  

It has been shown \citep{ZS2010a,ZS2011} that the ``linear scaling limiter" in \eqref{eq:limitedPolynomial}-\eqref{eq:limiter} maintains uniform high order of accuracy.  
Also, note that the limiting procedure is conservative since it preserves the cell averaged distribution function; i.e.,
\begin{equation}
  \f{1}{V_{\bK}}\int_{\bK}\tilde{f}_{\mbox{\tiny DG}}^{n}\,dV=\bar{f}_{\bK}^{n}.  
\end{equation}

%% file: numerics.tex
\section{Numerical Examples}
\label{sec:numericalExamples}

In this section we present numerical results that are obtained with the bound-preserving DG method for each of the cases discussed in detail in Sections \ref{sec:maximumPrincipleSphericalSymmetry}-\ref{sec:maximumPrincipleAxialSymmetry}.  
In addition to the bound-preserving properties, we also demonstrate high order of accuracy for smooth problems, as well as other aspects of solving the phase space advection problem in curvilinear coordinates with the DG method (e.g., errors near the origin in spherical and axial symmetry, and ray effects in axial symmetry). 

For first-, second-, and third-order spatial discretization we employ tensor product polynomial bases (cf. Section \ref{sec:maximumPrincipleGeneral}), constructed by forming the tensor product of one-dimensional piecewise polynomials of degree up to $k=0$, $1$, and $2$, respectvely, which we refer to as DG(0), DG(1), and DG(2), respectively.  
We use Legendre polynomials in each dimensions.  

For the explicit time stepping we use the forward Euler method (FE), or the strong stability preserving Runge-Kutta (SSP-RK) methods \citep[e.g.,][]{gottlieb_etal_2001} for second (RK2) or third (RK3) order temporal accuracy.  
Thus, schemes with overall first, second and third order formal accuracy will be referred to as DG(0)+FE, DG(1)+RK2, and DG(2)+RK3, respectively.  

\subsection{1D~x+1D~p}
\label{sec:numericalExamplesSphericalSymmetry}

The numerical tests in this section involve the spherically symmetric phase space in flat spacetime, using spherical polar position and momentum coordinates.  
That is, we solve Equation (\ref{eq:ConservativeBoltzmannEquationSphericalSymmetryFlat}) for $\fDG^{n}(r,\mu)$ at discrete time levels $t^{n}$.  

\subsubsection{Test with Smooth Analytic Solution}

First we consider a smooth test problem involving both the position and angle coordinates.  
An analytical solution to Equation (\ref{eq:ConservativeBoltzmannEquationSphericalSymmetryFlat}) is given by
\begin{equation}
  f(r,\mu,t)=\exp\big(\,r\,\mu-t\,\big).  
  \label{eq:freeStreamingExponential1D1D}
\end{equation}
This test is of purely academic interest with little practical value, but it is very useful for evaluating the accuracy of the DG method.  
It is similar to the one considered in \citep{machorro_2007} for a steady-state problem with a non-zero right-hand side.  

The computational domain $D=\{(r,\mu)\in\bbR^{2} : r\in[1,3],\,\mu\in[-1,1]\}$ is divided into $N_{r}\times N_{\mu}$ elements, using $N_{r}$ radial zones and $N_{\mu}$ angular zones.
We use the  analytical solution to specify incoming radiation on the boundary $\partial D$ and simulate the evolution from $t=0$ to $t=1$, using the bound-preserving CFL conditions given in Equation~\eqref{eq:CFLsphericalSymmetry}, with $s_{1}=s_{2}=1/2$ to set the time step.  
To evaluate the accuracy and the convergence rates of the different DG schemes, we compute the $L^{1}$-error norm
\begin{equation}
  E^{1}=\f{1}{V_{D}}\sum_{\bK\in\cT}\int_{\bK}|\fDG^{n}(r,\mu)-f(r,\mu,t^{n})|\,dV,
  \label{eq:errorL1}
\end{equation}
at $t^{n}=1$ for various grid resolutions; each using $N_{r}=N_{\mu}$.  
(The integral in \eqref{eq:errorL1} is computed with $3$-point Gaussian quadratures in the $r$ and $\mu$ dimensions, and $V_{D}=52/3$.)  
$L^{1}$ and $L^{\infty}$ errors, and associated convergence rates, for DG(0)+FE, DG(1)+RK2, and DG(2)+RK3 schemes are listed in Table \ref{tab:convergenceFreeStreamingExponential1D1D}.  
\begin{table}
  \begin{center}
  \caption{$L^{1}$, $L^{\infty}$ error norms and convergence rates for the smooth 1D~x+1D~p test.  
  \label{tab:convergenceFreeStreamingExponential1D1D}}
  \begin{tabular}{ccccc}
    \midrule
    & \multicolumn{4}{c}{DG(0)+FE} \\
    \cmidrule(r){2-5}  
    $N_{r}$ & $L^{1}$ Error & Rate & $L^{\infty}$ Error & Rate \\
    \midrule
    8        & $1.68\times10^{-1}$ &       $-$ & $1.71$                      &   $-$ \\  
    16      & $8.37\times10^{-2}$ & $1.01$ & $1.12$                      & 0.60 \\  
    32      & $4.18\times10^{-2}$ & $1.00$ & $6.57\times10^{-1}$ & 0.77 \\ 
    64      & $2.09\times10^{-2}$ & $1.00$ & $3.57\times10^{-1}$ & 0.88 \\  
    128    & $1.05\times10^{-2}$ & $1.00$ & $1.86\times10^{-1}$ & 0.94 \\ 
    256    & $5.24\times10^{-3}$ & $1.00$ & $9.52\times10^{-2}$ & 0.97 \\  
    512    & $2.62\times10^{-3}$ & $1.00$ & $4.81\times10^{-2}$ & 0.98 \\  
    1024  & $1.31\times10^{-3}$ & $1.00$ & $2.42\times10^{-2}$ & 0.99 \\  
    2048  & $6.55\times10^{-4}$ & $1.00$ & $1.21\times10^{-2}$ & 1.00 \\  
    \midrule
    & \multicolumn{4}{c}{DG(1)+RK2} \\
    \cmidrule(r){2-5}
    $N_{r}$ & $L^{1}$ Error & Rate & $L^{\infty}$ Error & Rate \\
    \midrule
    8       & $1.54\times10^{-2}$ &       $-$ & $1.30\times10^{-1}$ &   $-$ \\
    16     & $3.98\times10^{-3}$ & $1.95$ & $4.47\times10^{-2}$ & 1.54 \\
    32     & $1.02\times10^{-3}$ & $1.96$ & $1.53\times10^{-2}$ & 1.54 \\ 
    64     & $2.62\times10^{-4}$ & $1.97$ & $4.73\times10^{-3}$ & 1.70 \\
    128   & $6.68\times10^{-5}$ & $1.97$ & $1.35\times10^{-3}$ & 1.81 \\ 
    256   & $1.69\times10^{-5}$ & $1.98$ & $3.66\times10^{-4}$ & 1.88 \\
    512   & $4.26\times10^{-6}$ & $1.99$ & $9.69\times10^{-5}$ & 1.92 \\
    1024 & $1.07\times10^{-6}$ & $1.99$ & $2.52\times10^{-5}$ & 1.94 \\
    2048 & $2.68\times10^{-7}$ & $2.00$ & $6.47\times10^{-6}$ & 1.96 \\
    \midrule
    & \multicolumn{4}{c}{DG(2)+RK3} \\
    \cmidrule(r){2-5}
    $N_{r}$ & $L^{1}$ Error & Rate & $L^{\infty}$ Error & Rate \\
    \midrule
    8       & $4.49\times10^{-4}$   &      $-$  & $4.02\times10^{-3}$ &   $-$ \\
    16     & $6.84\times10^{-5}$   & $2.72$ & $8.63\times10^{-4}$ & 2.22 \\
    32     & $9.83\times10^{-6}$   & $2.80$ & $1.41\times10^{-4}$ & 2.61 \\
    64     & $1.34\times10^{-6}$   & $2.88$ & $2.01\times10^{-5}$ & 2.81 \\
    128   & $1.81\times10^{-7}$   & $2.88$ & $2.87\times10^{-6}$ & 2.81 \\
    256   & $2.85\times10^{-8}$   & $2.67$ & $5.88\times10^{-7}$ & 2.29 \\
    512   & $3.91\times10^{-9}$   & $2.87$ & $1.59\times10^{-7}$ & 1.88 \\
    1024 & $5.01\times10^{-10}$ & $2.96$ & $2.80\times10^{-8}$ & 2.51 \\
    2048 & $6.41\times10^{-11}$ & $2.97$ & $4.94\times10^{-9}$ & 2.50 \\
    \midrule
  \end{tabular}
  \end{center}
\end{table}
These results confirm the expected order of accuracy for the different schemes.  

We have also computed some results where the computational domain extends to $r=0$.  
For practical purposes we set $r\ge1$ in the convergence study above to avoid small time steps; cf. Equation (\ref{eq:CFLsphericalSymmetry}).  
However, for many applications the origin must the included in the computational domain, even though this may introduce significant numerical errors.  
In particular, \citep[][]{machorro_2007} discuss inaccuracies near $r=0$ in the numerical solution to the transport equation in spherical symmetry, which appear in the form of a ``flux-dip."%
\footnote{In the context of finite volume methods for hydrodynamics, see \citep{monchmeyerMuller_1989,blondinLufkin_1993} for discussions on numerical errors associated with including the origin in spherical polar coordinates.} 
To investigate inaccuracies near $r=0$, we solve \eqref{eq:ConservativeBoltzmannEquationSphericalSymmetryFlat} using the models with $N_{r}=N_{\mu}=16$ in Table~\ref{tab:convergenceFreeStreamingExponential1D1D}, but with the computational domain given by $D=\{(r,\mu)\in\bbR^{2} : r\in[0,2],\,\mu\in[-1,1]\}$.  

\begin{figure}
  \centering
  \includegraphics[scale=0.2]{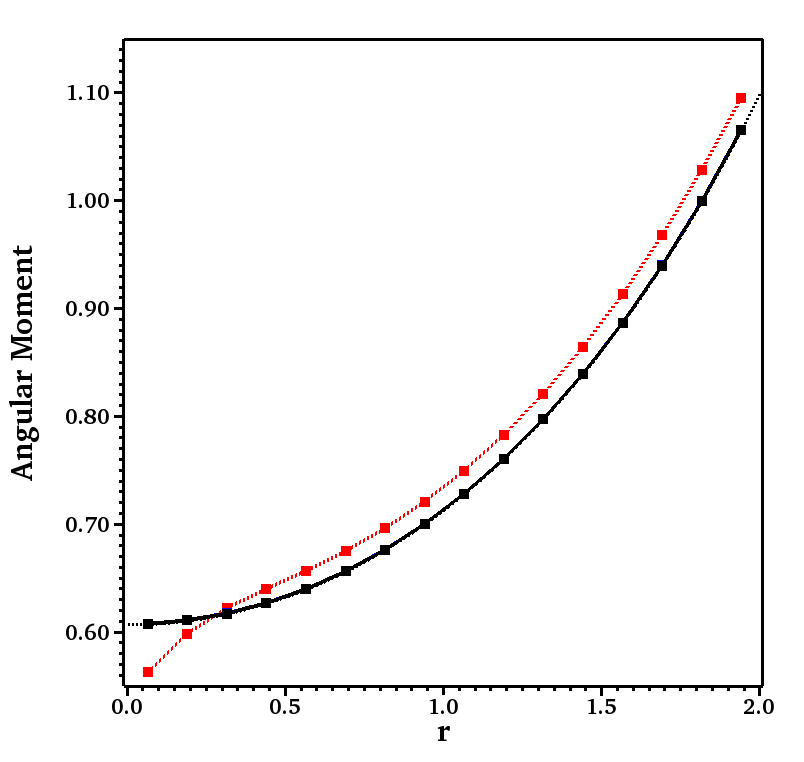}
  \caption{Plot of the zeroth angular moment of the distribution function in \eqref{eq:zerothMoment1D} versus radius at $t=0.5$ in a test where the computational domain extends down to $r=0$.  
  Results obtained with a $16\times16$ mesh are plotted for DG(0)+FE (dotted red), DG(1)+RK2 (dashed blue), DG(2)+RK3 (solid black).  
  The analytical solution is also plotted (dotted black line).}
  \label{fig:FreeStreamingExponential1D1D_AngularMoment}
\end{figure}

In Figure \ref{fig:FreeStreamingExponential1D1D_AngularMoment} we plot the ``zeroth" angular moment of the distribution function; i.e.,
\begin{equation}
  \rho(r,t^{n})=\f{1}{2}\int_{-1}^{1}\fDG^{n}(\mu,r)\,d\mu, 
  \label{eq:zerothMoment1D}
\end{equation}
versus radius at $t^{n}=0.5$.  
The results were obtained with DG(0)+FE (dotted red line), DG(1)+RK2 (dashed blue line), and DG(2)+RK3 (solid black line) using a $16\times16$ mesh.  
(The analytical solution is plotted with the dotted black line).  
The results obtained with DG(0)+FE appear to be offset by a constant factor from the analytical solution for $r\gtrsim0.4$.  
However, we observe a ``dip" in the numerical result inside $r\simeq0.4$, and the error is largest in the innermost cell.  
The results obtained with DG(1)+RK2 and DG(2)+RK3 are indistinguishable on the scale chosen for the plot, and follow the analytical solution well.  
Moreover, they do not show any sign of increased error near the origin.  
These results are consistent with those reported in \citep{machorro_2007}.  

\subsubsection{Radiating Sphere Test}
\label{subsubsec:radiating_sphere}

Next we include a test with discontinuous solutions.  
We consider a radiating sphere with radius $R_{0}=1$ centered at $r=0$ (see Figure \ref{fig:FreeStreaming1D1D_diagram}).  
(A version of this test was also considered in \citep{pons_etal_2000}; cf. their TEST~3.)
The sphere radiates steadily and isotropically at the surface --- which coincides with our inner boundary --- into a near vacuum ($f\ll1$).  
For $r>R_{0}$, once a steady state has been established in $D$, the distribution function becomes more and more ``forward-peaked" with increasing radius; i.e., its support is contained within the cone with opening angle $\tetradTheta_{\mbox{\tiny m}}$ which satisfies
\begin{equation}
  \cos\tetradTheta_{\mbox{\tiny m}}(r)=\mu_{\mbox{\tiny m}}(r)=\sqrt{1-\big(R_{0}/r\big)^{2}}.  
  \label{eq:muVersusRadius}
\end{equation}
As $r\to\infty$, the distribution function approaches a delta function in angle cosine, centered on $\mu=1$.  
We solve this problem on the computational domain $D=\{(r,\mu)\in\bbR^{2} : r\in[1,3],\,\mu\in[-1,1]\}$, and we initialize the test with an isotropic background by setting the distribution function to $\fDG=f_{0}=10^{-6}$ everywhere inside the domain.  
\begin{figure}
  \centering
  \includegraphics[scale=0.4]{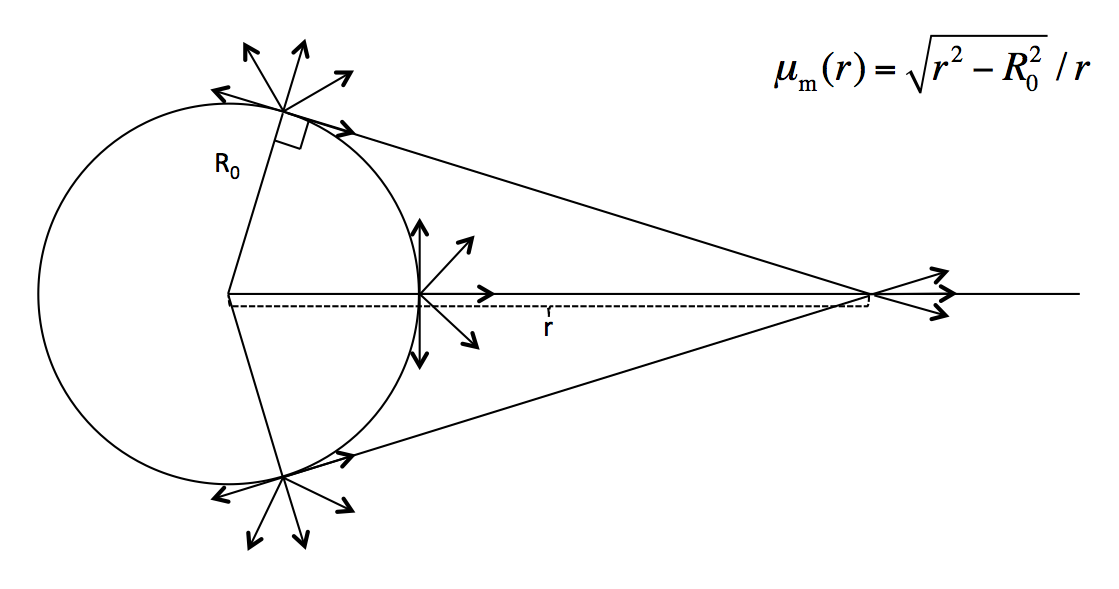}
  \caption{Geometry of the radiating sphere test.}
  \label{fig:FreeStreaming1D1D_diagram}
\end{figure}
We also keep $\fDG=f_{0}$ for the incoming radiation at the outer radial boundary, while at the inner radial boundary $r=R_{0}$, we set $\fDG=1$ for the incoming radiation. 
For $t>0$, a radiation front propagates through the domain.  
After a steady state is reached, the boundary defined by $\mu_{\mbox{\tiny m}}$ separates regions where $\fDG=1$ and $\fDG=f_{0}$ (dashed line in Figure~\ref{fig:FreeStreaming1D1D_pseudoColor}).  

\begin{figure}
  \centering
  \includegraphics[scale=0.3]{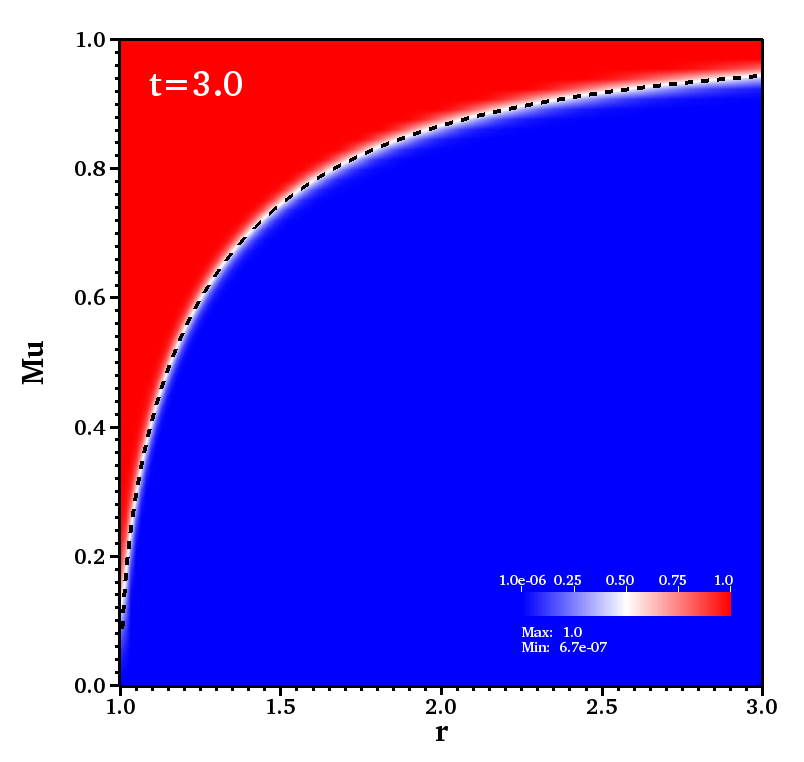}
  \caption{Color plot of the distribution function $\fDG(r,\mu)$ at $t=3.0$, obtained with the second-order scheme DG(1)+RK2 using $N_{r}\times N_{\mu}=128\times128$ cells.}
  \label{fig:FreeStreaming1D1D_pseudoColor}
\end{figure}
In Figure \ref{fig:FreeStreaming1D1D_pseudoColor} we plot the distribution function versus $r$ and $\mu$ at $t=3$.  
The numerical results were obtained with the second-order scheme DG(1)+RK2 using $128\times128$ cells.  
The DG method maintains the sharp boundary between the two regions, and $\fDG\in[0,1]$ over the entire computational domain.  
In the figure, we also plot $\mu_{\mbox{\tiny m}}$ versus $r$ (cf. Eq. (\ref{eq:muVersusRadius}); dashed line), which shows that the numerical result agrees well with the geometric considerations in Figure \ref{fig:FreeStreaming1D1D_diagram}.  

\begin{figure}
  \centering
  \begin{tabular}{cc}
    \includegraphics[scale=0.2]{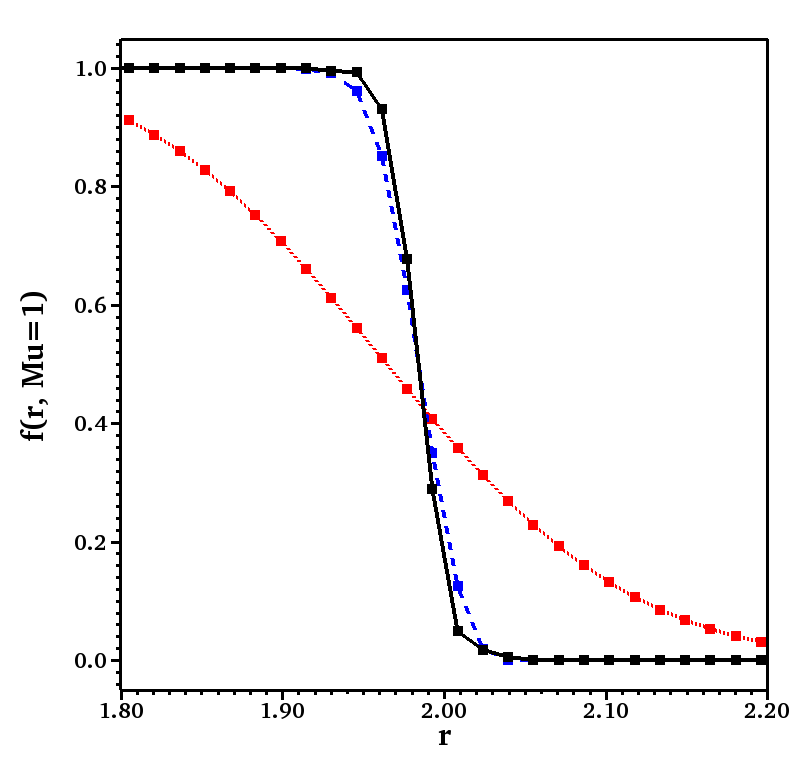} &
    \includegraphics[scale=0.2]{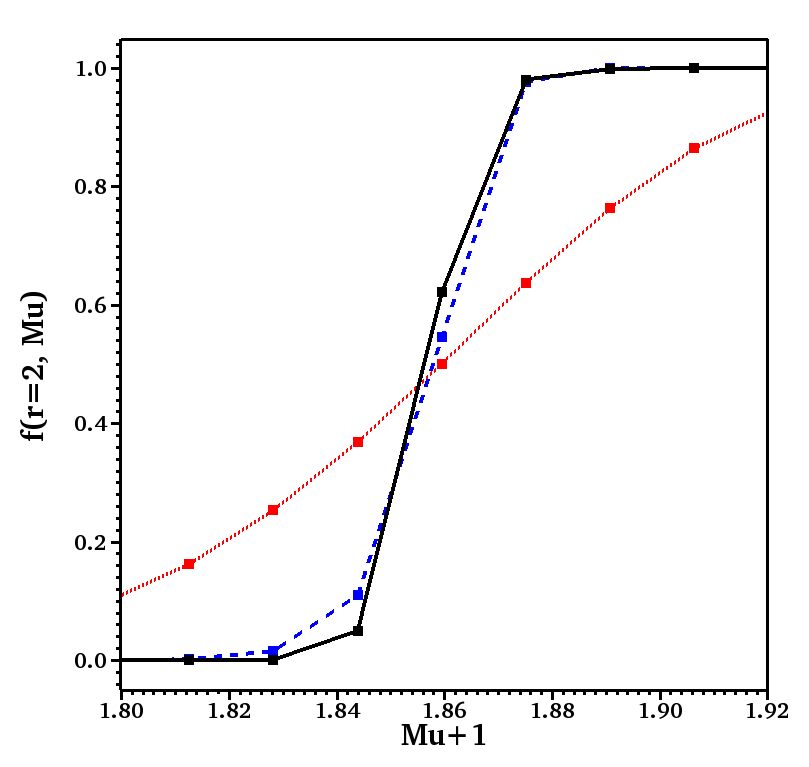}
  \end{tabular}
  \caption{Numerical results from the radiating sphere test (cf. Figure \ref{fig:FreeStreaming1D1D_diagram}) comparing the different schemes: DG(0)+FE (dotted red), DG(1)+RK2 (dashed blue), and DG(2)+RK3 (solid black).  
  In the left panel we plot the distribution function versus radius for constant $\mu=1$, at $t=1$; i.e., $\fDG(r,\mu=1,t=1)$.  
  In the right panel we plot the distribution function versus $\mu$ for constant radius $r=2$, at $t=3$; i.e., $\fDG(r=2,\mu,t=3)$.}
  \label{fig:FreeStreaming1D1D_compareOrderSchemes}
\end{figure}
We compare numerical results obtained with the first, second, and third order schemes in Figure \ref{fig:FreeStreaming1D1D_compareOrderSchemes}.  
In the left panel we plot the distribution function versus radius for $\mu=1$ at time $t=1$, when the radiation front is located at $r\approx2$.  
In the right panel we plot the distribution function versus $\mu$ for constant radius $r=2$ at time $t=3$, when a steady state configuration has been established in $D$.  
The first-order scheme is clearly very diffusive and unable to maintain the sharp edge.  
Both the second-order scheme and the third-order scheme capture the edge with only a few grid cells, with DG(2)+RK3 maintaining the sharpest edge.  

\begin{figure}
  \centering
  \begin{tabular}{cc}
    \includegraphics[scale=0.2]{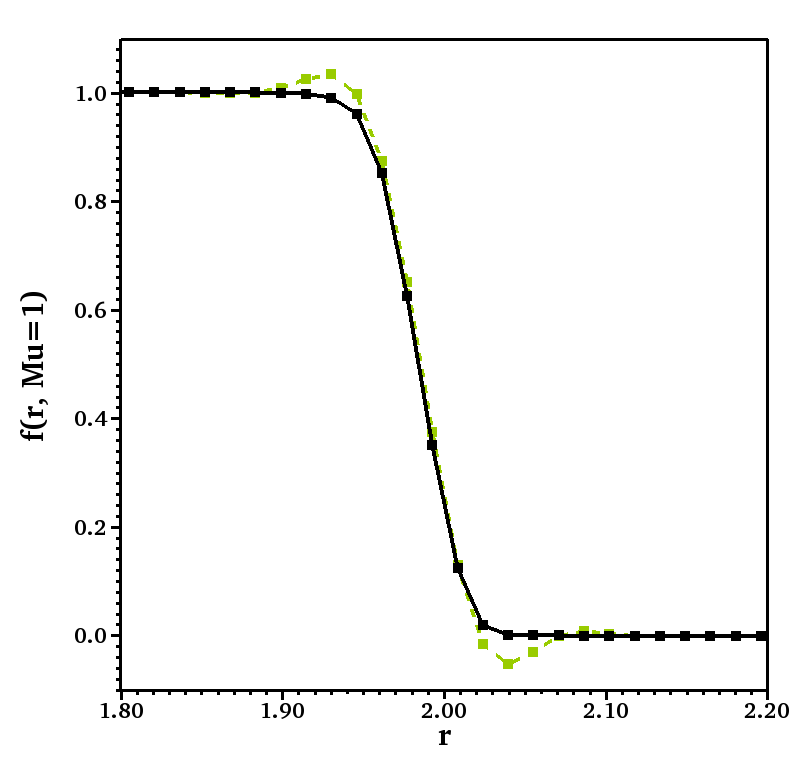} &
    \includegraphics[scale=0.2]{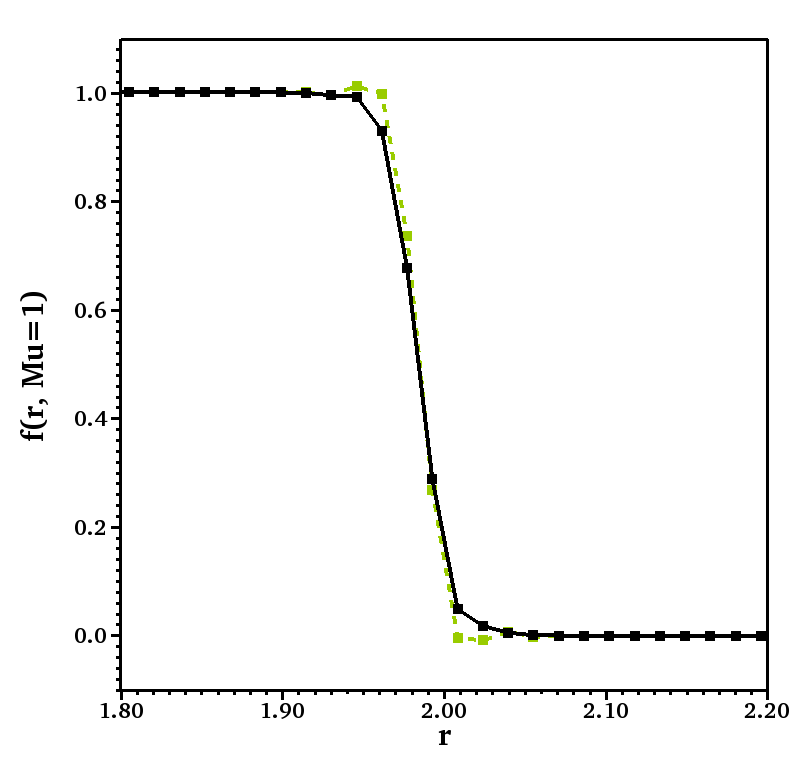}
  \end{tabular}
  \caption{Comparing numerical results obtained when running with (solid black) and without (dashed green) the bound-enforcing limiter.  
  We plot the distribution function versus radius for $\mu=1$ and $t=1$, obtained with the DG(1)+RK2 scheme (left panel) and the DG(2)+RK3 scheme (right panel).}
  \label{fig:FreeStreaming1D1D_LimiterVersusNoLimiter}
\end{figure}
We compare numerical results obtained when running with and without the bound-enforcing limiter (cf. Section \ref{sec:limiter}) in Figure \ref{fig:FreeStreaming1D1D_LimiterVersusNoLimiter}.  
(We use the CFL conditions in \eqref{eq:CFLsphericalSymmetry} for all the runs.)  
Without the limiter, the numerical results exhibit $\fDG<0$ ahead of the radiation front, and $\fDG>1$ behind the radiation front.  
These violations become less severe with the higher-order scheme (right panel).  
With the limiter on, $\fDG\in[0,1]$ for all times.  

\subsection{1D~x+2D~p}
\label{sec:numericalExamplesSphericalSymmetryGR}

In this section we present results obtained by solving the general relativistic phase space advection problem in spherical symmetry as modeled by Equation~\eqref{eq:ConservativeBoltzmannEquationSphericalSymmetryGR}.  
We adopt the Schwarzschild metric (i.e., Eq.~\eqref{eq:metricDiagonalCFC}, with $\alpha$ and $\psi$ given in \eqref{eq:schwarzschildMetric}), and compute results for various spacetime masses $M$.  
For reference, in Figure~\ref{fig:lapseFunctionAndConformalFactor}, we plot the lapse function $\alpha$ (solid lines) and the conformal factor $\psi$ (dashed lines) for $r\in[1,3]$ and $M=0.0, 0.2, 4-2\sqrt{3}$, and $2/3$.  
The Schwarzschild radius $r_{\mbox{\tiny S}}=M/2$ is well inside the inner boundary for all models.  
For $M=0.0$, we have $\alpha=\psi=1$, and Equation~\eqref{eq:ConservativeBoltzmannEquationSphericalSymmetryGR} reduces to the flat spacetime case in \eqref{eq:ConservativeBoltzmannEquationSphericalSymmetryFlat}.  
For $M>0$, we have $\p_{r}\ln\alpha=(M/r^{2})\,(1-(M/2r)^2)^{-1}$, $\p_{r}\ln\psi^{2}=-(M/r^{2})\,(1+M/2r)^{-1}$, so that
\begin{equation}
  \Psi=1-\f{M}{\psi\,r}\Big(\,1+\f{1}{\psi\,\alpha}\,\Big).  
\end{equation}

\subsubsection{Radiating Sphere Test in Schwarzschild Geometry}

The test we consider is an extension of the radiating sphere test in Section~\ref{subsubsec:radiating_sphere} 
However, at the inner radial boundary ($r=1$) we also specify an energy spectrum (Gaussian or Fermi-Dirac) for the isotropic radiation entering the computational domain $D=\{(r,\mu,\tetradEpsilon)\in\bbR^{3} : r\in[1,3],\,\mu\in[-1,1],\,\tetradEpsilon\in[0,1]\}$.  
Since $\p_{r}\ln\alpha>0$, the energy spectrum of radiation propagating out of the gravitational well ($\mu>0$) will be redshifted (cf. the energy derivative term in Equation~\eqref{eq:ConservativeBoltzmannEquationSphericalSymmetryGR}).  
We also expect gravitational corrections to the angular aberration (cf. the angle derivative term in Equation~\eqref{eq:ConservativeBoltzmannEquationSphericalSymmetryGR}).  
In particular, for $M=4-2\sqrt{3}$ we have $\Psi=0$ at $r=1$.  
For larger $M$, $\Psi<0$ near $r=1$ for $\mu>0$, and we expect some of the radiation entering the computational domain at the inner radial boundary to be ``bent inward" and exit the computational domain through the inner radial boundary (cf. the model with $M=2/3$).  

First we consider a Gaussian spectrum for the radiation entering $D$; i.e.,
\begin{equation*}
  \fDG(r=1,\mu,\tetradEpsilon)
  =\exp\big\{\,-100\,\big(0.5-\tetradEpsilon\big)^{2}\,\big\}\quad\text{ for } \mu\ge0.  
\end{equation*}
Initially, the distribution function is set to zero everywhere in the computational domain.  
We use the appropriate bound-preserving CFL conditions in \eqref{eq:CFLSphericalSymmetryRadius}-\eqref{eq:CFLSphericalSymmetryEnergy} with $s_{1}=s_{2}=s_{3}=1/3$, and the phase space resolution is $N_{r}\times N_{\mu}\times N_{\tetradEpsilon}=128\times128\times64$.  
We run the simulations until a steady state in $D$ is reached ($t\approx3$ for $M=0.0$ and $t\approx20$ for $M=2/3$).  
The numerical results are plotted in Figures \ref{fig:energySpectra}-\ref{fig:limiterVsNoLimiter_M0_667}.  

\begin{figure}
  \centering
  \includegraphics[scale=0.25]{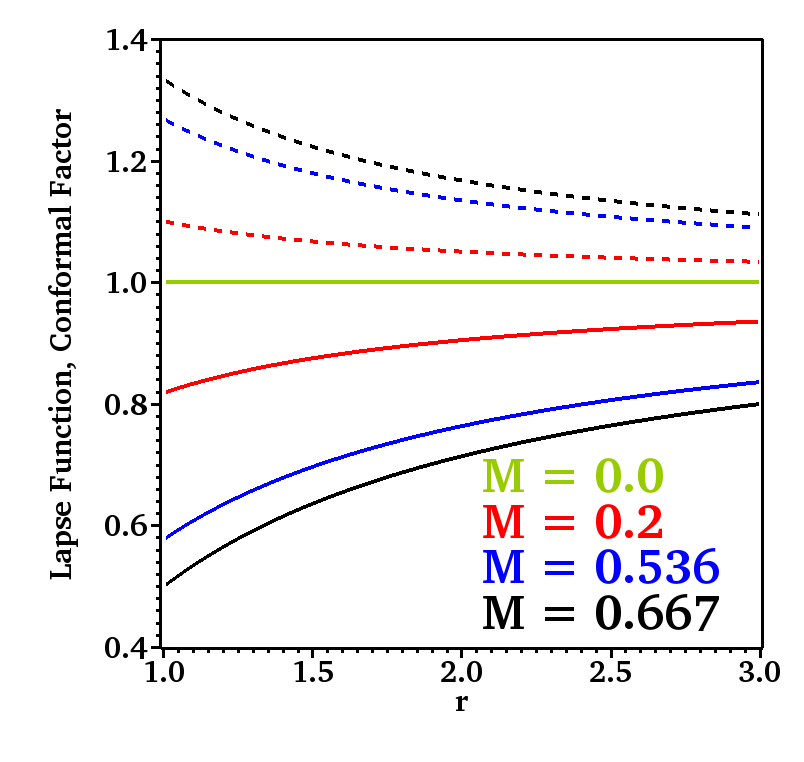} 
  \caption{Plot of the lapse function (solid lines) and the conformal factor (dashed lines) for the Schwarzchild metric (cf. Equation (\ref{eq:schwarzschildMetric})) for various spacetime masses: $M=0.0$ (green), $M=0.2$ (red), $M=4-2\sqrt{3}$ (blue), and $M=2/3$ (black).}
  \label{fig:lapseFunctionAndConformalFactor}
\end{figure}

In the left panel in Figure \ref{fig:energySpectra}, we plot energy spectra at the outer radial boundary (for the angle $\mu=1$) for the model with $M=2/3$.  
Results for the various schemes are plotted; i.e., DG(0)+FE (solid red line), DG(1)+RK2 (solid blue line), and DG(2)+RK3 (solid black line).  
For reference, the spectrum at the inner radial boundary is also plotted (dashed line) --- illustrating the gravitational redshift as the radiation propagates out of the gravitational well.  
As expected, the first-order scheme is more diffusive than the second and third order schemes, while the second and third order schemes are indistinguishable on this plot.  
At the outer radial boundary, we find that the peak of the spectrum has shifted from $\tetradEpsilon=0.5$ to about $\tetradEpsilon=0.3$.  
(Since $\alpha\,\tetradEpsilon=\text{const}.$, $0.5\times\alpha(r=1)/\alpha(r=3)=0.3125$ is expected for $M=2/3$.)  
We also note that the the widths of the spectra have decreased slightly at $r=3$.  
In the right panel of Figure \ref{fig:energySpectra}, we plot energy spectra for various masses $M$, obtained with the second-order scheme (DG(1)+RK2).  
The spectra become increasingly ``redshifted" (i.e., shifted to lower energies) as the mass increases.  
The spectral width also decreases with increasing mass $M$.  
At $r=3$, the width of the spectrum for the model with $M=2/3$ is almost halved when compared with the $M=0.0$ model.  
The effective resolution of the various energy spectra decreases as a result of the decreased spectral width.  
Moreover, the lower effective resolution results in a slight decrease in the spectral peak with increasing $M$.  
For the model with $M=2/3$, we have found that the third-order scheme performs slightly better (i.e., maintains a higher peak) than the second-order scheme.  

\begin{figure}
  \centering
  \begin{tabular}{cc}
    \includegraphics[scale=0.23]{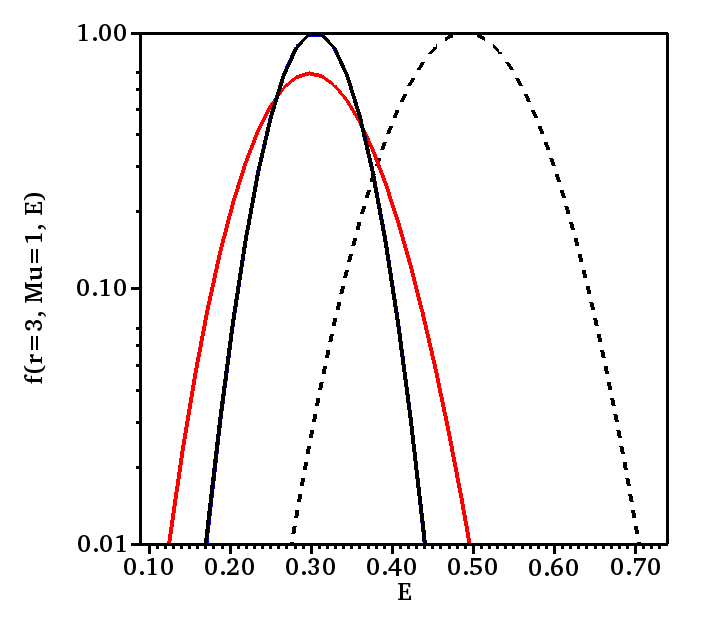} &
    \includegraphics[scale=0.23]{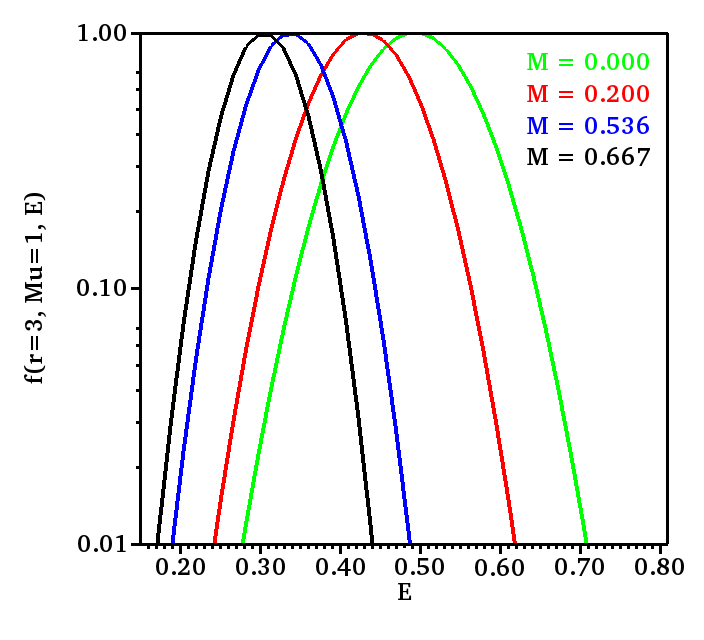} \\
  \end{tabular}
  \caption{Redshifted energy spectra at $r=3$ for various general relativistic computations.  
  Left panel: DG(0)+FE (solid red line), DG(1)+RK2 (solid blue line), and DG(2)+RK3 (solid black line) for mass $M=2/3$.  
  (The ``emitted" spectrum at $r=1$ is also included; dashed black line.)  
  Right panel: results obtained with the second order scheme (DG(1)+RK2) for various masses; $M=0.0$ (green), $M=0.2$ (red), $M=4-2\sqrt{3}$ (blue), and $M=2/3$ (black).}
  \label{fig:energySpectra}
\end{figure}

Figure \ref{fig:distributionsVsRadiusMu} provides a different perspective on the computed models, with color plots of the distribution function versus radius $r$ and angle $\mu$ for a constant energy $\tetradEpsilon$.  
Results are shown after a steady state is reached (similar to Figure \ref{fig:FreeStreaming1D1D_pseudoColor}).  
The results from the $M=0.0$ model for $\tetradEpsilon=0.5$, which correspond to the model in Figure \ref{fig:FreeStreaming1D1D_pseudoColor}, are shown in the upper left panel.  
In particular, the distribution is uniform in radius and angle in two regions, separated by the dashed line predicted by the geometric considerations in Figure \ref{fig:FreeStreaming1D1D_diagram}.  
The effects of gravitational redshift and aberration are visible in the model with $M=0.2$, which is shown in the upper right panel (also for $\tetradEpsilon=0.5$).  
Aberration results in a slightly less forward-peaked distribution function at $r=3$, while the redshift causes a reduction in the amplitude of the distribution near the outer boundary for this particular energy bin.  
The two lower panels show results from the $M=2/3$ model, for energies $\tetradEpsilon=0.5$ (left) and $\tetradEpsilon=0.3$ (right), which exhibits more extreme gravitational effects.  
First, as is also seen in Figure \ref{fig:energySpectra}, the gravitational redshift causes the peak of the distribution to shift from $\tetradEpsilon=0.5$ at $r=1$ to about $\tetradEpsilon=0.3$ at $r=3$.  
Second, at the outer boundary, the distribution function is significantly less forward-peaked than it is in the other models.  
Third, some of the radiation that enters the computational domain at $r=1$ ($\mu\ge0$), exits the computational domain through the inner radial boundary; i.e., $\fDG(r=1,\mu<0,\tetradEpsilon=0.5)>0$.  

\begin{figure}
  \centering
  \begin{tabular}{cc}
    \includegraphics[scale=0.235]{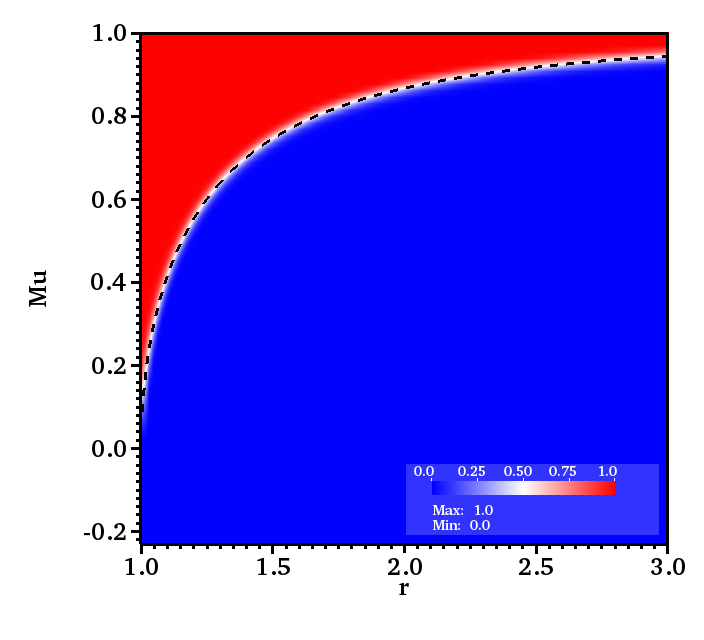} &
    \includegraphics[scale=0.235]{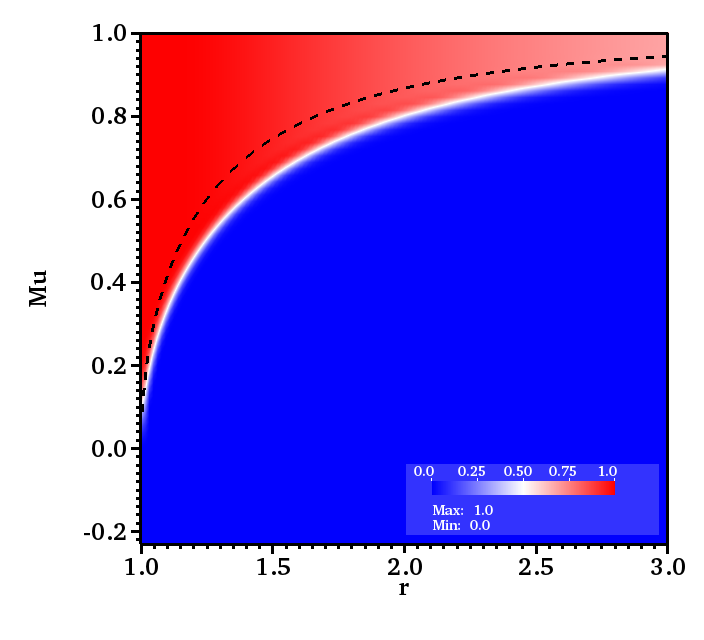} \\
    \includegraphics[scale=0.235]{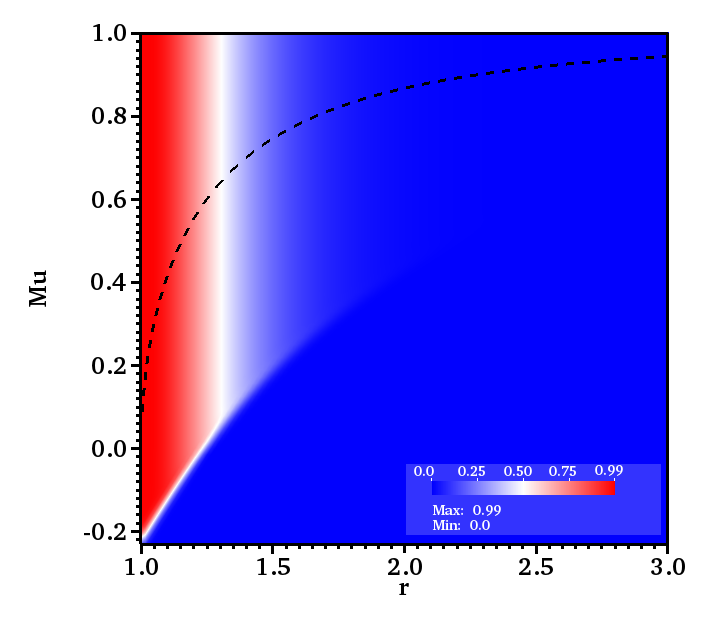} &
    \includegraphics[scale=0.235]{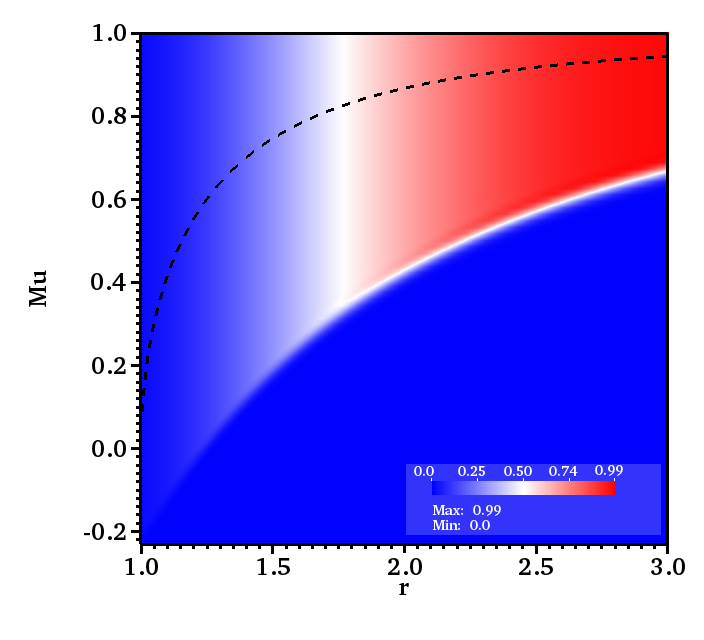} \\
  \end{tabular}
  \caption{Color plots of the distribution function versus radius $r$ and angle $\mu$ for a constant energy $\tetradEpsilon$, computed with the DG(1)+RK2 scheme for various spacetime masses $M$.  
  Selected energy bins are shown:  $M=0.0, \tetradEpsilon=0.5$ (upper left); $M=0.2, \tetradEpsilon=0.5$ (upper right); $M=2/3, \tetradEpsilon=0.5$ (lower left); and $M=2/3, \tetradEpsilon=0.3$ (lower right).}
  \label{fig:distributionsVsRadiusMu}
\end{figure}

In Figure \ref{fig:limiterVsNoLimiter_M0_667} we demonstrate the effect of using the bound-enforcing limiter and the appropriate CLF condition for the model with $M=2/3$.  
We plot the number density 
\begin{equation}
  \cN(r,t^{n})
  =\int_{0}^{1}\int_{-1}^{1}\fDG^{n}(r,\mu,\tetradEpsilon)\,d\mu\,\tetradEpsilon^{2}\,d\tetradEpsilon
\end{equation}
and energy density
\begin{equation}
  \cE(r,t^{n})
  =\int_{0}^{1}\int_{-1}^{1}\fDG^{n}(r,\mu,\tetradEpsilon)\,d\mu\,\tetradEpsilon^{3}\,d\tetradEpsilon
\end{equation}
versus radius (red and black curves, respectively) at time $t=1$, computed using the DG(1)+RK2 scheme, with the limiter (solid lines) and without it (dashed lines).  
The inset illustrates the effect of the limiter, which prevents $\cN$ and $\cE$ from becoming negative.  
Without the limiter both $\cN$ and $\cE$ take on negative values in some places.

\begin{figure}
  \centering
  \includegraphics[scale=0.3]{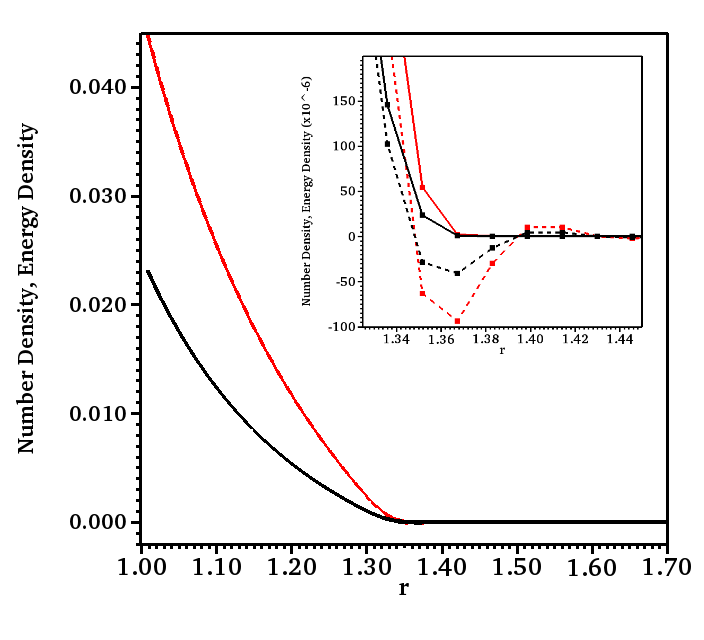} 
  \caption{Number density (red lines) and energy density (black lines) at time $t=1$ for the general relativistic model with $M=2/3$, computed with the DG(1)+RK2 scheme.  
  The model  was computed with and without the bound-enforcing limiter (solid and dashed lines, respectively).  
  The inset is a zoomed-in view to highlight the effect of the limiter.}
  \label{fig:limiterVsNoLimiter_M0_667}
\end{figure}

Finally, we have computed some additional models where we specify a Fermi-Dirac spectrum for the incoming radiation at the inner radial boundary; i.e., 
\begin{equation*}
  \fDG(r=1,\mu,\tetradEpsilon)
  =\big[\,\exp\big\{\,100\,\big(\tetradEpsilon-0.5\big)\,\big\}+1\,\big]^{-1} \quad\text{ for } \mu\ge0.  
\end{equation*}
``Fermi-blocking" plays an important role during the collapse phase of core-collapse supernovae \citep[e.g.,][]{lentz_etal_2012b}, where the neutrino phase space occupation increases with increasing core density due to electron capture on nuclei.  
This process fills up the low-energy portion of the spectrum and prohibits down-scattering of higher energy neutrinos.  
It is important to maintain $\fDG\le1$ during the advection part of the algorithm.  
In Figure \ref{fig:energySpectraFermiDirac} we plot energy spectra at the outer radial boundary for $\mu=1$, for the model with $M=2/3$.  
Results for various schemes are plotted; i.e., DG(0)+FE (solid red line), DG(1)+RK2 (solid blue line), and DG(2)+RK3 (solid black line).  
The spectrum at the inner boundary is also plotted (dashed line).  
As with the results displayed in Figure \ref{fig:energySpectra}, the energy spectra are significantly redshifted at the outer radial boundary.  
The first-order scheme is very diffusive when compared to the second and third order schemes, while the results obtained with the second and third order schemes are similar, and differ only in the high-energy tail.  
All schemes maintain positivity of $\fDG$ and $1-\fDG$.  

\begin{figure}
  \centering
  \includegraphics[scale=0.35]{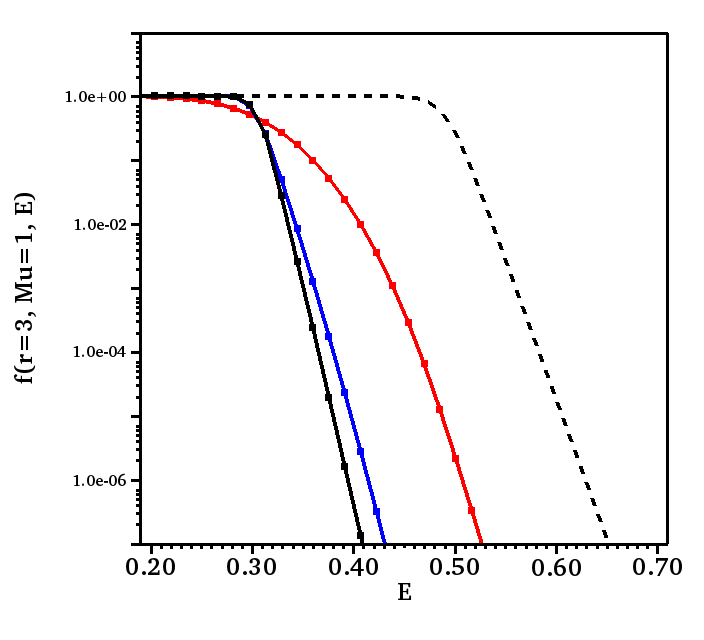} 
  \caption{Energy spectra at $r=3$ for computations with $M=2/3$ and a Fermi-Dirac spectrum specified at the inner radial boundary: DG(0)+FE (solid red line), DG(1)+RK2 (solid blue line), and DG(2)+RK3 (solid black line).  
  The spectrum at $r=1$ is also plotted (dashed black line).}
  \label{fig:energySpectraFermiDirac}
\end{figure}

To further demonstrate the effectiveness of our bound-preserving DG scheme, in Figure \ref{fig:limiterVsNoLimiterFermiDirac} we compare the spectrum at $r=3$ from one model computed with the bound-enforcing limiter on (solid lines) with the spectrum from one model computed without the limiter (dashed lines).  
Both models were computed with the DG(1)+RK2 scheme with $M=2/3$.  
As can be seen, without the limiter, the distribution function overshoots unity (left panel) and becomes negative (right panel).  
For the bound-preserving scheme, we have $0\le\fDG\le1$ at all times.  

\begin{figure}
  \centering
  \begin{tabular}{cc}
    \includegraphics[scale=0.21]{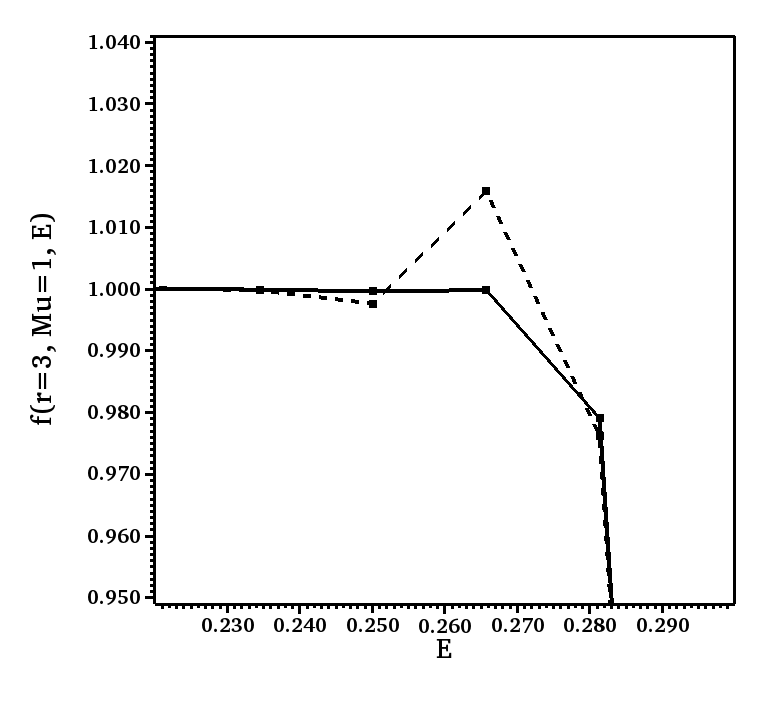} &
    \includegraphics[scale=0.21]{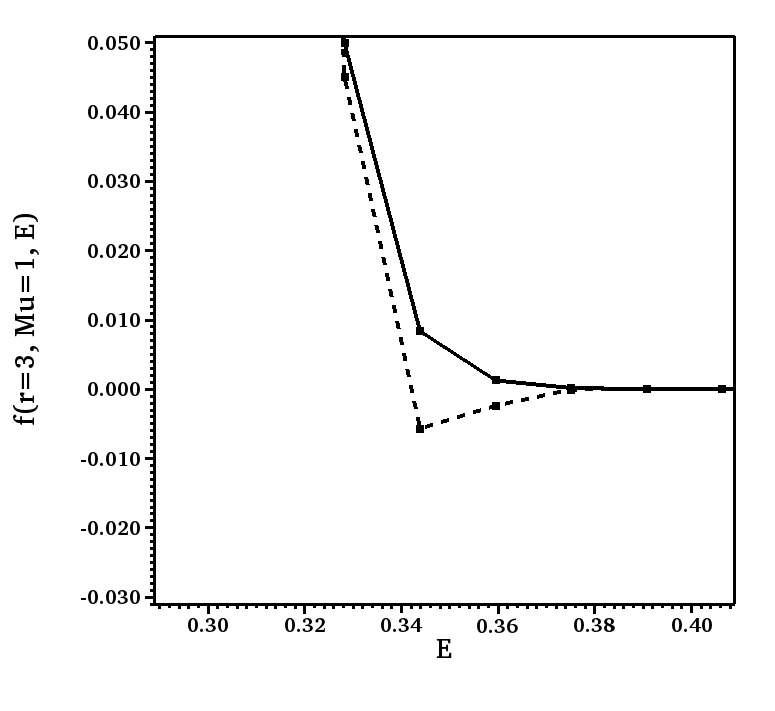} \\
  \end{tabular}
  \caption{Comparison of spectra at $r=3$ near the ``Fermi surface" obtained with the DG(1)+RK2 scheme with $M=2/3$, computed with (solid) and without (dashed) the bound-enforcing limiter.}
  \label{fig:limiterVsNoLimiterFermiDirac}
\end{figure}

\subsection{2D~x+2D~p}
\label{sec:numericalExamplesAxialSymmetry}

The tests in this section involve the axially symmetric phase space for flat spacetimes.  
That is, we employ cylindrical spatial coordinates and spherical momentum coordinates, and solve Equation (\ref{eq:ConservativeBoltzmannEquationAxialSymmetryFlat}) for $\fDG^{n}(\rPerp,z,\mu,\tetradPhi)$ at discrete time levels $t^{n}$.  

\subsubsection{Test with Smooth Analytical Solution}

First we consider a test problem with smooth solutions.  
An analytical solution to Equation (\ref{eq:ConservativeBoltzmannEquationAxialSymmetryFlat}) is given by
\begin{equation}
  f(\rPerp,z,\mu,\tetradPhi,t)=\exp\big(\,\sqrt{1-\mu^{2}}\,\cos\tetradPhi\,\rPerp+\mu\,z-t\,\big).
\end{equation}
This function is not sufficiently smooth at $\mu=\pm1$ to demonstrate high-order accuracy.  
Thus, we reduce angular extent in the $\mu$-direction, and take the computational domain to be given by $D=\{(\rPerp,z,\mu,\tetradPhi)\in\bbR^{4} : \rPerp\in[1,2],\,z\in[-0.5,0.5],\,\mu\in[-0.5,0.5],\,\tetradPhi\in[0,\pi]\}$.  
We evolve from $t=0$ to $t=0.1$, and use the analytical solution to set the boundary conditions for the incoming radiation.  
We use the bound-preserving CFL conditions given in Equation (\ref{eq:CFLaxialSymmetry}) with $s_{1}=s_{2}=s_{3}=1/3$.  
We note again that this test is of purely academic interest and is included to measure the accuracy and convergence rate of the DG schemes in the axially symmetric case.  

To evaluate the accuracy and the convergence rates, we evaluate the $L^{1}$-error norm
\begin{equation}
  E^{1}=\f{1}{V_{D}}
  \sum_{\bK\in\cT}\int_{\bK}|\fDG^{n}(\rPerp,z,\mu,\tetradPhi)-f(\rPerp,z,\mu,\tetradPhi,t^{n})|\,dV
  \label{eq:errorL1_2D2D}
\end{equation}
at $t=0.1$ for various grid resolutions.  
(The integral in \eqref{eq:errorL1_2D2D} is computed with $3$-point Gaussian quadratures in all dimensions, and $V_{D}=8\pi$.)  
Each resolution satisfies $N_{\rPerp}=N_{z}=N_{\mu}=\f{1}{3}\,N_{\tetradPhi}$ (i.e., approximately ``square" phase space cells).  
Results obtained with the DG(0)+FE, DG(1)+RK2, and DG(2)+RK3 schemes are listed in Table \ref{tab:convergenceFreeStreamingExponential2D2D}.  
\begin{table}
  \begin{center}
  \caption{$L^{1}$, $L^{\infty}$ error norms and convergence rates for the smooth 2D~x+2D~p test.  
  \label{tab:convergenceFreeStreamingExponential2D2D}}
  \begin{tabular}{ccccc}
    \midrule
    & \multicolumn{4}{c}{DG(0)+FE} \\
    \cmidrule(r){2-5}  
    $N$ & $L^{1}$ Error & Rate & $L^{\infty}$ Error & Rate \\
    \midrule
    $4^{3}\times12$   & $1.23\times10^{-1}$ &       $-$ & $1.20$                      &   $-$ \\  
    $8^{3}\times24$   & $6.16\times10^{-2}$ & $0.99$ & $6.48\times10^{-1}$ & 0.89 \\  
    $16^{3}\times48$ & $3.09\times10^{-2}$ & $1.00$ & $3.49\times10^{-1}$ & 0.89 \\ 
    $32^{3}\times96$ & $1.55\times10^{-2}$ & $1.00$ & $1.84\times10^{-1}$ & 0.93 \\  
    \midrule
    & \multicolumn{4}{c}{DG(1)+RK2} \\
    \cmidrule(r){2-5}
    $N$ & $L^{1}$ Error & Rate & $L^{\infty}$ Error & Rate \\
    \midrule
    $4^{3}\times12$   & $7.09\times10^{-3}$ &       $-$ & $8.32\times10^{-2}$ &   $-$ \\
    $8^{3}\times24$   & $1.92\times10^{-3}$ & $1.88$ & $2.55\times10^{-2}$ & 1.71 \\
    $16^{3}\times48$ & $5.00\times10^{-4}$ & $1.94$ & $6.97\times10^{-3}$ & 1.87 \\ 
    $32^{3}\times96$ & $1.28\times10^{-4}$ & $1.97$ & $1.82\times10^{-3}$ & 1.94 \\
    \midrule
    & \multicolumn{4}{c}{DG(2)+RK3} \\
    \cmidrule(r){2-5}
    $N$ & $L^{1}$ Error & Rate & $L^{\infty}$ Error & Rate \\
    \midrule
    $4^{3}\times12$   & $1.50\times10^{-4}$   &      $-$  & $1.45\times10^{-3}$ &   $-$ \\
    $8^{3}\times24$   & $2.72\times10^{-5}$   & $2.46$ & $2.94\times10^{-4}$ & 2.30 \\
    $16^{3}\times48$ & $3.82\times10^{-6}$   & $2.83$ & $4.17\times10^{-5}$ & 2.82 \\
    $32^{3}\times96$ & $4.67\times10^{-7}$   & $3.03$ & $4.63\times10^{-6}$ & 3.17 \\
    \midrule
  \end{tabular}
  \end{center}
\end{table}
The numerical results confirm the expected order of accuracy for the different schemes (first, second and third order, respectively).  
For this test, for a given phase space resolution, the additional cost (i.e., increased memory footprint) of the higher-order DG schemes is offset by higher accuracy.  
For example, with a resolution of $16^{3}\times48$, the $L^{1}$-error norm obtained with the second-order scheme ($2^{4}$ degrees of freedom per cell) is $5\times10^{-4}$, while the $L^{1}$ error norm obtained with the third-order scheme ($3^{4}$ degrees of freedom per cell) is reduced by more than two orders of magnitude, to $\sim3.8\times10^{-6}$.  
Moreover, the $L^{1}$ error norm obtained with DG(2)+RK3 using $4^{3}\times12$ cells is of the same order of magnitude as the $L^{1}$ error norm obtained with DG(1)+RK2 using $16^{3}\times48$ cells, but with a factor of 50 reduction in total memory cost to store the distribution function.  

As in the spherically symmetric case, we have computed results for models extending to the symmetry axis, $\rPerp=0$.  
The results are similar to those displayed in Figure \ref{fig:FreeStreamingExponential1D1D_AngularMoment}.  
For the first-order scheme, we observe the ``dip" in the numerical result near $\rPerp=0$, with the largest error in the cell with $\RL=0$.  
The results obtained with DG(1)+RK2 and DG(2)+RK3 do not show such signs of increased error near the $z$-axis.  

\subsubsection{Two-Beam Test}

As a second test we consider two beams with Gaussian shape entering the computational domain at the inner boundary ($\rPerp=\rPerp_{0}=1$); i.e., we set
\begin{eqnarray}
  \fDG\big(\rPerp_{0},z,\mu,\tetradPhi\big)
  &=&
  \exp
  \Big\{
    -\big(z_{1}-z\big)^{2}/L_{z}^{2}
    -\big(\mu_{1}-\mu\big)^{2}/L_{\mu}^{2}
    -\tetradPhi^{2}/L_{\tetradPhi}^{2}
  \Big\} \nonumber \\
  &&
  +\exp
  \Big\{
    -\big(z_{2}-z\big)^{2}/L_{z}^{2}
    -\big(\mu_{2}-\mu\big)^{2}/L_{\mu}^{2}
    -\tetradPhi^{2}/L_{\tetradPhi}^{2}
  \Big\}, 
\end{eqnarray}
with $z_{1}=-49/64$, $z_{2}=39/64$, $\mu_{1}=9/16$, $\mu_{2}=-11/16$, and $L_{z}=L_{\mu}=L_{\tetradPhi}=0.1$.  
Initially, the distribution function is set to $10^{-6}$ in the computational domain, which is given by $D=\{(\rPerp,z,\mu,\tetradPhi)\in\bbR^{4} : \rPerp\in[1,3],\,z\in[-1,1],\,\mu\in[-1,1],\,\tetradPhi\in[0,\pi]\}$.  
We evolve until $t=2.6$, when a steady state is reached.  
We use the positivity-preserving CFL conditions in Equation (\ref{eq:CFLaxialSymmetry}) with $s_{1}=s_{2}=s_{3}=1/3$.  
(To save computational time, we run with a single energy group with $\tetradEpsilon\in[0,1]$.)  
This test is also relevant to core-collapse supernova simulations as ``beams" of neutrino radiation may emanate from localized hotspots on the surface of the proto-neutron star \citep{bruenn_etal_2014}.  

First we compare results obtained with the various schemes, using various resolutions (denoted by ${N_{\rPerp}}\times{N_{z}}\times{N_{\mu}}\times{N_{\tetradPhi}}$).  
In Figure \ref{fig:TwoBeamTest} we display the angular moment of the distribution function,
\begin{equation}
  \mathcal{E}(\rPerp,z,t^{n})
  =\f{\langle\tetradEpsilon\rangle}{2\,\pi}\int_{0}^{\pi}\int_{-1}^{1} \fDG^{n}(\rPerp,z,\mu,\tetradPhi)\,d\mu\,d\tetradPhi,
  \label{eq:angularMomentAxialSymmetry}
\end{equation}
versus radius $\rPerp$ and distance along the symmetry axis $z$, at $t=2.6$.  
In \eqref{eq:angularMomentAxialSymmetry}, $\langle\tetradEpsilon\rangle=3/4$.  
In the two upper panels we plot results obtained with DG(0)+FE using $64^{2}\times24\times36$ cells (upper left) and $256^{2}\times96\times144$ cells (upper right).  
In the two middle panels we display results obtained with DG(1)+RK2 ($64^{2}\times24\times36$; middle left) and DG(2)+RK3 ($64^{2}\times24\times36$; middle right).  
In the two lower panels we plot results where we have increased the spatial resolution by a factor of two in each dimension for DG(1)+RK2 ($128^{2}\times24\times36$; lower left) and DG(2)+RK3 ($64^{2}\times24\times36$; lower right).  
\begin{figure}
  \centering
  \begin{tabular}{cc}
    \includegraphics[scale=0.205]{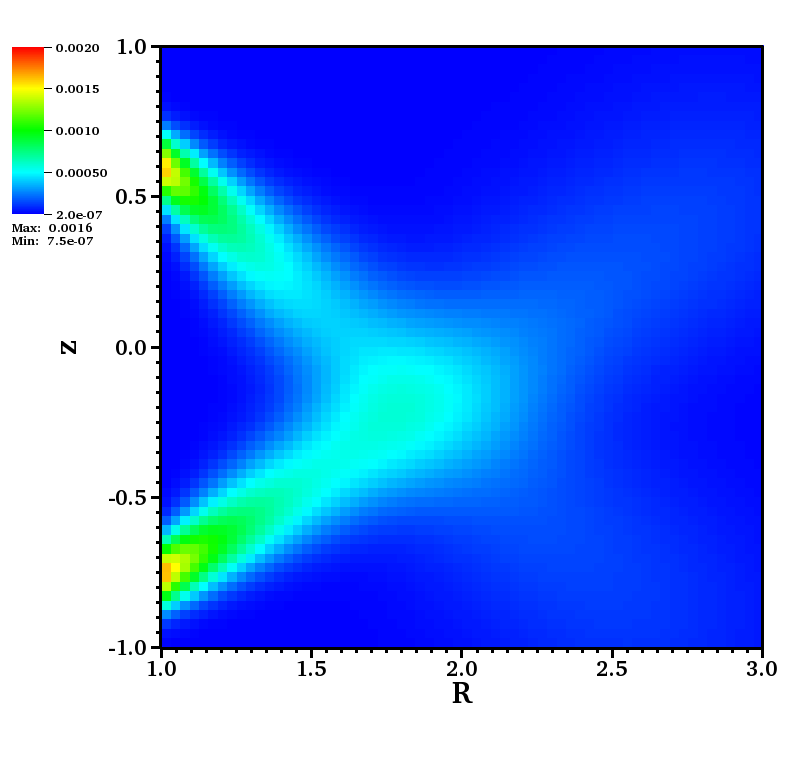} &
    \includegraphics[scale=0.205]{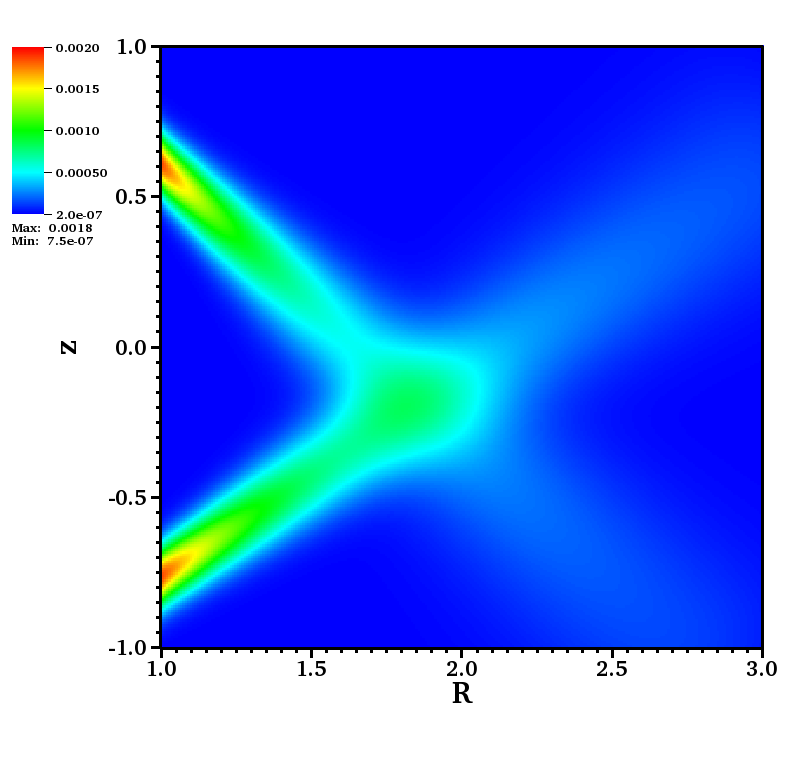} \\
    \includegraphics[scale=0.205]{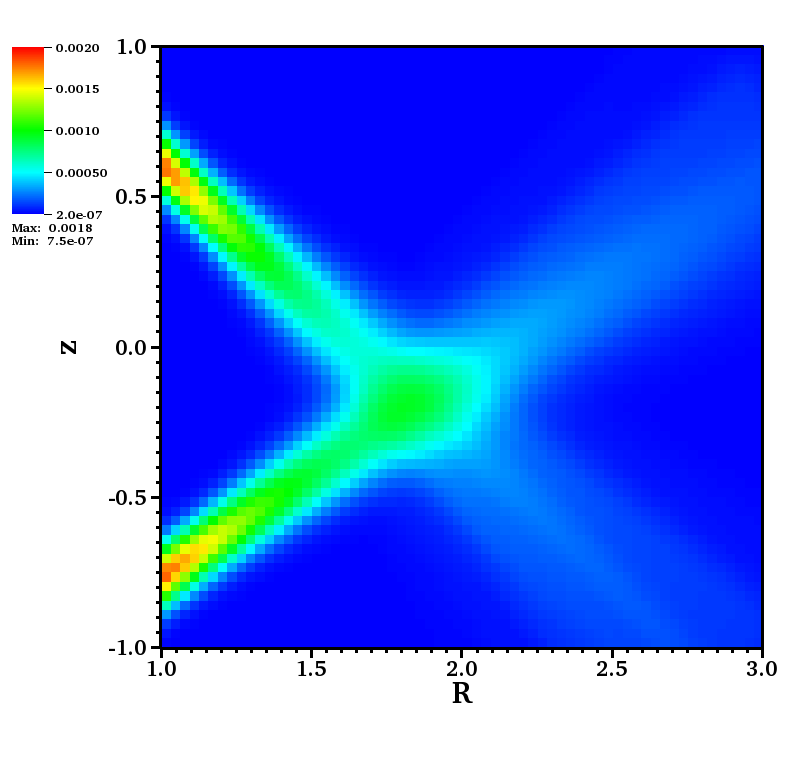} &
    \includegraphics[scale=0.205]{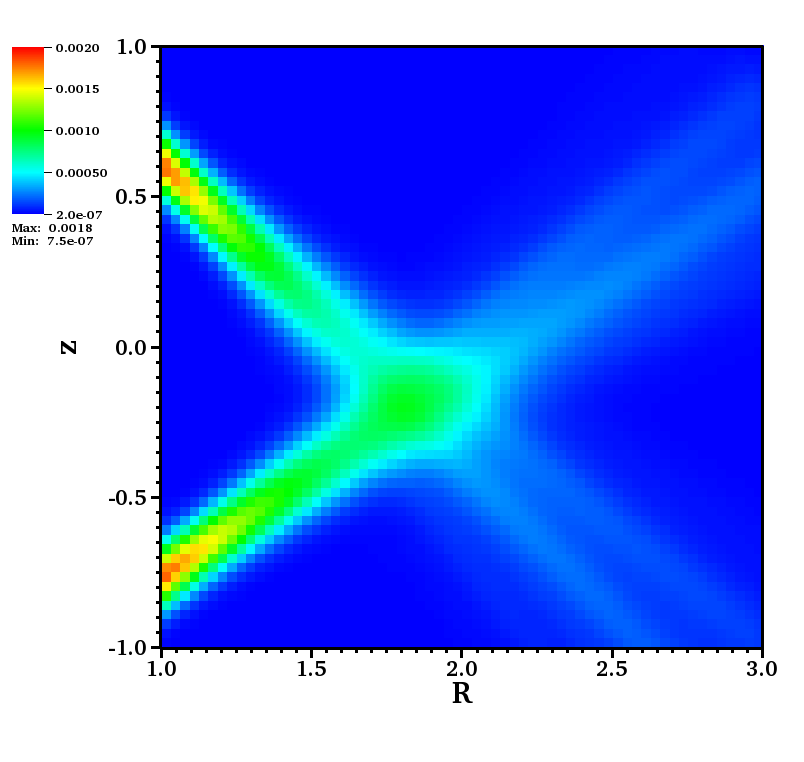} \\
    \includegraphics[scale=0.205]{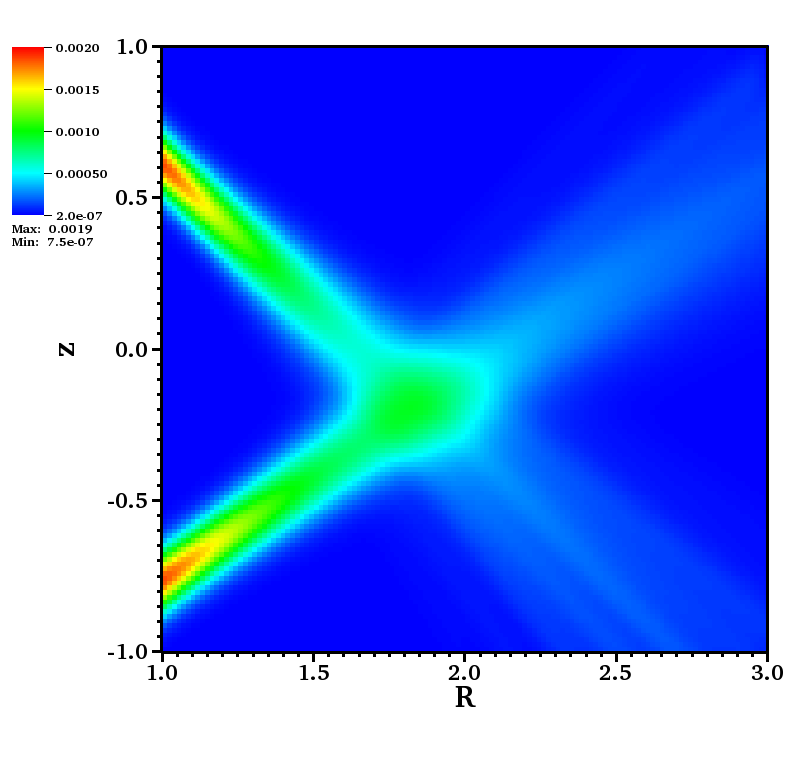} &
    \includegraphics[scale=0.205]{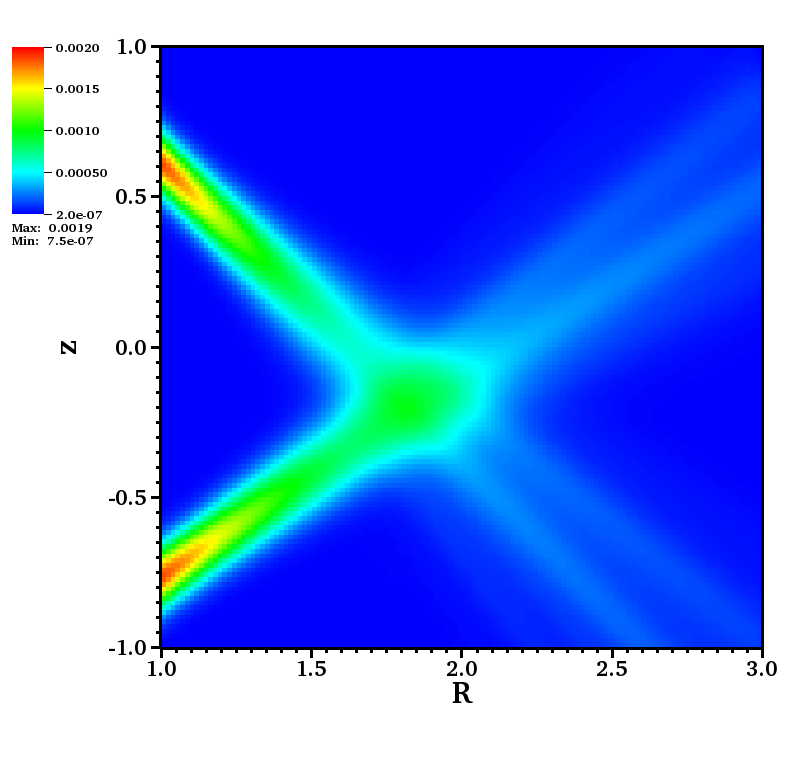}
  \end{tabular}
  \caption{Results from running the ``two-beam" test using various schemes and resolutions (see text for details).  
  The images show the angular moment of the distribution function in \eqref{eq:angularMomentAxialSymmetry} at $t=2.6$.  
  (The pixelation is due to the visualization, which assigns a constant value to each cell.)}
  \label{fig:TwoBeamTest}
\end{figure}

In Figure \ref{fig:TwoBeamTest_LineOuts_1_and_2}, to complement the images in Figure \ref{fig:TwoBeamTest}, we plot horizontal cuts (for $z=0.185$) through the images displayed in Figure \ref{fig:TwoBeamTest}.  
In the left panel we plot the angular moment versus distance from the symmetry axis for DG(0)+FE ($64^{2}\times24\times36$; green), DG(0)+FE ($256^{2}\times96\times144$; red), DG(1)+RK2 ($64^{2}\times24\times36$; blue), and DG(2)+RK3 ($64^{2}\times24\times36$; black).  

All models maintain positivity of the cell averaged distribution function during the evolution.  
For the higher-order schemes, the bound-preserving limiter is required to prevent negative distributions in certain cells, especially near the beam-fronts when they propagate through the computational domain.  
(Maintaining $\fDG\le1$ is not considered an issue in this test.)  
The four upper panels in Figure \ref{fig:TwoBeamTest} and the left panel in Figure \ref{fig:TwoBeamTest_LineOuts_1_and_2} demonstrate the effect of using a high-order method.  
At low resolution, the first-order scheme (green line in Figure \ref{fig:TwoBeamTest_LineOuts_1_and_2}) is clearly too diffusive for this problem.  
The Gaussian peak to the left is reduced by almost a factor of two, when compared to the results obtained with the higher order schemes using the same phase space resolution.  
The Gaussian peak to the right is virtually smeared out.  
The results obtained with the second and third order schemes (blue and black lines, respectively) appear similar for this problem.  
Even, when the phase space resolution is increased by a factor of four in each dimension (red line in Figure \ref{fig:TwoBeamTest_LineOuts_1_and_2}), the results obtained with the first-order method appear smeared out when compared to the results obtained with the second- and third-order schemes, which use factors of $16$ and $\sim3$ fewer total degrees of freedom, respectively.  

We observe ``ray effects" \citep[e.g.,][]{lewisMiller_1993} in the results obtained with the high-order DG schemes.  
The ray effects appear as oscillations in the numerical solution; cf. the black line around the Gaussian peak to the right in the left panel in Figure \ref{fig:TwoBeamTest_LineOuts_1_and_2}.  
In the right panel in Figure \ref{fig:TwoBeamTest_LineOuts_1_and_2} we plot a zoomed-in view of this second Gaussian.  
We plot the results obtained with DG(1)+RK2 ($64^{2}\times24\times36$; dashed blue) and DG(2)+RK3 ($64^{2}\times24\times36$; dashed black).  
We also plot results obtained by increasing the spatial resolution by a factor of two in each position space dimension, while keeping the momentum space angular resolution fixed; i.e., DG(1)+RK2 ($128^{2}\times24\times36$; solid blue) and DG(2)+RK3 ($128^{2}\times24\times36$; solid black).  
From Figure \ref{fig:TwoBeamTest_LineOuts_1_and_2}, and the middle and bottom rows in Figure \ref{fig:TwoBeamTest}, we see that increasing the \emph{spatial} resolution does not reduce the appearance of the ray effects.  

\begin{figure}
  \centering
  \begin{tabular}{cc}
    \includegraphics[scale=0.225]{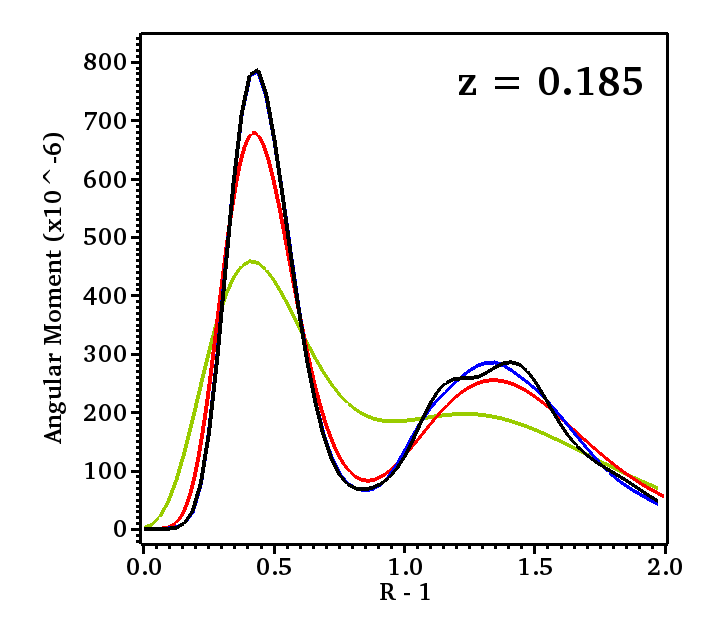} &
    \includegraphics[scale=0.225]{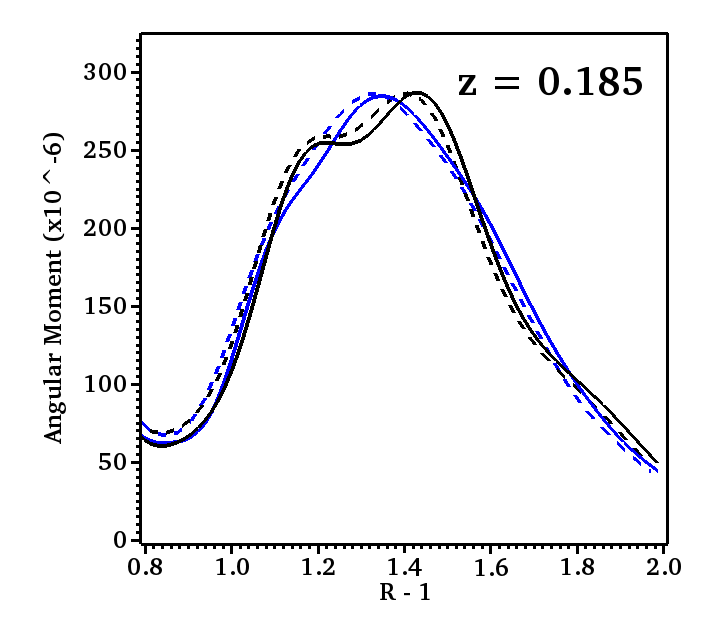}
  \end{tabular}
  \caption{Horizontal cuts ($z=0.185$) through the images displayed in Figure \ref{fig:TwoBeamTest} (see text for details).}
  \label{fig:TwoBeamTest_LineOuts_1_and_2}
\end{figure}

We have computed additional models to examine the appearance of ray effects in the DG scheme.  
In Figures \ref{fig:TwoBeamTestAngularResolutions} and \ref{fig:TwoBeamTest_LineOuts_3} we plot results obtained with the second-order scheme, DG(1)+RK2, using various momentum space angular resolution, while we keep the position space resolution fixed to $128\times128$.  
In Figure \ref{fig:TwoBeamTestAngularResolutions} we display the spatial distribution of the angular moment of the distribution function for the different resolutions: $128^{2}\times16\times24$ (upper left), $128^{2}\times24\times36$ (upper right), $128^{2}\times32\times48$ (lower left), and $128^{2}\times48\times72$ (lower right). 
In Figure \ref{fig:TwoBeamTest_LineOuts_3}, we plot horizontal cuts ($z=0.185$), angular moment versus distance $\rPerp$, through the same models: $128^{2}\times16\times24$ (green), $128^{2}\times24\times36$ (red), $128^{2}\times32\times48$ (blue), and $128^{2}\times48\times72$ (black).  
The appearance of ray effects diminish with increasing momentum space angular resolution.  
Strong ray effects are present in the low-resolution model ($\Delta\mu/L_{\mu}=1.25$).  
However, they are barely noticeable to the eye in the $128^{2}\times32\times48$-model ($\Delta\mu/L_{\mu}=0.625$; cf. lower left panel in Figure \ref{fig:TwoBeamTestAngularResolutions}), while they are not present at all in the $128^{2}\times48\times72$-model ($\Delta\mu/L_{\mu}\approx0.42$).  

\begin{figure}
  \centering
  \begin{tabular}{cc}
    \includegraphics[scale=0.205]{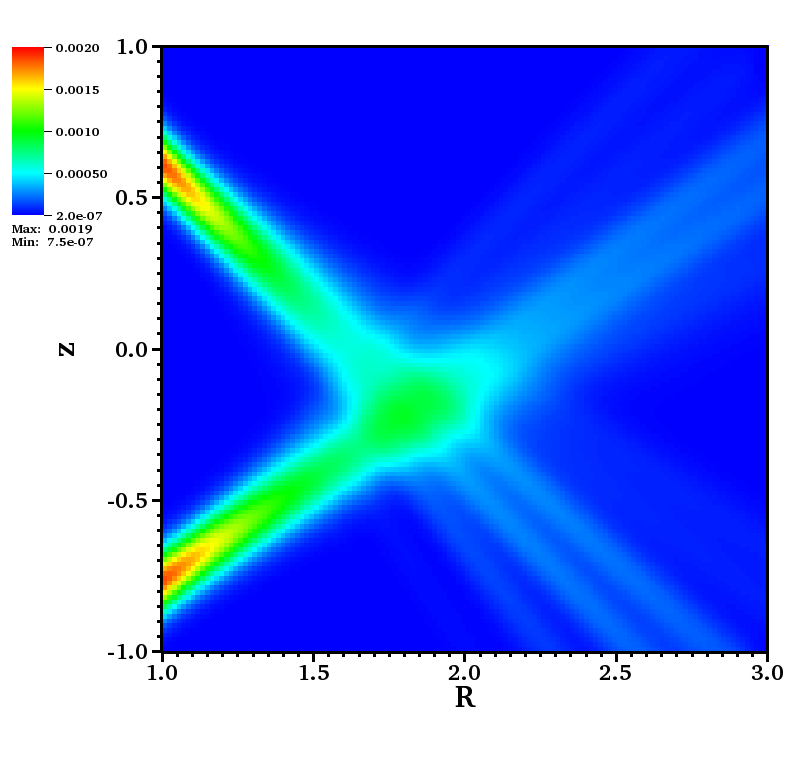} &
    \includegraphics[scale=0.205]{./TwoBeamTest_128x128x24x36_DG1_RK2.png} \\
    \includegraphics[scale=0.205]{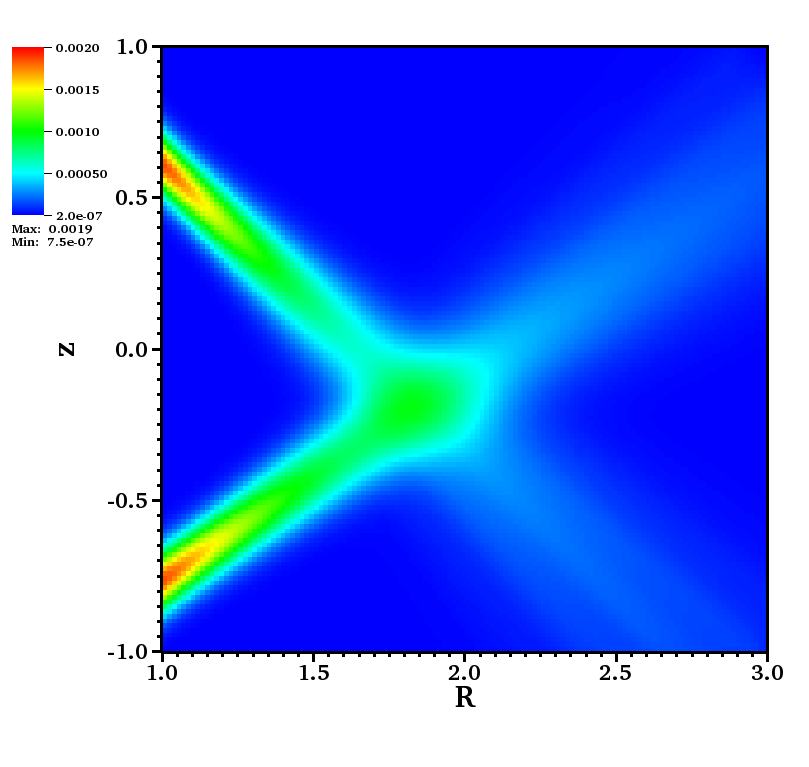} &
    \includegraphics[scale=0.205]{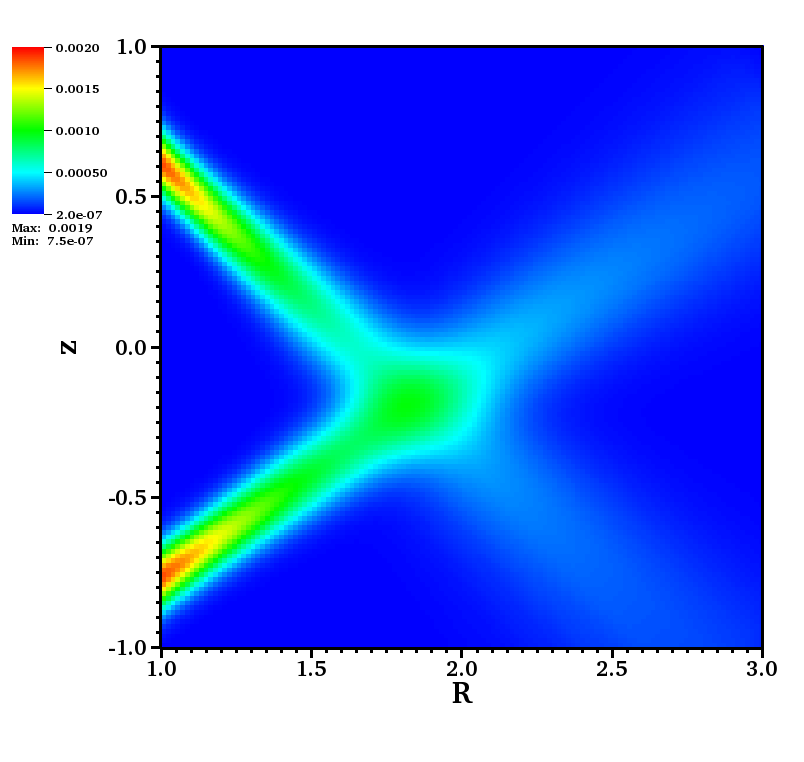}
  \end{tabular}
  \caption{Results from running the ``two-beam" test with DG(1)+RK2 using various momentum space angular resolutions.  
  The images show the angular moment of the distribution function at $t=2.6$.  (See text for details.)}
  \label{fig:TwoBeamTestAngularResolutions}
\end{figure}

\begin{figure}
  \centering
  \includegraphics[scale=0.25]{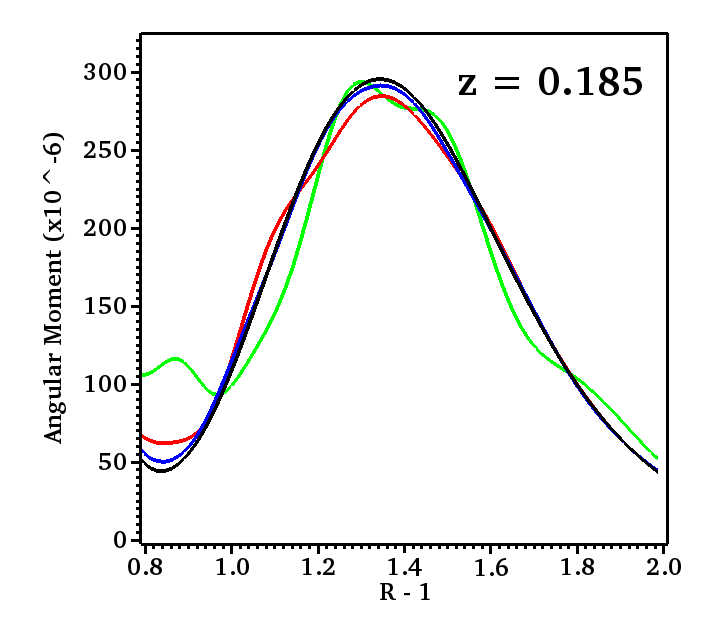}
  \caption{Horizontal cuts ($z=0.185$) through the images displayed in Figure \ref{fig:TwoBeamTestAngularResolutions}.  
  The results were obtained with the DG(1)+RK2 scheme with fixed position space resolution and various momentum space angular resolutions: $128^{2}\times16\times24$ (green), $128^{2}\times24\times36$ (red), $128^{2}\times32\times48$ (blue), and $128^{2}\times48\times72$ (black).}
  \label{fig:TwoBeamTest_LineOuts_3}
\end{figure}

%% file: conclusion.tex
\section{Summary and Conclusions}
\label{sec:conclusions}

We have developed high-order, bound-preserving methods for solving the conservative phase space advection problem for radiation transport.  
We have presented discontinuous Galerkin (DG) methods for solving the conservative, general relativistic collision-less Boltzmann equation in up to six dimensions assuming \emph{time-independent} spacetimes.  
Specific examples are given for problems with reduced dimensionality from imposed symmetries; namely, spherical symmetry in flat and curved spacetime (Sections \ref{sec:maximumPrincipleSphericalSymmetry} and \ref{sec:maximumPrincipleSphericalSymmetryGR}, respectively) and axial symmetry in flat spacetime (Section \ref{sec:maximumPrincipleAxialSymmetry}).  
With the eventual goal of simulating neutrino transport in dense nuclear matter, which obey Fermi-Dirac statistics, we have taken special care to ensure that the high-order DG method preserves the maximum principle for the phase space distribution function; i.e., $f\in[0,1]$. 
The combination of suitable CFL conditions and the use of the conservative, high-order bound-preserving limiter in \citep{ZS2010a} are sufficient to ensure positivity of the distribution function (i.e., $f\ge0$).  
For the conservative formulation we employ, the additional requirement that the phase space discretization preserves the divergence-free character of the Liouville flow is necessary to ensure that the distribution function satisfies the full maximum principle during the evolution (i.e., $0\le f\le1$).  
High-order accuracy, bound-preserving properties, as well as other properties of the DG scheme are demonstrated with numerical examples in Section~\ref{sec:numericalExamples}.  

In our opinion, the DG method is an attractive option for simulating supernova neutrino transport.  
However, several challenges --- which we defer to future studies --- remain to be solved before it can be deployed with confidence in large-scale multiphysics simulations with all the relevant physics included.  
In particular, the bound-preserving DG scheme must be extended to include necessary neutrino-matter interactions.  
Here, the use of implicit-explicit methods may be used in order to bypass timescales imposed by short radiation mean-free paths in neutrino opaque regions (i.e., in the proto-neutron star).  
Our bound-preserving scheme must be extended to the case with time-dependent spacetimes (we assumed $\p_{t}\smdet=0$ in Section~\ref{sec:maximumPrinciple}).  
Moreover, velocity-dependent effects (i.e., Doppler shift and aberration) must be correctly accounted for when the radiation is interacting with a moving stellar fluid \citep[e.g.,][]{mihalasMihalas_1999}.  
One the one hand, the neutrino-matter interactions are most easily handled in a frame that is comoving with the fluid.  
On the other hand, the Liouville equation is mathematically simpler in the so-called laboratory-frame formulation \citep[see discussions in e.g.,][]{mihalasAuer_2001,cardall_etal_2013}.  
Mihalas \& Klein \citep{mihalasKlein_1982} formulated the ``mixed-frame" approach, valid to $\mathcal{O}(v/c)$, which combines the advantages of these two formulations, but this approach is, as far as we know, not extendable to relativistic flows; however see the approach proposed in \citep{nagakura_etal_2014}.  
Finally, we note that the numerical methods must be developed to conserve neutrino four-momentum in limits where such conservation laws can be stated (e.g., flat or asymptotically flat spacetimes).  
Our numerical phase space advection scheme conserves particles by construction, but is in general not guaranteed to conserve energy and momentum.  
The possibility of extending the approach in \citep{liebendorfer_etal_2004} to higher dimensions and high-order accuracy should be investigated.  

We note that high-order DG methods are computationally expensive in terms of memory usage for high-dimensional problems.  
In this paper, the numerical solutions are constructed from the so-called tensor product basis, $\bbQ_{(d)}^{k}$; the $d$-dimensional polynomial space formed from tensor products of one-dimensional polynomials of degree $\le k$.  
The total number of degrees of freedom per phase space cell is then $||\bbQ_{(d)}^{k}||=k^{d}$.  
To save computational resources, one may use the total degree polynomial basis, denoted $\bbP_{(d)}^{k}$, by constructing the numerical solution from multi-dimensional polynomials of total degree $\le k$.  
The number of degrees of freedom per phase space cell is then $||\bbP_{(d)}^{k}||=(k!+d!)/(k!\,d!)$, which is significantly smaller than $||\bbQ_{(d)}^{k}||$ for high-dimensional problems (i.e., $d=6$) when high-order accuracy is desired (i.e., $k>2$).  
In order to further reduce the overall memory footprint, the filtered spherical harmonics approach to momentum space angular discretization \citep{MH2010,radice_etal_2013} may be an attractive option for core-collapse supernova neutrino transport simulations.  
However, proper inclusion of all the relevant physics discussed above remains a forefront research topic in computational physics.  

%% file: appendix.tex
\appendix

\section{Conservative Boltzmann Equations}
\label{sec:equations}

Our long-term goal is to develop robust and efficient numerical methods for solving the general relativistic Boltzmann equation for neutrino transport, coupled with corresponding fluid and gravitational field equations, to study the explosion mechanism of massive stars.  
This is a formidable task, which is far beyond the scope of this paper.  
In this study we ignore radiation-matter interactions on the right-hand side of the Boltzmann equation, and focus on numerical methods for the left-hand side; i.e., the phase space advection problem.  
To this end, we consider the fully general relativistic case, but assume a \emph{time-independent} spacetime.  
For reference and completeness, we include general and special relativistic Boltzmann equations in this appendix.  
We adopt a `geometrized' unit system in which the vacuum speed of light and the Planck constant are unity.  
Where appropriate, we adopt the usual Einstein summation convention where repeated Greek indices run from $0$ to $3$, and repeated Latin indices run from $1$ to $3$.  
We use the metric signature $(-,+,+,+)$.  

\subsection{General Relativistic Boltzmann Equation}
\label{sec:BoltzmannGR}

It is necessary to employ a general relativistic description in order to study non-equilibrium transport processes in systems involving dynamical spacetimes (e.g., neutrino transport simulations aimed at understanding the explosion mechanism of massive stars).  
General relativistic formulations of kinetic theory (including the Boltzmann transport equation) have been presented in various forms and discussed in detail by several authors \citep[see e.g.,][]{lindquist_1966,ehlers_1971,israel_1972,thorne_1981,riffert_1986,mezzacappaMatzner_1989,cercignaniKremer_2002,cardallMezzacappa_2003,cardall_etal_2013,shibata_etal_2014}.  
Thus, the presentation given here is intentionally brief.  

\paragraph{Conservative General Relativistic Formulation}
For numerical solution we employ the \emph{conservative} form of the Boltzmann equation.  
The conservative form has desirable mathematical properties when the solution can develop discontinuities.  
It is also better suited for tracking conserved quantities (e.g., particle number and energy).  
The conservative, general relativistic Boltzmann equation can be written as \citep[see][for details]{cardallMezzacappa_2003,cardall_etal_2013}
\begin{equation}
  \f{1}{\mdet}\pderiv{}{\xu{\mu}}\Big(\,\mdet\,\f{\pu{\mu}}{\tetradEpsilon}\,f\,\Big)
  +\f{1}{\pmdet}\pderiv{}{\putld{\imath}}\Big(\,\pmdet\,\jacPutlddb{\imath}{\imath}\,\pu{\nu}\,\pu{\rho}\,\big(\covderiv{\rho}\tetubd{\imath}{\nu}\big)\,\f{1}{\tetradEpsilon}\,f\,\Big)
  =\f{1}{\tetradEpsilon}\,\collision{f}.  
  \label{eq:ConservativeBoltzmannEquationGeneralRelativistic}
\end{equation}
Cardall et al. \citep{cardall_etal_2013} derived the conservative form of the Boltzmann equation from the corresponding non-conservative form by showing that the ``Liouville flow" is divergence-free; i.e., 
\begin{equation}
  \f{1}{\mdet}\pderiv{}{\xu{\mu}}\Big(\,\mdet\,\f{\pu{\mu}}{\tetradEpsilon}\,\Big)
  +\f{1}{\pmdet}\pderiv{}{\putld{\imath}}\Big(\,\pmdet\,\jacPutlddb{\imath}{\imath}\,\pu{\nu}\,\pu{\rho}\,\big(\covderiv{\rho}\tetubd{\imath}{\nu}\big)\,\f{1}{\tetradEpsilon}\,\Big)=0.  
  \label{eq:divergenceFreePhaseSpaceFlowGeneralRelativistic}
\end{equation}
In Equation \eqref{eq:ConservativeBoltzmannEquationGeneralRelativistic}, $\{\xu{\mu}\}$ are spacetime position components in a global coordinate basis.  
The geometry of spacetime is encoded in the metric tensor $\gdd{\mu}{\nu}$, whose determinant is denoted $g$.  
The components of the particle four-momentum are $\{\pu{\mu}\}$.  
The collision term on the right-hand side, $\collision{f}$, describes energy and momentum exchange due to point-like collisions (e.g., radiation-matter interactions).  
In Equation \eqref{eq:ConservativeBoltzmannEquationGeneralRelativistic}, the particle distribution function is a function of spacetime position coordinates in the global coordinate basis, while momentum coordinates are defined with respect to a local orthonormal basis\footnote{In the general theory of relativity, the existence of a local orthonormal basis at every spacetime point is assumed.}.  
(We take only the spatial four-momentum components as independent variables due to the mass shell constraint $\pu{\mu}\pd{\mu}=0$.)
The coordinate transformation $\tetudb{\mu}{\mu}=\partial\xu{\mu}/\partial\xub{\mu}$ (and its inverse $\tetubd{\mu}{\mu}$) locally transforms between four-vectors associated with the coordinate basis (unadorned indices) and four-vectors associated with an orthonormal (tetrad) basis (indices adorned with a bar); e.g., $\pu{\mu}=\tetudb{\mu}{\mu}\,\pub{\mu}$.  
Equivalently, $\tetudb{\mu}{\mu}$ locally transforms the spacetime metric into the Minkowskian; i.e.,
\begin{equation}
  \tetudb{\mu}{\mu}\,\tetudb{\nu}{\nu}\,\gdd{\mu}{\nu}=\mbox{diag}\big[-1,1,1,1\big].  
\end{equation}
In Equation \eqref{eq:ConservativeBoltzmannEquationGeneralRelativistic}, we allow for the use of curvilinear three-momentum coordinates (indices adorned with a tilde), defined with respect to the local orthonormal basis.  
The Jacobian matrix $\jacPutlddb{\imath}{\imath}=\partial\putld{\imath}/\partial\pub{\imath}$ is due to a change to curvilinear from ``Cartesian" three-momentum coordinates.  
As an example used in this paper, the Cartesian momentum components can be expressed in terms of spherical momentum coordinates $\{\putld{\imath}\}=\{\tetradEpsilon,\tetradTheta,\tetradPhi\}$ (the energy $\tetradEpsilon$ and two angles $\tetradTheta$ and $\tetradPhi$) as
\begin{equation}
  \{\,\pub{1},\pub{2},\pub{3}\,\}=\tetradEpsilon\,\{\cos\tetradTheta,\sin\tetradTheta\cos\tetradPhi,\sin\tetradTheta\sin\tetradPhi\}, 
\end{equation}
from which the transformation $\jacPubdtld{\imath}{\imath}=\partial\pub{\imath}/\partial\putld{\imath}$ and its inverse $\jacPutlddb{\imath}{\imath}$ can be computed directly \citep[see for example Equations (24) and (25) in][]{cardall_etal_2013}.  
The momentum space three-metric $\lambdadtlddtld{\imath}{\jmath}$ (with inverse $\lambdautldutld{\imath}{\jmath}$ and determinant $\lambda$) provides the proper distance between points in three-dimensional momentum space; i.e., $ds_{\vect{p}}^{2}=\lambdadtlddtld{\imath}{\jmath}\,d\putld{\imath}\,d\putld{\jmath}$.  

We have written the distribution function in terms of spacetime position components in a global coordinate basis and three-momentum components in a local orthonormal basis; i.e., $f=f\big(\,\xu{\mu},\putld{\imath}\,\big)$.  
The use of distinct position and momentum coordinates for radiation transport was discussed in detail in \citep{cardallMezzacappa_2003,cardall_etal_2013}.  
The use of an orthonormal basis for the radiation four-momentum eliminates (locally) the effects of the curved spacetime geometry (i.e., the gravitational field), which is advantageous when describing local physics (i.e., radiation matter interactions).  
However, in curved spacetime it is not possible to globally eliminate the gravitational field by any coordinate transformation.  

\paragraph{Conservative 3+1 Formulation}
For numerical simulations involving dynamical spacetimes, the so-called 3+1 splitting of spacetime \citep[e.g.,][]{misner_etal_1973,gourgoulhon_2007,baumgarteShapiro_2010} is commonly employed.  
In the 3+1 approach, the four-dimensional spacetime is foliated into a ``stack" of three-dimensional spatial hypersurfaces $\Sigma_{t}$ labeled with time coordinate $t$.  
The 3+1 form of the invariant interval between neighboring points in four-dimensional spacetime is given by
\begin{equation}
  ds^{2}=-\alpha^{2}\,dt^{2}+\gmdd{i}{j}\,\big(d\xu{i}+\betau{i}dt\big)\big(d\xu{j}+\betau{j}dt\big), 
  \label{eq:threePlusOneMetric}
\end{equation}
where $\alpha\,dt$ is the proper time between spatial hypersurfaces $\Sigma_{t}$ and $\Sigma_{t+dt}$, $\gmdd{i}{j}$ is the spatial three-metric, and $ds_{\vect{x}}^{2}=\gmdd{i}{j}\,\big(d\xu{i}+\betau{i}dt\big)\big(d\xu{j}+\betau{j}dt\big)$ gives the proper distance within a spatial hypersurface \citep[e.g.,][]{baumgarteShapiro_2010}.  
The lapse function $\alpha$ and the (spatial) shift vector $\betau{i}$ are freely specifiable functions associated with the freedom to arbitrarily specify time and space coordinates.  
A straightforward calculation of the determinant of the spacetime metric gives $\mdet=\alpha\smdet$, where $\gamma$ is the determinant of the spatial metric.  

The normal vector to a spacelike hypersurface can be written in terms of coordinate basis metric components as
\begin{equation}
  \fourvelocityEu{\mu}=\alpha^{-1}\big(1,-\betau{i}\big), 
  \label{eq:eulerianFourvelocity}
\end{equation}
where the normalization condition $\fourvelocityEd{\mu}\fourvelocityEu{\mu}=-1$ implies $\fourvelocityEd{\mu}=\big(-\alpha,0,0,0\big)$.  
For the derivation of the $3+1$ form of the Boltzmann equation, we use the ``Eulerian" decomposition of the four-momentum,
\begin{equation}
  \pu{\mu}=\tetradEpsilon\,\big(\,\fourvelocityEu{\mu}+\lu{\mu}\,\big),
  \label{eq:fourMomentumEulerian}
\end{equation}
where $\tetradEpsilon=-\fourvelocityEd{\mu}\,\pu{\mu}$ is the particle energy seen by an `Eulerian observer' with timelike four-velocity $\fourvelocityEu{\mu}$, and $\lu{\mu}$ is a spacelike coordinate basis unit four-vector orthogonal to $\fourvelocityEu{\mu}$ (i.e., $\ld{\mu}\lu{\mu}=1$ and $\fourvelocityEd{\mu}\lu{\mu}=0$).  
Then, the conservative general relativistic 3+1 Boltzmann equation can be written as
\begin{eqnarray}
  \f{1}{\alpha\smdet}
  \Big[\,
    \pderiv{}{t}\Big(\,\smdet\,f\,\Big)
    +\pderiv{}{\xu{i}}\Big(\,\smdet\,\big[\,\alpha\,\lu{i}-\betau{i}\,\big]\,f\,\Big)
  \,\Big]
  +\f{1}{\pmdet}\pderiv{}{\putld{\imath}}\Big(\,\pmdet\,\mathcal{R}^{\tilde{\imath}}\,f\,\Big)
  =\f{1}{\tetradEpsilon}\,\collision{f}, 
  \label{eq:ConservativeBoltzmannEquationGeneralRelativisticThreePlusOne}
\end{eqnarray}
where
\begin{eqnarray}
  \mathcal{R}^{\tilde{\imath}}
  &=&
  \jacPutlddb{\imath}{\imath}\,\pu{\nu}\,\pu{\rho}\,\big(\covderiv{\rho}\tetubd{\imath}{\nu}\big)\,\f{1}{\tetradEpsilon} \nonumber \\
  &=&
  -\tetradEpsilon\,\lambdautldutld{\imath}{\jmath}\,\pderiv{\tetradEpsilon}{\putld{\jmath}}\,\lu{i}\,
  \Big\{\,
    \f{1}{\alpha}\pderiv{\alpha}{\xu{i}}
    -\lu{j}\,K_{ij}
  \,\Big\} \nonumber \\
  && \hspace{12pt}
  -\tetradEpsilon^{2}\,\lambdautldutld{\imath}{\jmath}\,\pderiv{\lu{i}}{\putld{\jmath}}\,
  \Big\{\,
    \tauderivc{\ld{i}}
    +\f{1}{\alpha}\pderiv{\alpha}{\xu{i}}
    -\f{\ld{j}}{\alpha}\pderiv{\betau{j}}{\xu{i}}
    -\f{1}{2}\,\lu{j}\,\lu{k}\,\pderiv{\gmdd{j}{k}}{\xu{i}}
  \,\Big\}.
  \label{eq:gravitationalRedShiftThreePlusOne}
\end{eqnarray}
describes momentum space advection (e.g., redshift and angular aberration) due to gravitational (i.e., curved spacetime) and other geometric effects (arising from the use of curvilinear coordinates).  
In Equation \eqref{eq:gravitationalRedShiftThreePlusOne}, $K_{ij}$ is the extrinsic curvature tensor \citep{baumgarteShapiro_2010}, and we have defined the derivative
\begin{equation}
  \tauderivc{}
  =\tauderiv{}+\lu{j}\,\pderiv{}{\xu{j}}
  =\f{1}{\alpha}\,\Big\{\,\pderiv{}{t}+\big(\,\alpha\,\lu{j}-\betau{j}\,\big)\,\pderiv{}{\xu{j}}\,\Big\}.
\end{equation}
The spacetime divergence part of Equation \eqref{eq:ConservativeBoltzmannEquationGeneralRelativisticThreePlusOne} arises easily from Equation \eqref{eq:ConservativeBoltzmannEquationGeneralRelativistic} with the Eulerian decomposition of $\pu{\mu}$ and the specification of $\fourvelocityEu{\mu}$, while the momentum space divergence is more complicated.  
We include details of the derivation of Equation \eqref{eq:gravitationalRedShiftThreePlusOne} in Section \ref{sec:gravitationalRedShiftThreePlusOne}.  

Note that our use of the Eulerian decomposition of the four-momentum as given in Equation \eqref{eq:fourMomentumEulerian} differs slightly from the formalism used in \citep{cardall_etal_2013}, where Eulerian decompositions of the ``tetrad" transformation (e.g., $L^{\mu}_{~\hat{\mu}}$ in their notation) was employed.  
Also note that we have expressed the radiation four-momentum in terms of an orthonormal ``lab-frame" basis, while an orthonormal ``comoving" basis was used in \citep{cardall_etal_2013}.  
This distinction is very important to consider when the radiation interacts with a moving fluid \citep[e.g.,][]{mihalasMihalas_1999,mihalasAuer_2001}.  
However, for a static fluid, the two formulations coincide.  
We defer the case where the radiation interacts with a moving fluid to a future study.  

\paragraph{Spherically Symmetric Spacetime}
As a simplification used for numerical implementation in this study, we adopt spherical polar spatial coordinates $\{\xu{i}\}=\{r,\theta,\phi\}$ and spherical polar momentum coordinates $\{\putld{\imath}\}=\{\tetradEpsilon,\tetradTheta,\tetradPhi\}$, and specialize Equation \eqref{eq:ConservativeBoltzmannEquationGeneralRelativisticThreePlusOne} to a spherically symmetric spacetime with a metric of the following form
\begin{equation}
  ds^{2}=-\alpha^{2}\,dt^{2}+\gmdd{i}{j}\,d\xu{i}\,d\xu{j}, 
  \label{eq:metricDiagonalCFC}
\end{equation}
(i.e., $\betau{i}=0$) with $\gmdd{i}{j}=\psi^{4}\,\mbox{diag}[1,r^{2},r^{2}\,\sin^{2}\theta]$, $\smdet=\psi^{6}\,r^{2}\,\sin\theta$, and where $\psi$ is the ``conformal factor."  
Furthermore we assume that the metric components are independent of the time coordinate, and we write $\alpha=\alpha(r)$ and $\psi=\psi(r)$.  
Then, all the components of the extrinsic curvature tensor vanish (i.e., $K_{ij}=0$).  

With the diagonal metric tensor in Equation \eqref{eq:metricDiagonalCFC}, we can easily write the transformation between the coordinate basis and the orthonormal tetrad basis as $\tetudb{\mu}{\mu}=\mbox{diag}[\,\alpha^{-1},\tetudb{i}{\imath}\,]$, where $\tetudb{i}{\imath}=\psi^{-2}\,\mbox{diag}[1,r^{-1},(r\,\sin\theta)^{-1}]$.  
We then obtain the conservative Boltzmann equation valid for spherically symmetric spacetimes under the assumptions stated above
\begin{align}
  &
  \f{1}{\alpha}\pderiv{f}{t}
  +\f{1}{\alpha\,\psi^{6}\,r^{2}}\pderiv{}{r}\Big(\,\alpha\,\psi^{4}\,r^{2}\,\mu\,f\,\Big)
  -\f{1}{\tetradEpsilon^{2}}\pderiv{}{\tetradEpsilon}
  \Big(\,\tetradEpsilon^{3}\,\f{1}{\psi^{2}\,\alpha}\pderiv{\alpha}{r}\,\mu\,f\,\Big) \nonumber \\
  & \hspace{12pt}
  +\pderiv{}{\mu}
  \Big(\,\big(1-\mu^{2}\big)\,\psi^{-2}\,
    \Big\{\,
      \f{1}{r}
      +\f{1}{\psi^{2}}\pderiv{\psi^{2}}{r}
      -\f{1}{\alpha}\pderiv{\alpha}{r}
    \,\Big\}\,f
  \,\Big)
  =\f{1}{\tetradEpsilon}\,\collision{f}, 
  \label{eq:ConservativeBoltzmannEquationSphericalSymmetryGRApp}
\end{align}
where the angle cosine is defined as $\mu=\cos\tetradTheta$.  
In particular, Equation \eqref{eq:ConservativeBoltzmannEquationSphericalSymmetryGRApp} is sufficiently general to accommodate the Schwarzschild metric (an exact solution of Einstein's field equations), where
\begin{equation}
  \alpha=\f{1-\f{M}{2\,r}}{1+\f{M}{2\,r}}
  \quad\mbox{and}\quad
  \psi=1+\f{M}{2\,r}, 
  \label{eq:schwarzschildMetric}
\end{equation}
and $M$ is the spacetime mass observed by a distant static observer \citep{baumgarteShapiro_2010}.  
We adopt the Schwarzschild metric and solve Equation \eqref{eq:ConservativeBoltzmannEquationSphericalSymmetryGRApp} numerically in Section \ref{sec:numericalExamplesSphericalSymmetryGR}.  

\subsection{Boltzmann Equation in Flat Spacetimes}
\label{sec:BoltzmannSR}

In this section, we present conservative Boltzmann equations which are considered in the numerical simulations where we use a flat spacetime metric.  
The equations presented here follow directly from simplification of the general relativistic equations in the previous section.   

\paragraph{Conservative Formulation for General Phase Space Coordinates}
For a flat spacetime, but allowing for general curvilinear phase space (spatial and momentum) coordinates, we write the spacetime metric as (i.e., obtained by setting $\alpha=1$ and $\betau{i}=0$ in Equation \eqref{eq:threePlusOneMetric})
\begin{equation}
  \gdd{\mu}{\nu}
  =\left(\begin{array}{cc}
    -1 & 0 \\
    0 & \gmdd{i}{j}
  \end{array}\right), 
\end{equation}
where the spatial metric $\gmdd{i}{j}$ provides the proper distance between points in three-dimensional position space; i.e., $ds_{\vect{x}}^{2}=\gmdd{i}{j}\,d\xu{i}\,d\xu{j}$.  
In this case, Equation \eqref{eq:ConservativeBoltzmannEquationGeneralRelativisticThreePlusOne} can be written as
\begin{eqnarray}
  &&
  \pderiv{f}{t}
  +\f{1}{\smdet}\pderiv{}{\xu{i}}\Big(\,\smdet\,\lu{i}\,f\,\Big)
  +\f{1}{\pmdet}\pderiv{}{\putld{\imath}}\Big(\,\pmdet\,\mathcal{R}^{\tilde{\imath}}\,f\,\Big)
  =\f{1}{\tetradEpsilon}\,\collision{f}, 
  \label{eq:ConservativeBoltzmannEquationCurvilinearFlat}
\end{eqnarray}
where ``geometric" terms describing momentum space advection due to the use of curvilinear coordinates (cf. Equation \eqref{eq:gravitationalRedShiftThreePlusOne}) are given by
\begin{equation}
  \mathcal{R}^{\tilde{\imath}}
  =
  -\tetradEpsilon^{2}\,\lambdautldutld{\imath}{\jmath}\,\pderiv{\lu{i}}{\putld{\jmath}}\,
    \Big\{\,
      \lu{j}\,\pderiv{\ld{i}}{\xu{j}}
      -\f{1}{2}\,\lu{j}\,\lu{k}\,\pderiv{\gmdd{j}{k}}{\xu{i}}
    \,\Big\}.
\end{equation}
Note that $\mathcal{R}^{\tilde{\imath}}=0$ when Cartesian coordinates are used; i.e., $\gmdd{i}{j}=\mbox{diag}[1,1,1]$.  

Below, we adopt spherical polar momentum coordinates $(\tetradEpsilon,\tetradTheta,\tetradPhi)$ and consider two specializations of Equation \eqref{eq:ConservativeBoltzmannEquationCurvilinearFlat}.  

\paragraph{Spherical Symmetry (Spherical Polar Spatial Coordinates)}
By adopting spherical polar spatial coordinates $\{\xu{i}\}=\{r,\theta,\phi\}$, the spatial metric tensor is given by $\gmdd{i}{j}=\mbox{diag}\big[\,1,r^{2},r^{2}\,\sin^{2}\theta\,\big]$.  
Then, by imposing spherical symmetry ($\p_{\theta},\p_{\phi}=0$), Equation \eqref{eq:ConservativeBoltzmannEquationCurvilinearFlat} becomes
\begin{equation}
  \pderiv{f}{t}
  +\f{1}{r^{2}}\pderiv{}{r}\Big(\,r^{2}\,\mu\,f\,\Big)
  +\pderiv{}{\mu}\Big(\,\big(1-\mu^{2}\big)\,\f{1}{r}\,f\,\Big)
  =0.  
  \label{eq:ConservativeBoltzmannEquationSphericalSymmetryFlatApp}
\end{equation}
We solve Equation \eqref{eq:ConservativeBoltzmannEquationSphericalSymmetryFlat} numerically in Section \ref{sec:numericalExamplesSphericalSymmetry}.

\paragraph{Axial Symmetry (Cylindrical Spatial Coordinates)}
In cylindrical spatial coordinates $\{\xu{i}\}=\{\rPerp,z,\phi\}$ the metric tensor is given by $\gmdd{i}{j}=\mbox{diag}\big[\,1,1,\rPerp^{2}\,\big]$.  
By imposing axial symmetry ($\p_{\phi}=0$), Equation \eqref{eq:ConservativeBoltzmannEquationCurvilinearFlat} becomes
\begin{equation}
  \pderiv{f}{t}
  +\f{1}{\rPerp}\pderiv{}{\rPerp}\Big(\,\rPerp\,\sqrt{1-\mu^{2}}\,\cos\tetradPhi\,f\,\Big)
  +\pderiv{}{z}\Big(\,\mu\,f\,\Big)
  -\f{1}{\rPerp}\pderiv{}{\tetradPhi}
  \Big(\,
    \sqrt{1-\mu^{2}}\,\sin\tetradPhi\,f
  \,\Big)
  =\f{1}{\tetradEpsilon}\,\collision{f}.  
  \label{eq:ConservativeBoltzmannEquationAxialSymmetryFlatApp}
\end{equation}
We solve Equation \eqref{eq:ConservativeBoltzmannEquationAxialSymmetryFlat} numerically in Section \ref{sec:numericalExamplesAxialSymmetry}.  

\subsection{General Relativistic 3+1 Momentum Space Flux}
\label{sec:gravitationalRedShiftThreePlusOne}

Here we provide details on the derivation of the momentum space flux appearing in the conservative 3+1 general relativistic Boltzmann equation given in Section \ref{sec:BoltzmannGR} (Equation \eqref{eq:ConservativeBoltzmannEquationGeneralRelativisticThreePlusOne}). 
Some useful relations we use are \citep[cf.][]{cardall_etal_2013}
\begin{eqnarray}
  \fourvelocityEu{\mu}\,\covderiv{\mu}\fourvelocityEd{\nu}
  &=&\f{1}{\alpha}\pderiv{\alpha}{\xu{\nu}}, \label{eq:nDotGradn} \\
  \gmud{\mu}{i}\,\gmud{\nu}{j}\,\covderiv{\mu}\fourvelocityEd{\nu}
  &=&-K_{ij}, \label{eq:extrinsicCurvature} \\
  z_{\mu}\,\pderiv{\fourvelocityEu{\mu}}{\xu{\nu}}
  &=&-\f{z_{i}}{\alpha}\,\pderiv{\betau{i}}{\xu{\nu}}. \quad (\mbox{for $z^{\mu}$ spacelike}). \label{eq:spacelikeDotGradn}
\end{eqnarray}

We elaborate on the term appearing in the momentum space divergence in Equation \eqref{eq:ConservativeBoltzmannEquationGeneralRelativistic}; i.e.,
\begin{equation}
  \jacPutlddb{\imath}{\imath}\,\pu{\nu}\,\pu{\rho}\,\covderiv{\rho}\tetubd{\imath}{\nu}.  
  \label{eq:momentumSpaceNumberFluxGeneralRelativistic}
\end{equation}
We have
\begin{equation}
  \jacPutlddb{\imath}{\imath}
  =\pderiv{\putld{\imath}}{\pub{\imath}}
  =\lambdautldutld{\imath}{\jmath}\,\pderiv{\pdb{\imath}}{\putld{\jmath}}. 
\end{equation}
Then, by employing the Eulerian decomposition of the four-momentum in \eqref{eq:fourMomentumEulerian}, and noting that $\tetubd{\imath}{\nu}\,\pdb{\imath}=\tetradEpsilon\,\ld{\nu}$, we write Equation \eqref{eq:momentumSpaceNumberFluxGeneralRelativistic} as
\begin{equation}
  -\lambdautldutld{\imath}{\jmath}\,\pderiv{\big(\tetradEpsilon\,\lu{\nu}\big)}{\putld{\jmath}}\,\pu{\rho}\,\covderiv{\rho}\pd{\nu}
  =
  -\lambdautldutld{\imath}{\jmath}\,\pderiv{\tetradEpsilon}{\putld{\jmath}}\,\lu{\nu}\,\pu{\rho}\,\covderiv{\rho}\pd{\nu}
  -\tetradEpsilon\lambdautldutld{\imath}{\jmath}\,\pderiv{\lu{\nu}}{\putld{\jmath}}\,\pu{\rho}\,\covderiv{\rho}\pd{\nu}, 
  \label{eq:momentumSpaceNumberFluxIntermediate}
\end{equation}
where we have expanded with the product rule to get two expressions; one parallel and one perpendicular to $\lu{\nu}$ \citep[cf.][]{cardall_etal_2013}, since
\begin{equation}
  \ld{\nu}\,\pderiv{\lu{\nu}}{\putld{\jmath}}
  =\tetubd{\imath}{\nu}\,\tetudb{\nu}{\jmath}\,\ldb{\imath}\,\pderiv{\lub{\jmath}}{\putld{\jmath}}
  =\ldb{\imath}\,\pderiv{\lub{\imath}}{\putld{\jmath}}=0.  
\end{equation}

We can write the term $\lu{\nu}\,\pu{\rho}\,\covderiv{\rho}\pd{\nu}$ appearing on the right-hand side of Equation \eqref{eq:momentumSpaceNumberFluxIntermediate} as
\begin{equation}
  \tetradEpsilon^{2}\,\lu{\nu}\,
  \Big\{\,
    \fourvelocityEu{\rho}\,\covderiv{\rho}\fourvelocityEd{\nu}
    +\lu{\rho}\,\covderiv{\rho}\fourvelocityEd{\nu}
  \,\Big\}
  +\tetradEpsilon\,\pu{\rho}\,\lu{\nu}\,\covderiv{\rho}\ld{\nu}
  =\tetradEpsilon^{2}\,\lu{i}\,
  \Big\{\,
    \f{1}{\alpha}\pderiv{\alpha}{\xu{i}}
    -\lu{j}\,K_{ij}
  \,\Big\}, 
\end{equation}
where we have used the fact that $\lu{\nu}\,\covderiv{\rho}\ld{\nu}=0$ and Equations \eqref{eq:nDotGradn} and \eqref{eq:extrinsicCurvature}.  

For the second term on the right-hand side of Equation \eqref{eq:momentumSpaceNumberFluxIntermediate} we write
\begin{equation}
  \pderiv{\lu{\nu}}{\putld{\jmath}}\,\pu{\rho}\,\covderiv{\rho}\pd{\nu}
  =\tetradEpsilon^{2}\,\pderiv{\lu{\nu}}{\putld{\jmath}}\,
  \Big\{\,
    \fourvelocityEu{\rho}\,\covderiv{\rho}\fourvelocityEd{\nu}
    +\lu{\rho}\,\covderiv{\rho}\fourvelocityEd{\nu}
    +\fourvelocityEu{\rho}\,\covderiv{\rho}\ld{\nu}
    +\lu{\rho}\,\covderiv{\rho}\ld{\nu}
    \label{eq:momentumSpaceNumberFluxIntermediatePerp}
  \,\Big\}.
\end{equation}
We use Equation \eqref{eq:nDotGradn} to rewrite the first term on the right-hand side of Equation \eqref{eq:momentumSpaceNumberFluxIntermediatePerp}; i.e.,
\begin{equation}
  \pderiv{\lu{\nu}}{\putld{\jmath}}\,\fourvelocityEu{\rho}\,\covderiv{\rho}\fourvelocityEd{\nu}
  =\pderiv{\lu{i}}{\putld{\jmath}}\,\f{1}{\alpha}\pderiv{\alpha}{\xu{i}}.  
\end{equation}
Similarly, since both $\partial\lu{\nu}/\partial\putld{\jmath}$ and $\lu{\nu}$ are spacelike, we use Equation \eqref{eq:extrinsicCurvature} to rewrite the second term on the right-hand side of Equation \eqref{eq:momentumSpaceNumberFluxIntermediatePerp}; i.e., 
\begin{equation}
  \pderiv{\lu{\nu}}{\putld{\jmath}}\,\lu{\rho}\,\covderiv{\rho}\fourvelocityEd{\nu}
  =\pderiv{\lu{i}}{\putld{\jmath}}\,\lu{j}\,K_{ij}.  
\end{equation}
For the third term we have
\begin{eqnarray}
  \pderiv{\lu{\nu}}{\putld{\jmath}}\,\fourvelocityEu{\rho}\,\covderiv{\rho}\ld{\nu}
  &=&
  \pderiv{\lu{\nu}}{\putld{\jmath}}\,\fourvelocityEu{\rho}\,
  \Big\{\,
    \pderiv{\ld{\nu}}{\xu{\rho}}
    -\cudd{\mu}{\nu}{\rho}\,\ld{\mu}
  \,\Big\} \nonumber \\
  &=&
  \pderiv{\lu{\nu}}{\putld{\jmath}}\,
  \Big\{\,
    \fourvelocityEu{\rho}\,\pderiv{\ld{\nu}}{\xu{\rho}}
    +\ld{\rho}\,\pderiv{\fourvelocityEu{\rho}}{\xu{\nu}}
    -\lu{\rho}\,\covderiv{\nu}\fourvelocityEd{\rho}
  \,\Big\} \nonumber \\
  &=&
  \pderiv{\lu{i}}{\putld{\jmath}}\,
  \Big\{\,
    \tauderiv{\ld{i}}
    -\f{\ld{j}}{\alpha}\,\pderiv{\betau{j}}{\xu{i}}
    +\lu{j}\,K_{ij}
  \,\Big\}, 
\end{eqnarray}
where we have used Equation \eqref{eq:spacelikeDotGradn}, and defined the ``proper time derivative" along constant coordinate lines
\begin{equation}
  \tauderiv{}=\f{1}{\alpha}\pderiv{}{t}-\f{\betau{i}}{\alpha}\,\pderiv{}{\xu{i}}. 
\end{equation}
Finally, for the fourth term on the right-hand side of Equation \eqref{eq:momentumSpaceNumberFluxIntermediatePerp} we have
\begin{eqnarray}
  \pderiv{\lu{\nu}}{\putld{\jmath}}\,\lu{\rho}\,\covderiv{\rho}\ld{\nu}
  &=&
  \pderiv{\lu{\nu}}{\putld{\jmath}}\,\lu{\rho}\,
  \Big\{\,
  \pderiv{\ld{\nu}}{\xu{\rho}}-\cudd{\mu}{\nu}{\rho}\,\ld{\mu}
  \,\Big\} \\
  &=&
  \pderiv{\lu{i}}{\putld{\jmath}}\,
  \Big\{\,
    \lu{j}\,\pderiv{\ld{i}}{\xu{j}}
    -\f{1}{2}\,\lu{j}\,\lu{k}\,\pderiv{\gmdd{j}{k}}{\xu{i}}
  \,\Big\}.  
\end{eqnarray}
Combining all the terms we obtain the momentum space flux appearing in Equation \eqref{eq:ConservativeBoltzmannEquationGeneralRelativisticThreePlusOne}. 